\documentclass[11pt,oneside]{article}

\usepackage{vmargin}

\setpapersize{A4}
\setmarginsrb{2.5cm}       
{3.5cm}                        
{2.5cm}                       
{3.5cm}                         
{10pt}                           
{1cm}                           
{1pt}                             
{2cm}                           

\usepackage{enumerate}
\usepackage{dcolumn}   
\usepackage{bm}        
\usepackage{amssymb}   
\usepackage{amsbsy}
\usepackage{stackrel}
\usepackage{stmaryrd}
\usepackage{float}

\usepackage[usenames,dvipsnames]{xcolor}

\usepackage[utf8]{inputenc}
\usepackage{amsmath,amsfonts,amssymb,amscd,amsthm,xspace}
\usepackage{slashed}
\usepackage{epsfig}
\usepackage{epstopdf}
\usepackage{mathrsfs}

\usepackage{dsfont}
\usepackage{graphicx}
\usepackage{xcolor}
\usepackage{hyperref}
\hypersetup{urlcolor=black, colorlinks=false}

\usepackage{tikz-cd}

\usepackage{ctable} 
\usepackage{multirow}
\usepackage{booktabs}

\usepackage{amssymb}
\usepackage{xfrac}
\usepackage{soul}

\newcommand\beq{\begin{equation}}
\newcommand\eeq{\end{equation}}
\newcommand\beqa{\begin{eqnarray}}
\newcommand\eeqa{\end{eqnarray}}
\newcommand{\dd}{\mathrm{d}}
\newcommand{\nn}{\nonumber}

\newcommand\beqal{\begin{align}}
\newcommand\eeqal{\end{align}}

\newcommand{\la}{\langle}
\newcommand{\ra}{\rangle}
\newcommand{\xa}{\rangle\langle}
\newcommand{\id}{\mathrm{Id}}

\newcommand{\calQ}{\mathcal{Q}}

\newcommand{\ve}{\varepsilon}

\newcommand{\UU}{\mathcal{U}}
\newcommand{\HH}{\mathcal{H}}

\newcommand{\A}{\mathcal{A}}
\newcommand{\BB}{\mathcal{B}}
\newcommand{\s}{\mathcal{S}_1}

\newcommand{\Tr}{\mathrm{Tr}}
\newcommand{\rmT}{\mathrm{T}}
\newcommand{\tl}{\tilde}

\newcommand{\CC}{\mathbb{C}}

\newcommand{\CPTP}{\mathrm{CPTP}}

\newcommand{\KK}{\mathcal{K}}

\newcommand{\D}{\mathcal{D}}

\newcommand{\Ss}{\mathcal{S}}

\newcommand{\T}{\mathrm{T}}

\newcommand{\g}{\boldsymbol{g}}

\usepackage{xcolor,cancel}

\usepackage[sorting=none,url=true,firstinits=true,backend=bibtex,maxcitenames=50,maxnames=50]{biblatex}
\newcolumntype{P}[1]{>{\centering\arraybackslash}p{#1}}
\newcolumntype{?}{!{\vrule width .5pt}}

\newtheorem{theorem}{Theorem}[section]

\newtheorem*{conjecture}{Conjecture 1}

\newtheorem{claim}{Claim}[section]

\newtheorem{definition}[theorem]{Definition}
\newtheorem{proposition}[theorem]{Proposition}

\newtheorem{corollary}[theorem]{Corollary}
\newtheorem{lemma}[theorem]{Lemma}
\theoremstyle{definition}

\newtheorem{remark}[theorem]{\textbf{Remark}}
\newtheorem{example}[theorem]{\textbf{Example}}

\newtheorem{comment}[theorem]{\textbf{Comment}}

\newtheorem{question}{\textbf{Question}}

\usepackage{authblk}

\title{\vspace{-2cm}Geometry of Banach spaces: a new route towards Position Based Cryptography}

\author[1]{Marius Junge\thanks{mjunge@illinois.edu}}
\author[2]{Aleksander M. Kubicki\thanks{amkubickif@gmail.com}}
\author[3]{Carlos Palazuelos\thanks{cpalazue@ucm.es}}
\author[3]{David Pérez-García\thanks{dperezga@ucm.es}}
\affil[1]{Department of Mathematics, University of Illinois, Urbana, IL 61801, USA}
\affil[2,3,4]{Departamento de Análisis Matemático y Matemática Aplicada, Universidad Complutense de Madrid, 28040 Madrid, Spain}
\affil[3,4]{Instituto de Ciencias Matemáticas, 28049 Madrid, Spain}
\date{}                     
\begin{document}



\maketitle

\abstract{In this work we initiate the study of Position Based Quantum Cryptography (PBQC) from the perspective of geometric functional analysis and its connections with quantum games. The main question we are interested in asks for the optimal amount of entanglement that a coalition of attackers have to share in order to compromise the security of any PBQC protocol. Known upper bounds for that quantity are exponential in the size of the quantum systems manipulated in the honest implementation of the protocol. However, known lower bounds are only linear. 
	
	In order to deepen the understanding of this question, here we  propose a Position Verification (PV) protocol and find lower  bounds on the resources needed to break it. The main idea behind the proof of these bounds is the understanding of cheating strategies as vector valued assignments on the Boolean hypercube. Then, the bounds follow from the understanding of some geometric properties of  particular Banach spaces, their type constants. Under some regularity assumptions on the former assignment, these bounds lead to exponential lower bounds on the quantum resources employed, clarifying  the  question in this restricted case. Known attacks indeed satisfy the assumption we make, although we do not know how universal this feature is. Furthermore, we show that the understanding of the type properties of some more involved Banach spaces would allow to drop out the assumptions and lead to  unconditional lower bounds on the resources used to attack our protocol. Unfortunately, we were not able to estimate the relevant type constant. Despite that, we conjecture an upper bound for this quantity and show some evidence supporting it. A positive solution of the conjecture would lead to stronger security guarantees for the proposed PV protocol providing a better understanding of the question asked above.   }

\newpage

\tableofcontents

	\fontdimen16\textfont2=4pt
\fontdimen17\textfont2=4pt

 \newpage
 
 \section{Introduction}\label{Sec1}

In the field of Position Based Cryptography (PBC) one aims to develop cryptographic tasks using the geographical position of a third party as its only credential. Once the party proves to the verifier that it is in fact located at the claimed position, they interact considering the identity of the third party as granted. Basing cryptographic security on the position of the communicating parties  might be  very appealing in practical  contexts such as the use of autonomous cars (see \cite{Malaney16} for an interesting digression on this topic), or the  secure communication between public services or banks. Besides that, at a more fundamental level, secure PBC could also serve as a way to circumvent insecurity under man-in-the middle attacks, a security leak suffered by standard cryptographic primitives. This vulnerability still prevails even in presence of information-theoretical security, as, for example, in the celebrated case of  Quantum Key Distribution. In these settings, the security guarantees always come after the assumption that  the identity of the trusted agents is granted. In PBC this assumption can be, at least, relaxed. Moreover, PBC proved to be a rich field of research emanating deep questions and connections from its study. To mention a few, attacks for PBC has been related with quantum teleportation \cite{KonigBeigi}, circuit complexity \cite{Speelman2015}, classical complexity theory \cite{Buhrman_2013} and, very recently, with properties of the boundary description of some processes in the context of  the holographic duality AdS/CFT \cite{May_2019,May_2020}. In this work, we add to this list a connection with deep questions on the geometry of Banach spaces.

The main task in PBC is the one of \emph{Position Verification} (PV). In PV a prover has to convince a  verifier (usually composed by several agents spatially distributed) that it is located at a claimed position.	This setting has been studied since the 90's in the context of classical cryptography. Nonetheless, in purely classical scenarios, PV is easily proven to be insecure against a team of colluding adversaries surrounding the honest location \cite{Chandran_09}.  This motivates the study of \emph{quantum} PV protocols, in which the communication between   prover and  verifier is in general quantum. This idea was initially developed by A. Kent \cite{Kent_2011} and made rigorous only  later on in \cite{Buhrman2011}. In this last paper, the authors construct a generic attack for any quantum PV protocol.  To construct the general attack of \cite{Buhrman2011}, the authors built on the work of L. Vaidman \cite{Vaidman03}, realizing that the cheating action in the setting of PV consists in performing what they called \emph{instantaneous non-local computation}. In this last task, two (or more) distant agents  have to implement a quantum operation on a distributed input when subjected to non-signalling constraints -- see \cite{Buhrman2011} or Section \ref{Sec2.2} below for more details. At a first sight, the existence of general attacks to quantum PV renders the development of secure PBQC  a  hopeless program. However, their attack did not come for free for the adversaries, as in the case of classical PV. On the contrary, in order to cheat, the dishonest agents have to use a huge amount of entanglement -- a delicate and expensive resource in quantum information processing. Even when in \cite{KonigBeigi} another generic attack to PV was proposed exponentially reducing  the entanglement consumption, the amount of entanglement required is still far from what is realizable in any practical situation. This leads naturally to the following question, which is the one motivating this work:

\vspace{1em}
\begin{question}\label{question1} 
	How much entanglement is necessary to break \emph{any} PV protocol?\end{question}
\vspace{1em}

Answering this question with a large enough lower bound  would lead to the existence of PV protocols which are \emph{secure for all practical purposes}, term coined in \cite{Buhrman_2013}. More concretely, we say that a PV protocol is secure for  all practical purposes if the resources needed to break it are significantly larger in order of magnitude than the resources manipulated by the honest parties. For us, the size of the resources in place is quantified by the dimension of the systems that are manipulated in the execution of the protocol. In a hypothetical future in which we have at our disposal large scale quantum computers, there is no clear reason to distinguish between classical and quantum resources and solving Question \ref{question1} in this sceptical setting is the final goal in the study of PBC. However, as an intermediate step towards this aim, we focus here in the study of quantum resources disregarding classical communication and computation as free resources (for both, honest and dishonest agents). We hope that the study of this scenario will contribute to the ultimate understanding of Question \ref{question1}. Indeed, some of the results presented here can be translated to  the sceptical framework described above. Although we will say a few words about how this is achieved in Section \ref{Sec4_2},  a full study of this more ambitious setting is out of the scope of the present manuscript. 
	
	We comment now on the progress in the field that is already available. In \cite{Buhrman2011}, the authors provide the first PV protocol secure against cheaters with \emph{no} entanglement. This was improved in \cite{KonigBeigi} and later in \cite{Tomamichel2013} providing PV protocols requiring a linear amount of entanglement (linear in the size of the system manipulated in the honest protocol). In terms of this figure of merit, the entanglement consumption in the generic attack of \cite{KonigBeigi} is exponentially large, hence leaving an exponential gap between lower and upper bounds for the amount of entanglement necessary to break  PV protocols. After almost ten years since \cite{Buhrman2011} this is still  essentially all it is known about Question \ref{question1} in its original formulation. Other works have studied attacks  with some specific structure \cite{Buhrman_2013}, have designed attacks that are efficient at emulating the computation of unitaries with low complexity  \cite{Speelman2015} or have studied security under additional cryptographic assumptions \cite{Unruh2014}. 

After the completion of this manuscript we learnt about the concurrent work \cite{bluhm21} which studies a similar setting as the one considered in this work, focusing on the trade-off between the \emph{quantum} resources used by the honest party in comparison with the \emph{quantum}  resources of the attackers.  In that work, the authors show the existence of qubit routing protocols in which the honest prover is required to manipulate a single qubit and a $2 n$-bits classical string and are secure against adversaries sharing an entangled state of dimension linear in the dimension of the classical message. In the intermediate setting commented on  before, when the focus is put on the study of quantum resources, the results reported in \cite{bluhm21} are  incomparably stronger than the results we obtain here. However, in order to contrast both works, we mention that while in \cite{bluhm21} the classical part of the challenge is required to be distributed symmetrically from both sides of the prover -- considering PV in a one-dimensional line --, in our setting the classical information is distributed asymmetrically only from one of the verifiers surrounding the honest location.  This can be understood as a  further step in-between the intermediate setting in which classical resources are completely disregarded and the final goal of finding \emph{secure for all practical purposes} PV protocols. Stressing this point,  we emphasize that the techniques and ideas we introduce here might serve as groundwork for a deeper study of the problem.  In fact, as we said before, it is possible to  extend  some of our results  to protocols in which the interaction between verifiers and prover is purely quantum and, in overall, of much lower dimension(with no distinction between classical and quantum systems). We leave for the future the study of such ramifications of our work. For completeness, we also mention that another possibility to achieve the goal of \emph{security for all practical purposes} in PV would be improving the bounds obtained in \cite{bluhm21}. Known attacks to the protocols proposed there consume exponentially more resources than the lower bounds of \cite{bluhm21}, a fact that invites to explore the pointed direction. Nevertheless, it seems that new techniques have to come into play for pursuing that aim.

\subsection{Summary of results}\label{Sec4_2}

Here we aim to go back to Question \ref{question1} in its simplest form: the one-dimensional case without any further assumptions. Unfortunately, we were not able to find a definite answer to the question but we report here some progress that  opens an avenue for a deeper understanding of the problem.

 From now on, we focus on the study of \emph{quantum} resources required to attack PV, considering classical communication as a free resource and unlimited computational power for all the agents involved. In this work,
\begin{itemize}
	\item we connect the study of Question \ref{question1} with powerful techniques coming from Banach space theory,
	\item consequently providing new lower bounds on the amount of entanglement necessary to break a specific PV protocol presented in Section \ref{Sec3}. However, these bounds are not completely general but depend on some properties of the strategies considered. Intuitively, \emph{smooth} strategies, i.e., strategies with a smooth dependence in the unitary to be implemented, lead to exponential lower bounds. 
	\item Finally, we consider the possibility of turning the previous bounds unconditional. We relate the validity of this with a collection of open problems in local Banach space theory. In particular, we relate the bounds on resources to break our PV protocol with estimates for type constants of tensor norms of $\ell_2$ spaces. In this direction, we put forward a conjecture that would imply  the desired unconditional exponential lower bounds and then  provide some evidence supporting it.
\end{itemize} 

 \paragraph{The protocol  \boldmath{$G_{Rad}$}.} To formalize this discussion, we propose a PV protocol that we denote $G_{Rad}$. This makes reference to a family  $\lbrace G_{Rad}^{(n)}  \rbrace_{n\in \mathbb{N}}$ rather than to a single task. The index $n$  represents a security parameter  that  determines the size of the quantum systems manipulated in the honest implementation of the protocol. From now on, this parameter will be implicitly referred to, allowing us to drop the superindex in $G_{Rad}^{(n)}$ and refer to it simply as $G_{Rad}$.

The general structure of a PV protocol in the studied setting -- one-dimensional PV -- proceeds in four basic steps  (see Figure \ref{fig1}, left panel, for a graphical description):
\begin{enumerate}
	\item The verifier prepares a bipartite system and distributes it to two verifying agents that surround the location to be verified, $x$. For the sake of concreteness, we locate these agents at points $x \pm \delta$ for some positive $\delta$.
	\item Agents at $x \pm \delta$, when synchronized, communicate the registers they hold to $x$.
	\item An honest prover located at $x$, upon receiving both registers, immediately applies a required computation resulting  in another bipartite system. The latter has to be returned to locations $ x \pm \delta$. One register should be sent to the agent at the left of $x$ ($x-\delta$), and the other, to its right ($x + \delta$).
	\item Finally, the verifiers check whether the prover's answer arrives  on time and whether the computation was performed correctly. Based on this information they declare the verification successful or not.
\end{enumerate} 

	\begin{figure}
	\includegraphics[width=0.95\textwidth]{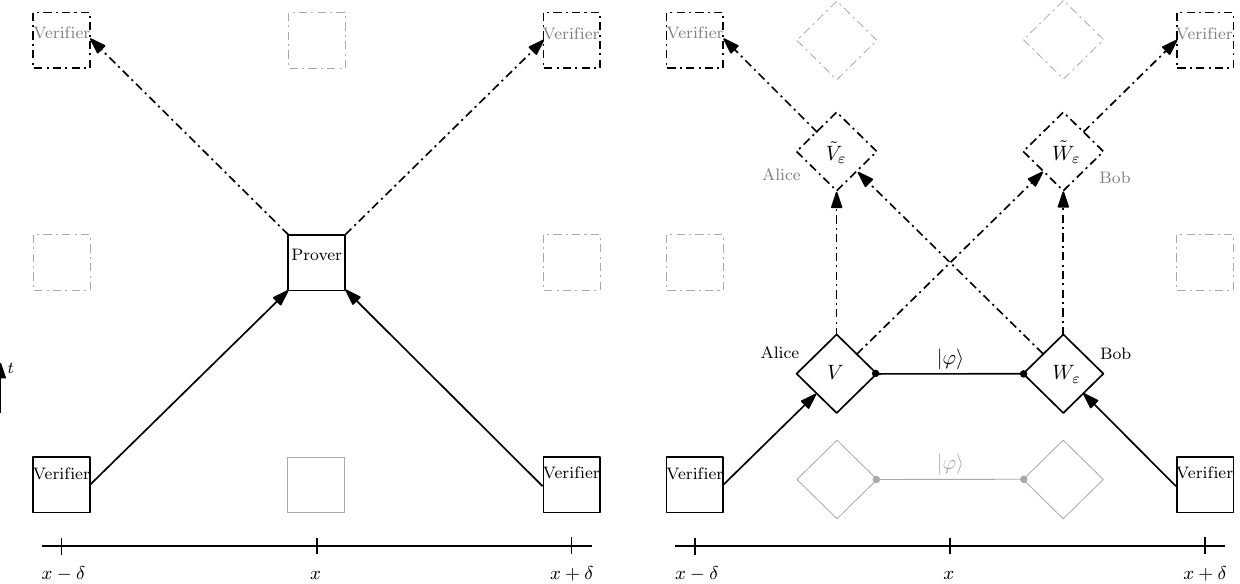}
	\caption{Causal structure of one-dimensional PV protocols. Honest implementation (left) vs.  adversarial scenario (right).}
	\label{fig1}
\end{figure}

In the dishonest scenario, two cheaters surrounding the location $x$  intercept the communication with the honest prover and try to emulate the ideal action in the honest protocol. In order to succeed, they have to prevent any delay in their response. This restricts  cheaters' action to consist of two rounds of local operations mediated by a step of \emph{simultaneous two-way communication} -- see Section \ref{Sec2.2}  for a detailed discussion of this model. 

Once we have fixed this basic setting, let us  describe the protocol $G_{Rad}$ involved in our main results.   Roughly  speaking, the challenge posed to the prover in our protocol is solved by the implementation of the set of diagonal unitaries determined by sign vectors $\ve \in \{\pm 1\}^{n^2}$. The intuition behind the choice of this set of unitaries can be supported by the fact that it contains instances with exponential circuit complexity, as a simple counting argument shows. Furthermore, in \cite{Kubicki_19} we noticed that this set of unitaries is almost as hard as possible  in terms of the memory required by a Programmable Quantum Processor that implements it.  Since Programmable Quantum Processors seem to be closely related with the existing teleportation based attacks to PV \cite{Buhrman2011,KonigBeigi}, we found the previously noted fact an indication that the referred set of unitaries might be a good choice for the study  of PV. More formally, the honest implementation of $G_{Rad}$ is as follows:
\begin{enumerate}
	\item  the verifiers start uniformly sampling  $\ve = (\ve_{ij})_{i,j=1}^n \in \{ \pm 1\}^{n^2}$ and preparing the state $|\psi\ra : = \frac{1}{n}  \sum_{i,j=1}^n |i\ra_A \otimes |j\ra_B \otimes |ij\ra_C $ in a tripartite Hilbert space $\HH_A \otimes \HH_B \otimes \HH_C$. The verifying agent at $x - \delta$ receives registers $\HH_A \otimes \HH_B$ while the one at $x + \delta$ is informed (classically) of the choice of $\ve$. Register $\HH_C$ is kept as private during the execution of the protocol. 
	\item Then, registers $\HH_A \otimes \HH_B$ are forwarded to the verifying location $x$ from  its left.  From the right,  the classical information about the choice of $\ve$ is communicated.
	\item An honest prover located at $x$, upon receiving both pieces of information, has to apply the diagonal unitary on $\HH_A \otimes \HH_B$ determined by $\ve$. Immediately,  registers $\HH_A \otimes \HH_B$ must be returned, but this time only $\HH_A$ should travel to the verifier at $x - \delta$. Register $\HH_B$ should be sent to the verifier at $x + \delta$.
	\item After receiving those registers, the verifiers check the answer's timing  and, at some later time, they perform the measurement  $\lbrace |\psi_\ve \xa \psi_\ve |, \id - |\psi_\ve \xa \psi_\ve | \rbrace$ on  system $\HH_A \otimes \HH_B \otimes \HH_C$, where $|\psi_\ve\ra : = \frac{1}{n}  \sum_{i,j} \ve_{ij} |i\ra_A \otimes |j\ra_B \otimes |ij\ra_C $. They accept the verification only if the arriving time was correct and the outcome of the measurement was the one associated to $|\psi_\ve \xa \psi_\ve |$. 
\end{enumerate}

Next, let us specify the implementation of $G_{Rad} $ in an adversarial scenario. In this situation, we consider that two cheaters located between the honest location $x$ and the verifying agents at $x\pm \delta$, intercept the communication in the honest protocol. In this work, we refer to these cheaters as Alice, at position $x -  \delta'$, and Bob, at position $x+ \delta'$, for some $0<\delta'<\delta$. Their general action proceeds as follows\footnote{ For simplicity, we state here the case in which Alice and Bob use what we call \emph{pure} strategies. The most general case can be reduced to this one by purification. See Section \ref{Sec3} for a detailed discussion.} (see again Figure \ref{fig1} for clarification): in advance, the cheaters share a state $|\varphi\ra$ in which Bob, after receiving the information about $\ve$, applies an isometry $W_\ve$ and sends part of the resulting system to Alice together with the classical information determining $\ve$. On her part, when Alice receives registers $\HH_A\otimes \HH_B$ of $|\psi\ra$, she applies another isometry $V$ (independent of $\ve$) on these registers and her part of the shared state $|\varphi\ra$. Part of her resulting system is communicated to Bob. After this step of simultaneous two-way communication Alice and Bob are allowed to apply another pair of local isometries $\tilde V_\ve \otimes \tilde W_\ve$ on the systems they hold. Then, they have to forward an answer to  agents at $x \pm \delta$. 

\paragraph{Main results.} The structure of $G_{Rad}$ allows  us to understand cheating strategies as vector valued assignments on the $n^2$-dimensional boolean hypercube, $ \mathcal{Q}_{n^2} = \lbrace \pm 1 \rbrace^{n^2}$. In our main result, we find lower bounds for the resources consumed in such an attack depending on the \emph{regularity} of the former assignment. Very informally, we can state:\vspace{0.2cm}

{\it Cheating strategies depending on the value of $\ve\in  \lbrace \pm 1 \rbrace^{n^2}$ in a sufficiently regular way require an amount of entanglement  exponential in $n$ in order to pass $G_{Rad} $ .}\\\vspace{-0.2cm}

To quantify the regularity of a strategy we introduce a parameter $\sigma$ that can be regarded as a measure of the  \emph{total influence} of the associated function on the Boolean hypercube. We give a precise definition  for this parameter in Section \ref{Sec4}. Here, we restrict ourselves to give an intuitive idea behind this definition presenting some approximate expressions below.  Based on two complementary ideas, given a strategy we construct  two different assignments leading to two parameters  $\sigma^i$ and $\sigma^{ii}$.   Given a cheating strategy $\Ss$,  characterized  by a sequence of elements $ \lbrace \tilde V_\ve, \tilde W_\ve, V,W_\ve, |\varphi\ra  \rbrace_{  \ve \in \mathcal{Q}_{n^2}  }  $, we can bound, up to logarithmic factors:
\beq   \label{eq_sigma1_intro} \sigma^i
\lesssim_{\log}   \ \mathbb{E}_\ve  \   
\left(  \sum_{i,j} \frac{1}{2} \big\|  \tilde  V_\ve \otimes \tilde W_\ve - \tilde V_{\overline{\ve}^{ij}} \otimes \tilde W_{\overline{\ve}^{ij}} \big\|   ^2   \right)^{1/2}  + O\left(\frac{1}{n} \right), 
\eeq
\beq \label{eq_sigma2_intro}
	\sigma^{ii}
\lesssim_{\log}  \ \mathbb{E}_\ve \  \left(   \sum_{i,j} \frac{1}{2}  \big \| \left( V \otimes (W_\ve -  W_{\overline{\ve}^{ij}})\right) |\varphi\ra \big\|_{\ell_2}^2 \right)^{1/2} + O\left(\frac{1}{n} \right) ,\eeq
\noindent where $\| \, \cdot \, \|$ and $\|\, \cdot \, \|_{\ell_2}$ are  the operator and euclidean norms respectively. Here,  $\overline{\ve}^{ij}$ denotes the sign vector $(\ve_{11},\ldots,$$- \ve_{ij},$$\ldots, $ $\ve_{nn})$. The first of these parameters is therefore related with how strongly the \emph{second round of local operations} in the strategy depends on $\ve$. In the other hand, $\sigma^{ii}$ is similarly concerned with the dependence on $\ve$ of the \emph{first round of local operations}. With this at hand, we can state -- yet informally -- our main result. Denoting  the success probability attained by a strategy $\Ss$ in $G_{Rad} $ as $\omega(G_{Rad} ;\Ss)$, we can say that:

\begin{theorem}[Informal] \label{mainThm} 
	Given  a cheating strategy for $G_{Rad} $, $\Ss$, in which the local dimension of the quantum systems manipulated by the cheaters during its execution is at most $k$,
	
	\begin{enumerate}[I.]
		\item 
		\beq\nn 
		\omega(G_{Rad} ;\Ss)  \le  C_1 +   C_2 \ {\sigma^i} \, \log^{1/2}(k) +  O \left(\frac{1}{n^{1/2}}\right)  ;
		\eeq 
		
		\item   
		
		\vspace{-1em}\begin{align*} 
		&\omega(G_{Rad} ; \Ss)	  \\ &\quad\le  \tilde C_1 + C_3 \    \sigma^{ii} \, n^{3/4}  \log^{3/2}(nk)    +  O \left(\frac{1}{n^{1/2}}+\frac{ \log^{3/2}( n k)}{n}\right) 
		;
		\end{align*}
	
	\end{enumerate}

	where $C_1,\, \tilde C_1 <1, \, C_2,\, C_3 $ are positive constants.
	
\end{theorem}

What this theorem tells us is that cheating strategies for $G_{Rad}$ for which $\sigma^i $ or $ \sigma^{ii} $ are small enough necessarily need to make use of quantum resources of size exponential in a power of $n$, (loosely) matching the exponential entanglement consumption of known attacks\footnote{ The attack from \cite{KonigBeigi} requires an entangled system of dimension $O(\exp(n^4))$, that  is still much larger than our bounds for smooth strategies. Nonetheless,  we consider that any separation on resources that is exponential in a power of $n$ is enough to discriminate between the relative power among different agents.  This is our main motivation in this work.  }. We give a more concrete statement in the form of a corollary:
\begin{corollary}[Informal]\label{mainCor}
	Consider  a cheating strategy for $G_{Rad} $, $\Ss$, attaining value $\omega(G_{Rad};\Ss) $ $\ge 1 -\epsilon$ for some $0\le \epsilon \le \frac{1}{8}$. Denote by $k$ the local dimension of the quantum resources used in $\Ss$.
	
	If  $ \sigma^i  = O( \mathrm{polylog}(n) / n^{\alpha}) $ or $ \sigma^{ii}  = O(  \mathrm{polylog}(n) / n^{3/4 + \alpha})   $ for some $\alpha >0$, then:
	\beq\nn 
	k  = \Omega \big( \exp\big( n^{\alpha'} \big) \big) \quad \text{for some }\alpha'>0.
	\eeq
\end{corollary}

As we see, the regularity parameters $\sigma^{i(ii)} $ play a key role in these results. We notice that known attacks in \cite{Buhrman2011,KonigBeigi} in fact fulfil the hypothesis of the previous corollary: the second round of local operations in these attacks is $\ve$-independent, hence\footnote{ Notice that in this case the first summand in the RHS of \eqref{eq_sigma1_intro}  vanishes. This leads to the estimate $\sigma^i  \lesssim_{\log} 1/n$. A look into the proof of the upper bound \eqref{eq_sigma1_intro}, Proposition \ref{Prop.A1}, i., reveals that the logarithmic term hidden in $\lesssim_{\log}$ is indeed proportional to $\log(n)$. }  $\sigma^i  \sim \log(n)/n$ . However, we do not know how generic this behaviour is. More generally, it turns out that from any Programmable Quantum Processor \cite{NielsenChuang_97} -- as  the already considered protocol of Port Based Teleportation, for example -- with the capability of implementing the diagonal unitaries required in $G_{Rad}$, we can construct an assignment $\Phi$ fulfilling Theorem \ref{mainThm} with regularity parameter again of order $\sigma^{i} \sim \log(n)/n$. Therefore, Corollary \ref{mainCor} also applies to this broader case allowing to recover some of the results obtained  in \cite{Kubicki_19}. This is not a coincidence, our approach here builds on ideas introduced in this previous work.

Turning our attention towards $\sigma^{ii}$, a trivial example of a family of \emph{smooth} attacks for which $\sigma^{ii} \sim \log(n)/n$ is given by cheaters sharing no entanglement in advance -- even when entanglement can be created in the first round of local operations and  distributed for the second round. By contrast, we   can  also easily compute $\sigma^{ii}$ for the attack in \cite{KonigBeigi} obtaining $\sigma^{ii} = O(1)$. Therefore, our second item in Theorem \ref{mainThm} is not able to predict good lower  bounds for this case. Still, we think that this second item  might be useful for restricting the structure of possible attacks to PV, especially in conjunction with the first part of the theorem.

More importantly, the second part of Theorem \ref{mainThm} leads us to put forward the possibility of an unconditional lower bound for $k$, i.e., a bound in the spirit of Corollary \ref{mainCor} but dropping out the assumptions regarding $\sigma^{i(ii)} $. Even when we were not able to prove such a bound, we  relate its validity  with a conjecture about the geometry of some Banach spaces. More concretely, our conjecture has to do with estimates of type constants of tensor norms on finite dimensional Hilbert spaces. Even when these  properties for the case of a single Hilbert space are very well understood -- in fact, in this case the study of type and cotype reduces to an elementary generalization of the parallelogram law --, the situation changes dramatically when tensor products of several such spaces are considered. For the latter, long-standing  questions remain open as, for example, whether the simple space $\ell_2 \otimes_{\pi} \ell_2 \otimes_{\pi} \ell_2$ has finite cotype (see Sections \ref{Sec.2.4.1} and \ref{Sec2.4.3} for the definition of the objects mentioned here). This is a famous question asked by Pisier decades ago -- see, for instance, \cite{Pisier90} -- and about which still very little is known.

 Once we formally state the conjecture in Section \ref{Sec6}, we provide some computation supporting it. We analyze the most direct approaches to disprove the conjectured statement  providing an estimate of the volume ratio of some relevant spaces. This might have interesting ramifications on the still not completely understood relation between volume ratio and cotype of  Banach spaces.

\paragraph{Further extensions of this work.}  To conclude this introductory summary, we highlight that there  is a natural way to remove the classical part of the input in $G_{Rad}$, obtaining protocols in which the overall dimension (classical and quantum) of the systems the honest agents are required to manipulate is polynomial in $n$. Taking inspiration from the definition of $G_{Rad}$, we now fix $\ve \in \calQ_{n^2}$ as publicly known and define a PV protocol $G_\ve$ that proceeds as follows:
\begin{enumerate}
	\item the verifier starts uniformly preparing the state $|\psi\ra : = \frac{1}{n}  \sum_{i,j} |i\ra_A \otimes |j\ra_B \otimes |ij\ra_C $ in a tripartite Hilbert space $\HH_A \otimes \HH_B \otimes \HH_C$. The agent at $x - \delta$ receives register $\HH_A $ while the one at $x + \delta$ receives $\HH_B$.  Register $\HH_C$ is kept as private during the execution of the protocol. 
	\item Then, registers $\HH_A \otimes \HH_B$ are forwarded to the verifying location $x$, $\HH_A$ from  its left and $\HH_B$ from its right. 
	\item An honest prover located at $x$, upon receiving both pieces of information, has to apply the diagonal unitary on $\HH_A \otimes \HH_B$ determined by $\ve$. Immediately,  registers $\HH_A \otimes \HH_B$ must be returned. $\HH_A$ should travel to the verifier at $x - \delta$ and $\HH_B$,  to the verifier at $x + \delta$.
	\item The verification is now carried out in the same way as in $G_{Rad}$. 
\end{enumerate}

Considering  the \emph{family} of protocols $\{G_\ve\}_{\ve \in \calQ_n^2}$ as a whole, it is possible to recover a notion of smooth strategies with some associated regularity parameter $\sigma$. Such notion of regularity allows us to obtain an equivalent result to Theorem \ref{mainThm} -- and, therefore, to Corollary \ref{mainCor} -- for this case.  A criticism that might be made at this point is that it is  less clear than before why one should expect any regularity among strategies that applies to different games. A possible line of argumentation against this criticism could be stated in terms of protocols for instantaneous non-local quantum computation: if one aims to construct protocols that are universal, in the sense they are able to non-locally implement any unitary, it seems rather difficult to come with something that depends on the unitary to be implemented in a very non-regular way.  The authors of \cite{KonigBeigi} seem to go along with that idea when stating the notion of \emph{``protocols which only make black-box use of the unitary''}.

Leaving aside the concerns triggered by the appearance of regularity assumptions, one could also pursue unconditional bounds for $\{G_\ve\}_{\ve \in \calQ_{n^2}}$ following a similar route as the one drawn in Section \ref{Sec6}. This time the Banach spaces that appear are even more convoluted and, at the moment of writing this manuscript, we do not have any serious evidence to guess the behaviour of their type properties. The study of the issues arising from the previous considerations is postponed for future research.

Finally, as a general comment, we note that the study of PV protocols can be phrased in terms of \emph{quantum games}, a framework that might provide the right level of abstraction for further generalizations of the present work. The interested reader can find a detailed account of such rephrasing in \cite[Chapter 4]{Kubicki_thesis}.

\subsection{Proof sketch}

Here we sketch the main ideas behind the proof of Theorem \ref{mainThm}. These ideas are also at the bottom of the constructions that allow us to establish the more general connection between Question \ref{question1} and type constants that leads to the conjecture indicated above.

As we have already mentioned, the starting point of our study is the identification of each cheating strategy, $\Ss$, with a vector-valued function $\Phi: \calQ_{n^2} \rightarrow X$, being $\calQ_{n^2} =\{\pm1\}^{n^2}$ and $X$, a well-suited Banach space. With an appropriate definition of $\Phi$ -- which also includes the choice of $X$ -- we can obtain a bound on the success probability  $\omega(G_{Rad}, \Ss) $ in terms of the average norm of the image of that function. We obtain bounds of the following kind:
      
      	\vspace{-0.8em}\beq\nn
      \omega ( G_{Rad}; \Ss) \le \, \mathbb{E}_\ve \, \big \| \Phi (\ve) \big \|_{X},
      \eeq
      
 \noindent where $\ve$ is taken uniformly distributed in $\calQ_{n^2}$. Therefore, the key quantity we study is precisely $\mathbb{E}_\ve \, \| \Phi (\ve)  \|_X$. For that, we bring together two main ingredients. On one hand, a Sobolev-type inequality of Pisier for vector-valued function on the Boolean cube and, on the other, the type-2 constant of the Banach space $X$, $\rmT_2(X)$. The combination of these two tools provides us with an inequality:
    	\beq\label{PisierIneq_Intro}
   \mathbb{E}_\ve\big \| \Phi (\ve) \big\|_X  \le \big \|\mathbb{E}_\ve \Phi (\ve) \big \|_X + C \ \sigma_\Phi \ \mathrm{T}_2 ( X ) ,
   \eeq
   where $C$ is an independent constant and $\sigma_\Phi$ is a regularity measure for $\Phi$\footnote{ See Section \ref{Sec2.4.4} for a detailed discussion. There, the more refined type-2 constant with $m$ vectors is considered, see Corollary \ref{Cor1}. For the sake of simplicity, we consider the plain type-2 constant in this introductory section.}. Specific choices for $\Phi$ and $X$ leads to parameters $\sigma^{i(ii)}$ appearing in Theorem \ref{mainThm}. 
   
   Now, depending on how $\Phi$ is constructed, $\|\mathbb{E}_\ve \Phi (\ve) \big \|_X $ can be upper bounded by a quantity strictly smaller than $1$-- see, for instance, Proposition \ref{Prop_Main2} . Once such a bound is obtained, the focus can be put on the second term in the RHS of \eqref{PisierIneq_Intro}.
   
    To obtain Theorem \ref{mainThm} we propose in Section \ref{Sec4} two possible choices for $\Phi$ and study the type constants of their image spaces.  Furthermore, in order to remove the dependence on $\sigma$ in the  bounds obtained in that way, we propose in Section \ref{Sec6} yet another choice for $\Phi$. This third function is regular enough by construction allowing to obtain bounds depending only on the dimension of the system used by the cheaters. The downside of this latter approach is that the space $X$ in this last case becomes more involved and its type properties cannot be estimated with the techniques at our disposal.

To finish this introduction we sum up the structure of the paper: we start introducing in Section \ref{Sec2} preliminary material needed to develop this work. Then, in Section \ref{Sec3} we study general aspects of cheating strategies for $G_{Rad}$ paving the ground for our main results. The analysis of strategies  leading to Theorem \ref{mainThm} is presented in Section \ref{Sec4}. 
 In Section \ref{Sec6} we discuss the possibility of pushing forward the techniques presented in this work to obtain unconditional lower bounds on the resources required by the cheaters, only dependent on the dimension of the quantum system they manipulate. We connect this question with the problem in local Banach space theory of obtaining precise estimates for the \emph{type constants} of particular Banach spaces. After establishing that connection in a precise and rigorous way, we provide some calculations supporting a positive resolution of a conjecture that would lead us to strengthening  the security of $G_{Rad}$. The paper ends with a discussion of the results presented and possible directions for future work. This corresponds to Section \ref{Sec7}.

 \section{Preliminaries}\label{Sec2}
 
 	\subsection{Notation}
 
 In order to simplify the presentation, we use symbols	 $\approx$ and $\lesssim$ to denote equality  and inequality up to multiplicative dimension independent constants and $ \approx_{\log}$ and $\lesssim_{\log}$,  equality and inequality up to multiplicative logarithmic factors on the dimensions involved.
 
 The quantum mechanical description of a system is based on an underlying complex Hilbert space, that we denote $\HH,\, \HH',\,\HH_A,\, \HH_B,\, \mathcal{K},\, \ldots$. When the dimension is known to be a specific natural number, say $n$, we use the notation $ \ell_2^n$. Given that, a density matrix is a trace one, positive operator $\rho:\HH \rightarrow \HH $. We denote the set of density matrices as $\D(\HH)$. Quantum operations are completely positive trace preserving linear maps $\D(\HH) \rightarrow\D(\HH')$. The set of these maps is denoted here as $\CPTP(\HH, \HH')$ or simply $\CPTP(\HH)$ when the input and output spaces are the same. The operation of discarding a subsystem is implemented by the partial trace. We specify the subsystem discarded by its underlying Hilbert space, e.g., in a composed system with underlying Hilbert space $\HH \otimes \HH'$ the operation of discarding $\HH'$ is denoted $\Tr_{\HH'} \in \CPTP(\HH\otimes \HH', \HH)$. To describe the evolution of a quantum system after a measurement, we make use of \emph{instruments}, that are collections of completely positive trace non-increasing maps summing up to a trace preserving map. To denote a completely positive (maybe non trace-preserving) map we use the symbol $\mathrm{CP}$ instead of the previous $\CPTP$. The set of instruments composed by finite collections of maps in $\mathrm{CP}(\HH,\HH')$ are denoted $\mathrm{Ins}(\HH,\HH')$.
 
 To denote Banach spaces we usually use letters $X,\, Y,\,\ldots$ and $X^*,\, Y^*,\,\ldots$ for the corresponding Banach duals. $B_X$ denotes the unit ball of a Banach space $X$.   $\BB(X,Y)$ is the space of bounded linear operators between arbitrary Banach spaces $X$ and $Y$, while $\ell_p(X)$ and $L_p(X)$, with $p\in (0, \infty]$, are the classical (vector valued)  spaces of $p-$summable sequences and $p-$integrable functions on the unit interval.  More specifically, we also fix now the notation for two Banach spaces that will appear repeatedly. Given two Hilbert spaces $\HH$, $\HH'$, we denote as $\BB(\HH,\HH')$ and $\Ss_1(\HH, \HH')$ the space of bounded and trace class operators from $\HH$ into $\HH' $, respectively. 
  In the finite dimensional case, $\HH = \ell_2^m$, $\HH' = \ell_2^n$, we simplify this notation to $M_{n,m}$ and $\Ss_1^{n,m}$ ($M_n$, $\s^n$ when $n=m$). To denote elements of the computational basis we use the quantum information oriented  convention of using the symbols $|i\ra,\, \la i|,\, |j\ra,\, \ldots$.  When working with elements in the complex vector space composed by $n\times m$ matrices -- as is the case of elements in $M_{n,m}$ or $\Ss_1^{n,m}$, case we consider repeatedly below -- the usual basis of matrices with only one non-zero entry is denoted here as $\lbrace | i \xa j | \rbrace_{\begin{subarray}{l}
 	i= 1 ,\ldots, n\\
 	j= 1, \ldots, m
 	\end{subarray}}$. Observing the range of each subindex, the convention chosen here matches the standard agreement on regarding \emph{kets} $|i\ra$ as column vectors and \emph{bras} $\la i | $ as rows. 

 	\subsection{Position Based Cryptography in 1-D}\label{Sec2.2}
 	
 	The major aim of this work is to make progress towards Question \ref{question1}. For that, we restrict ourselves to the simplest scenario: position verification in 1-D. In this situation, we restrict the world to a line in which we consider a preferred location $x$ -- the position to be verified. The verifier, composed by two agents, $V_{L}$ and $V_{R}$, is located around this honest position. Let us consider $V_L $ in position $x- \delta$ and $V_R $ in position $x +  \delta $. Then, $V_L$ and $V_R$ perform an interactive protocol sending (possibly quantum) messages in the direction of $x$.   These messages arrive to $x$ at the same time, so that a honest prover located at $x$ could receive them and, accordingly, generate answers   for $V_L$ and   $V_R$. The verifier accept the verification if and only if: \begin{itemize}
 		\item (correctness) the answers are correct with respect to  verifier's messages (according to some public rule);
 		\item (timeliness) the answers arrive on time to the locations of $V_L$ and $V_R$. Assuming that the signals between verifiers and prover travel at some known velocity $c$, the answers should arrive to $V_L$ and $V_R$ at time $2 \delta / c$  after the start of the protocol.
 	\end{itemize}
 

Before continuing, let us set a generic structure for such a protocol. To prepare the messages $V_L$ and $V_R$ must forward, the verifier  prepares a (publicly known) state in a composite system with some underlying Hilbert space  $\HH_L \otimes \HH_{R} \otimes \HH_C$. That is, he prepares a density matrix $ \rho_0$ on that state space and sends the register $\HH_L$ to $V_L$ and $\HH_R$ to $V_R$. $\HH_C$ is considered to take into account the possibility that the verifier keeps some part of the initial system as private during the protocol. Then,  $V_L$ and $V_R$ send their systems in the direction of $x$.   Now, the agent(s) interacting in the middle with $V_L$ and $V_R$ apply some quantum operation on the communicated system $\HH_L \otimes \HH_R$ obtaining as output another state $\rho_{ans}\in \D(\HH_L'  \otimes \HH_R')$. The subsystems $\HH_L' $, $\HH_R'$ are forwarded to $V_L$, $V_R$, respectively. To decide whether the verification is correct or not, the verifiers first check the \emph{timeliness} condition is fulfilled and then perform a (publicly known) dichotomic measurement on the system $\HH_L' \otimes \HH_R' \otimes \HH_C$.

\begin{remark}
Above, $\rho_0$ and $\rho_{ans}$ are in general quantum states but they could perfectly describe also classical messages as well as quantum-classical messages. This will be indeed the case in the concrete scheme analized in this work.
\end{remark}
\begin{remark}
  Note that a honest prover, that is, an agent at position $x$, shouldn't have any problem to pass the test: at time $\delta /c$ he would receive the whole system $\HH_L \otimes \HH_R$ from the verifiers, having the capability to perform any global operation on it to prepare his answer. This answer can  still arrive on time to $V_L$ and $V_R$. The  action described in the previous lines  is the most general operation that can be performed on verifier's messages, which are the only information transmitted  in the protocol. Therefore, if the challenge is well designed (it can be passed), the honest prover must be able to succeed at it\footnote{ We don't take into account here the computational limitations at which the agents might be subjected.}.
 	\end{remark}
 	
 	Next, let us focus on how the general protocol described above can be cheated. In order to impersonate the identity of a honest prover at position $x$, a couple of adversaries, Alice and Bob, at positions $x \pm \delta'$, $0 < \delta'< \delta $, can intercept the message systems $\HH_L$, $\HH_R$, interact between themselves to generate answers for the verifier and forward those answers in correct timing. In order to respect the timeliness of the protocol, the most general action of the cheaters proceeds as follows:
 	 	\begin{figure}[H]
 		\centering
 		{\setlength{\fboxsep}{10pt}
 			\framebox{%
 				\parbox{0.9\textwidth}{
 					\begin{enumerate}
 						\item Before the start of the protocol, Alice and Bob prepare some shared entangled state in a private register $ \HH_{E_a} \otimes \HH_{E_b}$;
 						\item Alice receives question register $\HH_L$ and applies a quantum channel $\mathcal{A}\in \CPTP(\HH_L \otimes \HH_{E_a},\,  \HH_{A \shortrightarrow B} \otimes \HH_{A \shortrightarrow A }) $. Similarly, Bob receives $\HH_R$ and applies $\mathcal{B} \in \CPTP(\HH_R \otimes \HH_{E_b},\, \HH_{B \shortrightarrow A} \otimes \HH_{B \shortrightarrow B} )$;
 						\item the cheaters interchange registers $\HH_{A \shortrightarrow B}$ and $\HH_{B \shortrightarrow A} $, keeping $\HH_{A \shortrightarrow A } $, $\HH_{B \shortrightarrow B } $\protect\footnotemark;
 						\item after this last step, Alice holds system $\HH_{A \shortrightarrow A } \otimes \HH_{B \shortrightarrow A } $, on which she applies another channel $\tl{\mathcal{A}} \in \CPTP (\HH_{A \shortrightarrow A } \otimes \HH_{B \shortrightarrow A }, \,\HH_L'  ) $. Similarly, Bob applies $\tl{\mathcal{B}} \in \CPTP( \HH_{B \shortrightarrow B } \otimes \HH_{A \shortrightarrow B } , \,  \HH_R' ) $;
 						\item finally, Alice sends $\HH_L'$ to $V_L$ and Bob $\HH_R'$ to $V_R$.
 					\end{enumerate}
 				}%
 		}}
 		\caption{ Structure of adversarial action attacking 1-D PV protocols.}\label{s2wStrategies}
 	\end{figure}
 	\footnotetext{ In general, we model in that way any kind of communication between Alice and Bob, classical or quantum.   However, in the particular setting studied later on in Section \ref{Sec3}, we will see that the dimension of $\HH_{A \shortrightarrow B}$ and $\HH_{B \shortrightarrow B} $  is essentially determined by the quantum resources the cheaters share, allowing us to disregard the classical communication that they might additionally use. See Section \ref{Sec3}, Lemma \ref{Lemma_ClassCom}, for a precise statement.}

 We call in this work \emph{simultaneous two-way communication scenario}, $s2w$,  the set of actions -- strategies from now on -- with this structure. This scenario is central for us and will appear repeatedly in the rest of this manuscript.

 	\subsection{Banach spaces, operator ideals and type constants}

 	At a technical level, the results of this work follow from the study of  Banach spaces formed by tensor products of Hilbert spaces. The spaces $M_{n,m}$ and its dual, $\Ss_1^{n,m}$, play a prominent role in the rest of the manuscript.  Properties of these spaces in conjunction with a classical Sobolev-type inequality of Pisier allow us to obtain our main result, Theorem  \ref{mainThm}.

 	 The key property we study of these spaces are type constants, that we introduce in Section \ref{Sec2.4.3}. Before that, we need to introduce some objects we work with in the following sections.
 	
 	\subsubsection{Operator ideals} \label{Sec.2.4.1}  \vspace{0.5cm}
 	
 	A deeper understanding of the constructions appearing in this work is provided by  the perspective of the theory of operator ideals. For the reader's convenience, we first sum up the contents of this section: given two finite-dimensional\footnote{ Even though in most cases the following material also applies in the infinite-dimensional case, for simplicity, we restrict to finite dimension that is all we will use here. This allows us to use the equivalence between operators and tensor products in a comfortable way ignoring the subtleties that appear at this point for infinite dimension.} Banach spaces $X$ and $Y$ we consider the space of bounded linear operators from $X$ into $Y$, $\BB(X,Y)$. An operator ideal is essentially an assignment of any pair of  Banach spaces $X$ and $Y$ with a subset of $\BB(X,Y)$ that has the \emph{ideal} property of being closed  under composition with  bounded linear maps. 
 	We provide  \cite{Pietsch86,DefantFloret} as standard references on this matter for the interested reader. In this section:
 	\begin{enumerate}
 		\item The first  examples of operator ideals we introduce are \emph{tensor norms} on pairs of Banach spaces. This includes the space of bounded operators, $\BB(X,Y)$, or $X^* \otimes_\ve Y$ in tensor norm notation, 2-summing operators $\pi_2(X^*,Y)$, or $X\otimes_{\pi_2} Y$,  and the ideal of nuclear operators, denoted as $\mathcal{N}(X,Y)$, or $X^* \otimes_\pi Y$\footnote{ Recall that here we restrict $X$ and $Y$ to be finite dimensional.}.
 		 
 		\item When $X$ and $Y$ are Hilbert spaces, another prominent family of operator ideals are the well-known Schatten classes $\Ss_p$, for $p\in [1,\infty] $. It turns out that these classes can be generalized to operators between arbitrary Banach spaces, leading to the definition of weak Schatten von-Neumman operators of type $p\in [1,\infty] $, denoted here as $\mathfrak{S}_p^w(X,Y)$ or $X^* \otimes_{\mathfrak{S}_p^w} Y$.
 		\item Finally, here we also define a variant of the space $\mathfrak{S}_p^w(X,Y)$ that appears naturally in our study and that seems to be new in the literature. We denote this space $\mathfrak{S}_p^{w-cb}(X,Y)$ or $X^* \otimes_{\mathfrak{S}_p^{w-cb}} Y$ and call it the space of \emph{weak-cb} Schatten von-Neumman operators of type $p\in [1,\infty] $. The appellative cb is reminiscent of the fact that this new structure makes use of constructions coming from operator space theory. Indeed, $ \mathfrak{S}_p^{w-cb} $ is an operator ideal but in the operator space sense, therefore belonging more naturally to   that category than to the one of Banach spaces. In any case, we state this as a matter of curiosity and completeness, and these fine-grained details are irrelevant for the scope of the present work. Nonetheless, it is possible that a further exploration of these structures  could lead to the clarification of some of the problems we leave open. 
 	
 	\end{enumerate}
 	 After this brief summary, we provide now the details of the contents cited above. We follow part of the exposition \cite[Chapter 2]{Pietsch86} with suitable simplifications adapted to the scope of this work.

 	For finite dimensional Banach spaces $X$ and $Y$,  the space of linear maps $X \rightarrow Y$ can be  identified in a simple way with the tensor product $X^* \otimes Y$, as was implicitly assumed above. The identification consists in associating to any element in $X^* \otimes Y$, $\, \hat f = \sum_{i} x_i^*  \otimes y_i$, the linear map $f: X \ni x \mapsto f(x) := \sum_{i} x_i^*(x) \, y_i \, \in Y$. Conversely, to any linear map $f: X \rightarrow Y$ we associate  the tensor $\hat f =  \sum_i x^*_i \otimes f(x_i)$, where $\lbrace x_i \rbrace_i,\, \lbrace x_i^* \rbrace_i $ are dual bases of $X$ and $X^*$, respectively. Based on that, we will tend to present our results making explicit the tensor product structure but sometimes, especially in this introductory part of the paper, it  will be more natural to talk about mappings, so we will use both conventions interchangeably.
 	
 	 The first  operator ideal we encounter is the one of bounded operators from $X$ into $Y$, that we denote $\BB (X,Y)$ and that is the Banach space of linear operators $f:X\rightarrow Y$ endowed with the operator norm, $\|f \| := \sup_{x\in B_X} \| f(x) \|_Y <  \infty$. Using the equivalence stated before, understanding this space as a tensor product is precisely how the injective tensor product is defined: $X^* \otimes_\ve Y \simeq \BB (X,Y)$.  If $X$ and $Y$ are finite dimensional spaces, the dual of $X^* \otimes_\ve Y $ coincides with the  projective tensor product, $X \otimes_\pi Y^* \simeq (\BB (X,Y))^*$. It is enough for the scope of this manuscript to take this equivalence as the definition of $X \otimes_\pi Y^*$. These norms satisfy the desirable \emph{metric mapping property}: for any Banach spaces $X_0$, $X_1$, $Y_0$, $Y_1$, and any operators  $ f \in \BB(X_0,X_1)$, $ g \in \BB(Y_0,\, Y_1)$,
 	 \beq\label{MetricMapProp}
\big \|  f \otimes g: X_0 \otimes Y_0 \rightarrow X_1 \otimes Y_1 \big \| \le \big \| f :X_0 \rightarrow X_1 \big \| \, \big\| g : Y_0 \rightarrow Y_1 \big\|.
 	 \eeq
 	Furthermore, we call \emph{tensor norm} to  any  $\alpha$ that associates to any pair of Banach spaces $X$, $Y$, a norm $\| \, \cdot \, \|_{X\otimes_\alpha Y}$ such that:
 	\begin{itemize}
 		\item $\alpha$ is \emph{in between} of the tensor norms $\ve$ and $\pi$. That is, $$ \text{for any }x \in  X\otimes_\alpha Y, \ \, \|x\|_{X\otimes_\ve Y} \le \| x \|_{X\otimes_\alpha Y} \le \| x \|_{X\otimes_\pi Y};$$
 		\item $\alpha$ satisfies the metric mapping property.
 	\end{itemize}
 
 Later on, in Section \ref{Sec6} we will more generally refer as tensor norms to  the \emph{tensorization} of  different tensor norms. For example, if $\alpha,\, \alpha'$ are tensor norms, the assignment on any three Banach spaces $X,\, Y,\, Z$ of the norm $(X \otimes_\alpha Y) \otimes_{\alpha'} Z$ will also be referred as tensor norm.
 
 The last tensor norm that we need is the 2-summing norm: for an operator $f \in \BB(X,Y)$, 
 \begin{equation}\label{Def_2summing}
 \| f \|_{\pi_2(X,Y) } := \big \|  \id \otimes f: \ell_2 \otimes_\ve X \rightarrow \ell_2(Y)  \big \|,
 \end{equation}
 where the norm in $\ell_2(Y)$ is defined by $\| (y_i)_{i\in \mathbb{N}} \|_{\ell_2(Y)} = ( \sum_{i\in \mathbb{N}} \| y_i \|_Y^2 )^{1/2}$ for any sequence  of elements $y_i \in Y$.
 	

Next we introduce Schatten classes of compact operators between Hilbert spaces, that are the model to define the generalizations in the theory of operator ideals that we use later on. To define the $p$-th Schatten class $\Ss_p(\HH)$, for $ 1 \le p\le \infty	$, we associate to any compact operator on a Hilbert space, $f: \HH \rightarrow \HH$, its sequence of singular values $( s_i(f) )_{i\in \mathbb{N}}$, where $s_1(f)\le s_2(f)\le \ldots$. With this sequence, we define the norm $\| f\|_{\Ss_p(\HH)} := \big\|  ( s_i(f) )_i \big\|_{\ell_p}$, which provides the normed structure on $\Ss_p(\HH)$. We use the simpler notation $\Ss_p$ to denote the $p$-th Schatten class of operators on the separable Hilbert space $\ell_2$.  In the finite dimensional case we use the notation $\Ss_p^{n,m}$ to refer to the $p$-th Schatten class of operators from $\ell_2^m$ into $\ell_2^n$. Notice that the case $p= \infty$ coincides with the operator ideal we denoted before as $M_{n,m}$, while for $p=1$ we obtain $\Ss_1^{n,m}$.

 	Now, moving into operators between arbitrary Banach spaces we define:
 	 	\begin{definition}\label{Def_WeakSchatten}
 		Given an operator $f: X \rightarrow Y$ and  $1 \le p \le \infty$ we say that $f$ is of weak Schatten-von Neumann type $\ell_p$ if
 		$$
 		\|   f \|_{\mathfrak{S}_{p}^w (X,Y) }  :=  \sup \left \lbrace   \Big\|  \Big ( s_i (g\circ f\circ h) \Big )_i \Big \|_{\ell_p}  \ : \ \begin{array}{l} \| g:Y \rightarrow \ell_2   \|\le 1 \\[0.5em]  \| h: \ell_2 \rightarrow X\|\le 1  \end{array}  \right \rbrace <\infty,
 		$$
 		where  $( s_i ( g\circ f\circ h) \big )_i $ is the sequence of singular values  of the operator $g\circ f\circ h: \ell_2 \longrightarrow \ell_2 $. 
 		
 		We denote by $\mathfrak{S}_{p}^w (X,Y)$ the space of operators $f: X \longrightarrow Y$ of weak Schatten-von Neumann type $\ell_p$.	 Alternatively, in the tensor product notation, we refer to this space by $ X^* \otimes_{\mathfrak{S}_{p}^w} Y$.
 	\end{definition}

To finish this section we introduce the space $\mathfrak{S}_p^{w-cb}(X,Y)$ announced at the beginning of this section. Its definition is based on Definition \ref{Def_WeakSchatten} and it incorporates elements of the theory of operators spaces. This forces us to endow $X$ and $Y$ with operator space structures (o.s.s), that is, norms on the \emph{matrix levels} of these spaces, $M_n (X)\equiv M_n \otimes X$, $M_n(Y)\equiv M_n \otimes Y$ for any $n\in \mathbb{N}$ -- see \cite{RuanBook} or \cite{pisier89_book} for a detailed exposition on operator spaces. With that, the natural notion for maps between operator spaces is the notion of \emph{completely bounded} operators, that is, linear operators $f:X\rightarrow Y$ such that
$$
\| f: X\rightarrow Y \|_{cb} : = \sup_{n\in \mathbb{N}}  \| \id \otimes f : M_n(X) \rightarrow M_n(Y) \| <\infty.
$$

The Banach space of completely bounded operators between $X$ and $Y$ is denoted by $\mathcal{C}\BB(X,Y) $. Identifying again linear maps with elements of the tensor product $X^*\otimes Y$,  in the finite dimensional case we denote $\mathcal{C}\BB(X,Y)  =: X^* \otimes_{min} Y$.

A Banach space can be endowed in general with several o.s.s. In the case of the space $\BB(\HH,\KK)$, with $\HH$ and $\KK$, Hilbert spaces, there is a natural o.s.s. determined by promoting the isomorphism  $M_n\left(\BB(\HH,\KK) \right) $ $ \simeq $ $ \BB( \HH^{\otimes n}, \KK^{\otimes n})$ to an isometry (fixing that way the norm in the \emph{matrix levels} of the space)\footnote{ Here we considered hilbertian tensor products in such a way that $\HH^{\otimes n}$ and $\KK^{\otimes n}$ are again Hilbert spaces.}. For a Hilbert space $\HH$, we introduce here the so-called \emph{row} and \emph{column} o.s.s., denoting the corresponding operator spaces $R$ and $C$, respectively.  $R$ is defined via the \emph{row} embedding:
$$
\HH \simeq \BB(\ell_2, \mathbb{C}),
$$
from which we define a norm on $M_n(\HH) $ considering the following isomorphism to be an isometry: $M_n(\HH) \simeq M_n\left( \BB(\HH, \mathbb{C}) \right)\simeq \BB (\HH^{\otimes n},   \ell_2^n).$

Similarly, $C$ is defined substituting the previous \emph{row} embedding by it \emph{column} version
$$
\HH \simeq \BB( \mathbb{C}, \ell_2).
$$

These last two operator spaces turn out to be non-isomorphic, on the contrary to what happens at the Banach level, where they are simply Hilbert spaces. They are still dual between themselves, that is, $C^*  \simeq R$ and $C \simeq R^*$ completely isometrically\footnote{ Meaning that not only $C^*  \simeq R$ and $C \simeq R^*$ stand isometrically but also $M_n(C^*)  \simeq M_n(R)$ and $M_n(C) \simeq M_n(R^*)$ for any $n \in \mathbb{N}$.}. However, to properly state those identifications we need a notion of duality for operator spaces. This notion is induced by that of completely bounded maps introduced before. We say that, for an operator space $X$, $X^*$ is its dual if
$$
M_n(X^*) = \mathcal{C}\BB(X, M_n) \quad \text{ for any } n\in \mathbb{N}.
$$
Notice that for $n=1$ the previous characterization of $X^*$ coincides with  the dual as Banach spaces\footnote{ For that it is necessary to consider the fact that $\mathcal{C}\BB(X, \CC) \simeq \BB(X,\CC)$.}. As a last comment on operator spaces, we note that this duality allows to endow $\Ss_1(\HH)$ with a natural o.s.s. as the dual of $\BB(\HH)$. Now we finally have all the ingredients to  define:
 	 	\begin{definition}\label{Def_cbWeakSchatten}
	Given an operator between operator spaces $f: X \rightarrow Y$ and  $1\le p \le \infty$ we say that $f$ is of weak-cb Schatten-von Neumann type $\ell_p$ if
	$$
	\|   f \|_{\mathfrak{S}_{p}^{w-cb} (X,Y) }  :=  
	\sup \left \lbrace \Big\| \left( s_i\left( g\circ f \circ h \right) \right)_i \Big\|_{\ell_p} \ : \ \begin{array}{l} \big \| \, g: Y \longrightarrow C\, \big \|_{cb} \le 1\\[0.5em] \big \| \,  h: R \longrightarrow X \,  \big \|_{cb} \le 1  \end{array}  \right \rbrace
	 <\infty,
	$$
	where  $( s_i ( g\circ f\circ h) \big )_i $ is the sequence of singular values  of the operator $g\circ f\circ h: \ell_2 \longrightarrow \ell_2 $. 
	
	We denote by $\mathfrak{S}_{p}^{w-cb} (X,Y)$ the space of operators $f: X \longrightarrow Y$ of weak-cb Schatten-von Neumann type $\ell_p$.	  Alternatively, in the tensor product notation, we refer to this space by $ X^* \otimes_{\mathfrak{S}_{p}^{w-cb}} Y$.
\end{definition}

\begin{remark}\label{Rmk_sigma^w}
	Since $B_{\mathcal{CB}(X,Y)} \subseteq B_{\BB(X,Y)}$ for any operator spaces $X$, $Y$, it follows that 
	\beq\label{Prop_OpI2}
	\|   f \|_{\mathfrak{S}_{p}^{w-cb} (X,Y) } \le \|   f \|_{\mathfrak{S}_{p}^{w} (X,Y) },
	\eeq
	for  any $1\le p \le \infty$ and any $f \in \mathfrak{S}_{p}^{w-cb} (X,Y) $.
\end{remark}

 	Before ending this section, we provide an alternative characterization of the norm introduced above when $X = M_{n,m}  $,  $Y=\s^{n, m}$ and $p=2$. That is the case appearing in our study of cheating strategies for PV in Section \ref{Sec4}. For that, we understand $\mathfrak{S}_{2}^{w-cb} ( M_{n,m}  ,\s^{n, m})$ as the tensor product $\s^{n,m} \otimes_{\mathfrak{S}_2^{w-cb}} \s^{n,m}$. Then,
 	\begin{lemma}\label{lemma_CharacSigma}
 		Given a tensor $f\in \s^{n,m} \otimes \s^{n, m}$, where $\s^{n, m}$ is  endowed with its natural o.s.s. (as the dual of $M_{n,m}$),  we have that:
 		$$
 			\|   f \|_{\s^{n,m} \otimes_{\mathfrak{S}_2^{w-cb}} \s^{n,m} }  
 			=
 			\sup_{ \begin{subarray}{c}
 				r \in \mathbb{N}\\
 				 g, h \in B_{ M_{nr,m} }
 				\end{subarray}     } 
 			 \big \|   ( h \otimes  g)(  f)  \big\|_{\ell_2^{r^2}} .
 		$$
 		Above, the action of $ h = \sum_{i=1}^n \sum_{j=1}^r  \sum_{l=1}^m    h_{ijl} |i j\xa  l|  \in M_{nr,m}$ on a tensor $ t = \sum_{i=1}^n  \sum_{l=1}^m $ $t_{il} |i\xa l|  \in \s^{n,m}$ is defined by
 		$$
 		h( t) :=  \sum_{j=1}^r   \left( \sum_{i=1}^n  \sum_{l=1}^m  h_{ijl}  t_{il}  \right) \, |j\ra \in \ell_2^r.
 		$$
 	\end{lemma}
 	
 \begin{proof}
 	The claim follows from the following observations:
 	\begin{itemize}
 		\item a standard argument shows that the supremum in Definition \ref{Def_cbWeakSchatten} can be taken over finite dimensional $C_r$ and $R_r$, where $r\in \mathbb{N} $ is arbitrarily large;
 		\item for an operator between Hilbert spaces, as $ g\circ f \circ h$ in Definition \ref{Def_cbWeakSchatten}, the $\ell_2$-sum of the singular values coincide with the Hilbert-Schmidt norm of the operator, which is the same as the Euclidean norm of the tensor associated. In our case, with a slight abuse of notation, the relevant tensor is $ ( h \otimes  g)(  f)$;
 		\item  finally, when we set $X = M_{n,m}  $,  $Y=\s^{n, m}$ in Definition \ref{Def_cbWeakSchatten}, the optimization is carried  over elements $g \in B_{ \mathcal{CB}  (\s^{n,m},C_r) }$ and  $ h  \in B_{ M_{n,m} }$. But now, it is again a standard result that the following are complete isometries \cite[Section 9.3]{RuanBook}: $ \mathcal{CB} (\Ss_1^{n,m} ,C_r ) \simeq M_{n r, m} \simeq  \mathcal{CB}  (R_r, M_{n,m} )$.  The claim of the lemma is obtained acting with $g,\, h $ viewed as elements in $B_{ M_{n r, m}  }$, as defined in the statement.
 	\end{itemize}

 \end{proof}
 	 	
 	\subsubsection{Interpolation of Banach spaces}\label{Sec.2.4.2}

 	Properties of interpolation spaces allow us to obtain estimates for the type constants of certain spaces that are useful for our purposes in this work. Here we restrict ourselves to the study of the complex interpolation space $(X_0,X_1)_\theta$ for $0< \theta < 1 $ and finite dimensional Banach spaces $X_0$, $X_1$. We decided to avoid  here a full treatment of the rather cumbersome definition of these spaces and focus on stating some natural properties they display. That is enough for the scope of our work. We redirect the interested reader to the classical references \cite{BerghLofstrom76,Triebel78}.
 	
 	In our case, in which $X_0$, $X_1$ are finite dimensional, the space $(X_0,X_1)_\theta$ can always be constructed. In the general case, for arbitrary Banach spaces, if we still can define $(X_0,X_1)_\theta$ we say that the couple $(X_0,\, X_1)$ is compatible\footnote{ Technically, this condition is usually stated as the requirement that $X_0$ and $X_1$ embed continuously in a common  Hausdorff topological vector space.}, so we fix this terminology from now on.  For the sake of concreteness, here we will consider the case in which $X_0$, $X_1$ and $(X_0,X_1)_\theta$ are algebraically the same space but endowed with different norms. The complex interpolation method, that assigns to any compatible couple $(X_0,\, X_1)$ the space $(X_0,X_1)_\theta$, is an \emph{exact interpolation functor of exponent} $\theta$. This means that it satisfies the following:
 	\begin{theorem}[\cite{BerghLofstrom76}, Thm. 4.1.2.]\label{Int_IntProp}
 		For any compatible couples $(X_0,\, X_1)$, $(Y_0,\, Y_1)$, and any linear map $f : (X_0,X_1)_\theta  \rightarrow (Y_0,Y_1)_\theta$:
 		$$
 		\big \| f : (X_0,X_1)_\theta  \rightarrow (Y_0,Y_1)_\theta \big\| \le \big \| f : X_0  \rightarrow Y_0 \big\|^{1-\theta} \  \big \| f : X_1  \rightarrow Y_1 \big\|^\theta,
 		$$
 		where $\| \, \cdot \, \| $ above denotes the usual operator norm.
 	\end{theorem}
 	
 	Now we turn our attention to the classical sequence $\ell_p$ spaces. Interpolation in this case becomes remarkably natural. We have the isometric identification $ \ell_p = (\ell_\infty,\, \ell_1 )_{\sfrac{1}{p}} $ for any $1\le p \le \infty$. Indeed, such an identification follows in a much more general setting. For a Banach space $X$ and $p \in (0,\infty]$, let us denote $L_p(X)$ the space of p-integrable $X$ valued functions on the unit interval. That is, measurable functions $f:[0,1] \rightarrow X$ such that
 	$$
 	\| f \|_{L_p(X)} := \left( \int_0^1 \| f(t) \|_X^p \dd \mu(t) \right)^{\frac{1}{p}},
 	$$
 	for an (implicitly) given measure $\mu$. With that we can state:

 	\begin{theorem}[\cite{BerghLofstrom76}, Thm. 5.6.1.]\label{Int_Lp's}
 	For any compatible couple $(X_0,\, X_1)$,	$p_0,\, p_1 \in [1,\infty]$ and $\theta\in (0,1)$ the following follows with equal norms:
 		$$
 		\Big( L_{p_0} (X_0),\, L_{p_1}(X_1) \Big)_\theta = L_p \Big( (X_0,\, X_1)_\theta \Big),
 		$$
 		where $\frac{1}{p} =\frac{1 - \theta}{p_0} + \frac{\theta}{p_1}$.
 	\end{theorem}
 
 	Notice that $\ell_p(X)$ spaces can be regarded as particular instances of $L_p(X)$ where the natural numbers are identified with a subset of the interval $[0,1]$ and $\mu$ is fixed as the discrete measure with unit weights on that subset. This allows us to translate the previous statement also to this case:
 	\begin{equation}\label{Int_ellp's}
 	 		\Big( \ell_{p_0} (X_0),\, \ell_{p_1}(X_1) \Big)_\theta = \ell_p \Big( (X_0,\, X_1)_\theta \Big),
 	\end{equation}
 		where $\frac{1}{p} =\frac{1 - \theta}{p_0} + \frac{\theta}{p_1}$.
 
 Pleasantly, an analogue result for Schatten classes is also true.
 
  	\begin{theorem}[\cite{Pisier98}, Cor. 1.4.]\label{Int_Sp's}
 	For a	$p_0,\, p_1 \in [1,\infty]$ and $\theta\in (0,1)$ the following follows with equal norms:
 	$$
 	( \Ss_{p_0} ,\, \Ss_{p_1} )_\theta = \Ss_p,
 	$$
 	where $\frac{1}{p} =\frac{1 - \theta}{p_0} + \frac{\theta}{p_1}$. When it applies, $\Ss_\infty$ must be understood as the Banach  space (with the operator norm) of compact operators in a separable Hilbert space.
 \end{theorem}
 
 These are all the basic results we need regarding complex interpolation. To finish this section, we now relate some of the norms introduced in Section \ref{Sec.2.4.1} with the space $(X^* \otimes_\ve Y, X^* \otimes_\pi Y)_{\frac{1}{2}}$.
 	
 	\begin{proposition}\label{Prop_RelNorms}
 		Given finite dimensional Banach spaces $X$, $Y$, for any $f\in X \otimes Y)$,
 		$$
 		\big\| f \big\|_{ X \otimes_{\mathfrak{S}_2^{w-cb}} Y } \le \big\| f \big\|_{ X \otimes_{\mathfrak{S}_2^{w}} Y} \le \big \| f \big\|_{(X^* \otimes_\ve Y, X^* \otimes_\pi Y)_{\frac{1}{2}}}.
 		$$
 	\end{proposition}
 	\begin{proof}
 		Recalling that we have already established the first inequality in Remark \ref{Rmk_sigma^w}. Therefore we focus on the second inequality.
 		 
 		According to the definition of $\mathfrak{S}_2^{w} (X,Y)$, Definition \ref{Def_WeakSchatten}, we can directly write:
 		$$
 		\big\| f \big\|_{ \mathfrak{S}_2^{w} (X,Y)} = \sup_{ \begin{subarray}{c}
 			g \in B_{\BB(Y,\ell_2)}  \\
 			h \in B_{\BB(\ell_2,X)}
 			\end{subarray} } 
 		\| g\circ f \circ h \|_{\Ss_2} 
 		= 
 		\sup_{ \begin{subarray}{c}
 			g \in B_{\BB(Y,\ell_2)}  \\
 			h \in B_{\BB(\ell_2,X)}
 			\end{subarray} } 
 		\| g\circ f \circ h \|_{(\Ss_\infty, \s)_{1/2}},
 		$$
 		 		where we have used Theorem \ref{Int_Sp's} to state the last equality.
 		 		
 		The map $g\circ f \circ h : \ell_2 \rightarrow \ell_2$ can be interpreted, as a tensor, as the image of the mapping $  h^*\otimes g:  X^* \otimes Y \rightarrow \ell_2 \otimes \ell_2$ acting on $f$. With this, the previous expression can be rewritten as:
 		\begin{align*}
 		\big\| f \big\|_{ \mathfrak{S}_2^{w} (X,Y)} &= \sup_{ \begin{subarray}{c}
 			g \in B_{\BB(Y,\ell_2)}  \\
 			h \in B_{\BB(\ell_2,X)}
 			\end{subarray} } 
 		\| (h^*\otimes g)(f) \|_{(\Ss_\infty, \s)_{\frac{1}{2}}} 
 		\\
 		&\le 
 		\| f \|_{(X^* \otimes_\ve Y, X^* \otimes_\pi Y)_{\frac{1}{2}}}  \ 
 		\sup_{ \begin{subarray}{c}
 			g \in B_{\BB(Y,\ell_2)}  \\
 			h \in B_{\BB(\ell_2, X)}
 			\end{subarray} }
 		\| h^*\otimes g : (X^* \otimes_\ve Y, X^* \otimes_\pi Y)_{\frac{1}{2}}\rightarrow (\Ss_\infty, \s)_{\frac{1}{2}} \|  .
 		\end{align*}

 		Now, it only remains to show that for any contractions $h^* : X^* \rightarrow \ell_2$, $g: Y  \rightarrow \ell_2$
 		$$
 		\| h^*\otimes g : (X^* \otimes_\ve Y, X^* \otimes_\pi Y)_{\frac{1}{2}}\rightarrow (\Ss_\infty, \s)_{\frac{1}{2}} \| \le 1.
 		$$
 		This follows from the interpolation property, Theorem \ref{Int_IntProp}:
 		$$
 		\| h^*\otimes g: (X^* \otimes_\ve Y, X^* \otimes_\pi Y)_{\frac{1}{2}}\rightarrow (\Ss_\infty, \s)_{\frac{1}{2}} \| 
 		\le 
 		\| h^*\otimes g : X^* \otimes_\ve Y  \rightarrow  \Ss_\infty \|^{\frac{1}{2}}  \ \|  h^*\otimes g : X^* \otimes_\pi Y \rightarrow   \s  \|^{\frac{1}{2}} ,
 		$$
 		together with the understanding of $\Ss_\infty$ and $\s$ as the tensor products $\ell_2 \otimes_\ve \ell_2$ and $\ell_2 \otimes_\pi \ell_2$, respectively. This allows us to bound
 		$$
 		\| h^*\otimes g : X^* \otimes_\ve Y  \rightarrow  \Ss_\infty \|  = \| h^* : X^*  \rightarrow  \ell_2 \| \ \| g :  Y  \rightarrow  \ell_2\| \le 1,
 		$$
 		thanks to the metric mapping property displayed by the injective tensor norm, $\ve$ \eqref{MetricMapProp}. Analogously
 		$$
 		\| h^*\otimes g : X^* \otimes_\pi Y  \rightarrow  \Ss_1 \|  = \| h^* : X^*  \rightarrow  \ell_2 \| \ \| g :  Y  \rightarrow  \ell_2\| \le 1.
 		$$
 		Hence, the claim in the statement follows.
 	\end{proof}
 	
 	Being more specific, when $X^* = Y = \s^{n,m}$, Proposition \ref{Prop_RelNorms} reads
 	\begin{equation} \label{Eq_RelNorms}
 	\| f \|_{\s^{n,m} \otimes_{\mathfrak{S}_2^{w-cb}} \s^{n,m}}  \le \| f \|_{\s^{n,m} \otimes_{\mathfrak{S}_2^w} \s^{n,m}} \le \| f \|_{ (\s^{n,m} \otimes_\ve \s^{n,m},\, \s^{n,m} \otimes_\pi \s^{n,m} )_{\frac{1}{2}} }. 
 	\end{equation}

 	\subsubsection{Type/cotype of a Banach space}\label{Sec2.4.3}
 	
 	The key properties of a Banach space we study are its \emph{type} and \emph{cotype}. These are probabilistic notions in the local theory of Banach spaces that build on \emph{Rademacher} random variables\footnote{ There also exists in the literature a \emph{gaussian} notion of type/cotype. See, e.g., \cite{Tomczak1989banach}. Both notions are in fact intimately related, but here we only consider the \emph{Rademacher} version of the story.}. We call a random variable $\ve$  Rademacher if it  takes values $-1$ and $1$ with probability $1/2$ each. We refer by $\lbrace \ve_i\rbrace_{i=1}^n$ to a family of $n$ i.i.d. such random variables. Then, $\mathbb{E}_\ve \, \phi(\ve)$ denotes the expected value of a function $\phi$ 	over any combination of signs $\lbrace \ve_i\rbrace_{i=1}^n$ with uniform weight $1/2^n$.

 	\begin{definition}\label{typedef}Let $X$ be a Banach space and $1 \le p \le 2$. We say  that $X$ is of (Rademacher) type $p$ if there exists a positive constant $\rmT$ such that for every natural number $n$ and every sequence $\lbrace x_i \rbrace_{i=1}^n \subset X$ we have
 	\beq \nn \hspace{-3mm}
 	\left( \mathbb{E}_\ve   \Big[ \big\|  \sum_{i=1}^n \varepsilon_i  x_i \big\|_X^2 \Big]   \right)^{1/2}
 	\hspace{-1mm} \le	 \rmT \left( \sum_{i=1}^n \|x_i\|_X^p \right)^{1/p},
 	\eeq
 	Moreover, we define the Rademacher type $p$ constant $\mathrm{T}_p(X)$ as the infimum of the constants $\rmT$ fulfilling the previous inequality.\end{definition}

 The notion of type of a normed space finds a dual notion in the one of cotype:
 
 \emph{
 For $2 \le q < \infty$, the Rademacher cotype $q$ constant  of $X$, $\mathrm{C}_q(X)$, is the infimum over the constants $\mathrm C$ (in case they exist) such that the following inequality holds for every natural number $n$ and every sequence $\lbrace x_i \rbrace_{i=1}^n \subset X$,
 \begin{align*}\mathrm{\mathrm C}^{-1}\Big( \sum_{i=1}^n \|x_i\|_X^q\Big)^{1/q} \le 
 \left( \mathbb{E}_\ve   \Big[ \big\|  \sum_{i=1}^n \varepsilon_i  x_i \big\|_X^2 \Big]   \right)^{1/2} .
 \end{align*}
 In parallel with the previous definition, we also say that $X$ is of cotype $q$ if $\mathrm{C}_q(X) < \infty$. 
}

	\begin{comment}\label{Kahane}
		The above definitions can be found elsewhere  in an alternative form in  which the term $\Big( \mathbb{E}_\ve    \big\|  \sum_{i=1}^n \varepsilon_i  x_i \big\|_X^2   \Big)^{1/2} $ above is replaced  by  $\mathbb{E}_\ve   \|  \sum_{i=1}^n \varepsilon_i  x_i \big\|_X   $ or, in other cases, by $ \Big(  \mathbb{E}_\ve   \|  \sum_{i=1}^n \varepsilon_i  x_i \big\|_X^p \Big)^{1/p}  $. Due to Kahane inequality \cite{KahaneBook} (see also \cite[Section 4]{Tomczak1989banach} for the specific application of Kahane inequality to the present context) both expressions are equivalent up to a universal constant and there is no essential difference between definitions.
\end{comment}

If the number of elements $x_i$ in the definitions above is restricted to be at most some natural number $m$, we obtain the related notion of \emph{type/cotype constants of $X$ with $m$ vectors}, denoted here as $\mathrm{T}_p^{(m)} (X)$ and $\mathrm{C}_q^{(m)} (X)$. This is the precise notion we will  use later on. Although it will be frequently enough to work with the  notion of type constants, sometimes we will need to make this distinction.
 	
Coming back to the better studied context of type and cotype (without any restriction on the number of elements), it is well known that $X$ being of type $p$ implies cotype $q$ for the dual, $X^*$, where $q$ is the conjugate exponent such that $1/p + 1/q = 1$. This can be summarized in the inequality -- see, e.g., \cite{Maurey03}:
\begin{equation} \label{typecotype_duality2}
\mathrm{C}_q (X^*) \le \mathrm{T}_p (X), \qquad \text{ for } 1 < p \le 2,\ 2 \le q < \infty\ : \  \frac{1}{p} + \frac{1}{q} = 1.
\end{equation}

The reverse inequality fails in general -- and, in fact, the pair of spaces considered in this work, $(M_n$, $\Ss_1^n)$, is an instance of that phenomenon. However, it turns out that the reverse inequality can be made true \emph{up to logarithmic} factors \cite{Pisier82,Maurey03}:
\beq \label{typecotype_duality}
\mathrm{T}_p (X)  \lesssim \log(\dim(X)) \, \mathrm{C}_q (X^*)  , \qquad \text{ for } 1 < p \le 2,\ 2 \le q < \infty\ : \  \frac{1}{p} + \frac{1}{q} = 1.
\eeq


Our interest now turns into the interaction between type and interpolation. In  fact,  type constants behave well w.r.t. interpolation methods,  {a} fact that will be extremely useful in next section. We state the following general known result:
\begin{proposition}\label{Prop_IntType}
	Let $X_0, \, X_1$ be an interpolation couple, where $X_i$ has type $p_i$ for some $1\le p_i \le 2 $, $i=0,\,1$. Let $0 < \theta <1$ and $1 < p < 2 $ such that $\frac{1}{p} = \frac{1-\theta}{p_0} + \frac{\theta}{p_1}$. Then,
	$$
	\mathrm{T}_p \big( (X_0,X_1)_{\theta} \big)  \le \big( \mathrm{T}_{p_0} (X_0) \big)^{1-\theta}   \big( \mathrm{T}_{p_1} (X_1) \big)^{\theta}.
	$$   
\end{proposition}
The proof  follows easily from the interpolation properties of vector valued $\ell_p$ and $L_p$ spaces. We decided to include a simple proof next without any claim of originality.
\begin{proof}
	 An alternative characterization of the type-p constant of a Banach space $X$ is given by the norm of the mapping:
	 	$$
	 \begin{array}{cccc}
	 \mathrm{Rad}:&  \ell_p  (  X ) & \longrightarrow & L_2 (  X ) \\
	 &(x_{i})_{i} &\mapsto &  \sum_{i} \ve_{i} \, x_{i}
	 \end{array},
	 $$
	 where $\lbrace \ve_i \rbrace_i$ are i.i.d. Rademacher random variables and\footnote{ Formally, to establish this identification we  consider a realization of the random variables $\ve_i$ as real valued functions on the interval $[0,1]$. A standard choice is setting $\ve_i (t) = \mathrm{sign} \left(\sin(2^i \pi t) \right)$. In that way, for a function $\phi$ of the random variable $\ve$, $\mathbb{E}_\ve \phi(\ve) = \int_0^1 \phi(\ve(t))\dd t$, which makes the connection with $L_p$ spaces.}
	 $$
	 \|  \sum_{i} \ve_{i} \, x_{i} \|_{L_2(X)} :=  \left( \mathbb{E}_\ve \big\| \sum_i \ve_i x_i  \big \|^2_X \right)^{\frac{1}{2}}.
	 $$
	  Then, we write
	  \begin{align*}
	  \mathrm{T}_p \left( (X_0,X_1)_{\theta} \right)  &= \left\| \mathrm{Rad}:  \ell_p  \left( (X_0,X_1)_{\theta} \right)  \longrightarrow  L_2  \left( (X_0,X_1)_{\theta} \right)  \right\|.
	  \end{align*}
	  
	  Taking into account the equivalences (Theorem \ref{Int_Lp's}):
	  $$
	  \ell_p  \left( (X_0,X_1)_{\theta} \right) = \left( \ell_{p_0}  (X_0) ,   \ell_{p_1}( X_1 )  \right)_{\theta} , 
	  \qquad 
	   L_2  \left( (X_0,X_1)_{\theta} \right) = \left( L_{2}  (X_0) ,   L_{2}( X_1 )  \right)_{\theta} ,
	   $$
	  we can bound:
	    \begin{align*}
	   \left\| \mathrm{Rad}:  \ell_p  \left( (X_0,X_1)_{\theta} \right)  \longrightarrow  L_2  \left( (X_0,X_1)_{\theta} \right)  \right\|&
	  \\
	  & \hspace{-2.4cm}  \le 
	  	   \left\| \mathrm{Rad}:  \ell_{p_0}  (X_0)  \longrightarrow  L_{2}  (X_0)   \right\|^{1-\theta}
	  	   \
	  	   	   \left\| \mathrm{Rad}: \ell_{p_1}  (X_1)  \longrightarrow  L_{2}  (X_1)    \right\|^\theta
	  \\
	  & \hspace{-2.4cm}  =
	  \left( \mathrm{T}_{p_0} (X_0)  \right)^{1-\theta}	 \left( \mathrm{T}_{p_1} (X_1)  \right)^{\theta} .
	  \end{align*}
\end{proof}

 	\subsubsection{Vector valued maps on the Boolean hypercube}\label{Sec2.4.4}
 	
 	The main idea in this work is based in the study of strategies to break a particular family of PV protocols -- referred to as $G_{Rad}$ -- as assignments on the boolean hypercube $\calQ_m = \lbrace -1, 1 \rbrace^m$. We will associate to any cheating strategy a vector valued mapping  $\Phi: \calQ_m \rightarrow X$, for some Banach space $X$.  Regular enough $\Phi$'s will lead to good lower bounds on resources required by the cheaters, contributing   to the understanding of Question \ref{question1}. To quantify the regularity of  such maps we introduce the following parameter (depending also on the choice of the space $X$):
 	\begin{definition}\label{Def_sigma}
 		To any Banach-space valued map  $\Phi: \calQ_m \rightarrow X$ we associate the parameter:
 		$$
 		\sigma_\Phi : = \log(m) \  \mathbb{E}_{\ve\in \calQ_m}\,  \left( \sum_{i=1 }^m \| \partial_i  \Phi_S(\ve)\|^2_{X} \right)^{1/2},
 		$$
 		where $\partial_i \Phi(\ve): = \frac{\Phi(\ve_1,\ldots, \ve_i ,\ldots ,\ve_m)-\Phi(\ve_1,\ldots, - \ve_i ,\ldots ,\ve_m)}{2}$ is the discrete derivative on the boolean hypercube in the  i-th direction.
 	\end{definition}
 	
 	Intuitively, $\sigma$ is an average on both the point $\ve$  and the direction $i$ (unnormalized in this last case) of the magnitude of the derivative of the map $\Phi$. The prefactor $\log(m)$ is of minor importance for our purposes and we added it to the definition of $\sigma_\Phi$ to obtain more compact expressions later on.
 	
 	 \begin{example}\label{Example_sigma1}
 	 In order to gain some familiarity, let us compute the parameter $\sigma$ of a linear  map 
 	 $$
 	  \begin{array}{rccc}
 	 \Phi  : &  \calQ_{m} & \longrightarrow &  X
 	 \\[1em]
 	 & \ve & \mapsto  & \Phi (\ve): = \frac{1}{m} \sum_j \ve_j x_j 
 	  	 \end{array} ,
 	 $$
 	  where $x_j \in B_X$ for $j= 1,\ldots m$
 	  
 	  First, for any point $\ve \in \calQ_m$, and a direction $i\in[m]$:
 	$$
 	\partial_i \Phi(\ve) 
 	=\frac{1}{2m} \ \left(  \sum_{j} \ve_j x_j - \ve_j (-1)^{\delta_{i,j}} x_j   \right)
 	= \frac{1}{m} \ve_i \,x_i.
 	$$
 	
 	Therefore,
 	$$
 	\sigma_\Phi = \frac{\log(m)}{m}  \left( \sum_i \| x_i \|_X^2 \right)^{\frac{1}{2}} \le   \frac{\log(m)}{m^{\frac{1}{2}}}.
 	$$
 	
 	This is the ideal case in which our results lead directly to powerful lower bounds on the resources required to break our PV protocols.
 \end{example}
Ultimately, the motivation for the definition of $\sigma_\Phi$ is the bound in Corollary \ref{Cor1} below. This is a consequence of  the following  Sobolev-type inequality due to Pisier for vector-valued functions on the hypercube:
 	\begin{lemma}[\cite{Pisier86}, Lemma 7.3] \label{lemmaPisier} In a Banach space $X$, let  $p\ge 1 $, $\Phi: \calQ_m \rightarrow X$ and $\ve ,\, \tl \ve$ be independent random vectors uniformly distributed  on $\calQ_m$. Then,
 		$$
 		\mathbb{E}_{\ve} \Big \| \Phi (\ve) - \mathbb{E}_{\ve} \Phi (\ve) \Big \|_X^p \le ( C \log m)^p \  \mathbb{E}_{\ve,\tl\ve} \Big \|   \sum_i \tl \ve_i \partial_i \Phi(\ve)   \Big{\|}^p_X ,
 		$$
 		where $\partial_i \Phi(\ve): = \frac{\Phi(\ve_1,\ldots, \ve_i ,\ldots ,\ve_n)-\Phi(\ve_1,\ldots, - \ve_i ,\ldots ,\ve_n)}{2}$.
 	\end{lemma}
 
 It is now very easy to combine this result with the type properties of $X$ in order to obtain:
 
 \begin{corollary}[of Lemma \ref{lemmaPisier}]\label{Cor1}
 	Consider a function $\Phi: \calQ_m \longrightarrow X$, where $X$ is a Banach space. Then 
 	$$
 	\mathbb{E}_\ve\big \| \Phi(\ve) \big\|_X  \le \big \|\mathbb{E}_\ve\Phi (\ve) \big \|_X + C \ \sigma_\Phi \ \mathrm{T}^{(m)}_2 ( X ) ,
 	$$
 	where $C$ is an independent constant.
 \end{corollary}

This is the cornerstone of the building leading to Theorem \ref{mainThm}.

\begin{proof}[Proof of Corollary \ref{Cor1}]
	Fix $p = 1 $ in Lemma \ref{lemmaPisier}. Therefore, we have that :
	$$
	 		\mathbb{E}_{\ve} \Big \| \Phi (\ve) - \mathbb{E}_{\ve} \Phi (\ve) \Big \|_X \le ( C \log m) \  \mathbb{E}_{\ve,\tl \ve} \Big \|   \sum_i \tl \ve_i \partial_i \Phi(\ve)   \Big{\|}_X .
	$$
	
	Additionally, we can trivially bound:
	$$
	\mathbb{E}_{\ve} \Big \| \Phi (\ve) - \mathbb{E}_{\ve} \Phi (\ve) \Big \|_X \ge 	\mathbb{E}_{\ve} \Big \| \Phi (\ve) \Big \|_X - \Big \| \mathbb{E}_{\ve} \Phi (\ve) \Big \|_X.
	$$
	
	On the other hand, according to the definition of the type-2 constant (with $m$ vectors) of $X$ -- recall also Comment  \ref{Kahane} -- we can say:
	$$
	\mathbb{E}_{\ve,\tl \ve} \Big \|   \sum_i \tl \ve_i \partial_i \Phi(\ve)   \Big{\|}_X
	\lesssim \mathrm{T}_2^{(m)}(X)  \ \mathbb{E}_{\ve} \,   \left( \sum_{i}  \| \partial_{i }  \Phi_S(\ve)\|^2_{X} \right)^{1/2}.
	$$
	
	That's enough to obtain the statement.
\end{proof}

\subsubsection{Some key  estimates of type constants}

Corollary \ref{Cor1} provides us with a tool to upper bound the expected norm of the image of a map $\Phi: \calQ_m \rightarrow X$, provided that we have some control over the RHS of the inequality in the statement. The only piece there that is independent of the map $\Phi$ is the type-2 constant (with $m$ vectors) $\rmT_2^{(m)}(X)$, to which the rest of this section is devoted.

Later on, the normed spaces $M_{n,m}$ and $\Ss_1^{n,m} \otimes_{\mathfrak{S}_2^{cb-w}} \Ss_1^{n,m}$ will play a prominent role. The type and cotype properties of $M_{n,m}$ as well as $\Ss_1^{n,m}$  are well known. In particular the following estimates hold:
\begin{equation}\label{Eq4_TypeConsts_1}
\begin{array}{c} \mathrm{C}_2 (M_{n,m}) \approx \min(n^{1/2}, m^{1/2}), \qquad \rmT_2( M_{n,m} ) \approx \log^{1/2} (\min(n,m) ), \end{array}
\end{equation}
\begin{equation}\label{Eq4_TypeConsts_1.1}
\begin{array}{c} \mathrm{C}_2 (\s^{n,m}) \approx 1, \qquad  \rmT_2( \s^{n,m} ) \approx  (\min(n,m) )^{1/2}. \end{array}
\end{equation}

For $\Ss_1^{n,m} \otimes_{\mathfrak{S}_2^{cb-w}} \Ss_1^{n,m}$ the situation is not that well understood at all. In fact, we were not able to obtain any non-trivial estimate for its type properties so far. Then, instead of dealing directly with this space, we will consider the interpolation space  $ (\s^{n,m} \otimes_\ve \s^{n,m},\s^{n,m} \otimes_\pi \s^{n,m} )_{\frac{1}{2}}$. The norm in this latter space turns out to be an upper bound to the norm in $\Ss_1^{n,m} \otimes_{\mathfrak{S}_2^{cb-w}} \Ss_1^{n,m}$, recall  Proposition \ref{Prop_RelNorms}. From now on we use the following notational short-cut: $ (\s^{n,m} \otimes_\ve \s^{n,m},\s^{n,m} \otimes_\pi \s^{n,m} )_{\theta} = \s^{n,m} \otimes_{ (\ve,\pi)_{\theta} } \s^{n,m}$. Thanks to the extra structure in $\s^{n,m} \otimes_{ (\ve,\pi)_{1/2} } \s^{n,m} $ provided by  interpolation, we are able to obtain a bound for its type constants. To simplify the presentation, we consider in the following that $\min(n,m) = n$. We can state:
\begin{proposition}\label{Type_Inttheta} Given $0< \theta <1$, and natural numbers $n\le m$: 
	$$
	\mathrm{T}_{ \frac{2}{1+\theta} }\left( \s^{n,m} \otimes_{ (\ve,\pi)_\theta } \s^{n,m} \right) \lesssim_{\log}    n^{\frac{1-\theta}{2}}.
	$$
\end{proposition}

An immediate consequence of the previous proposition is a bound for the type-2 constant with $n^2$ vectors:
\begin{equation*} 
\mathrm{T}_{ 2 }^{(n^2)} \left( \s^{n,m} \otimes_{ (\ve,\pi)_\theta } \s^{n,m}  \right) \le  n^{\theta} \ \mathrm{T}_{ \frac{2}{1+\theta} }\left( \s^{n,m} \otimes_{ (\ve,\pi)_\theta } \s^{n,m} \right)  \lesssim_{\log}    n^{ \frac{1+\theta}{2} },
\end{equation*}
where the first inequality follows as an application of Hölder inequality in the definition of $\rmT_2^{(n^2)}(X)$ (recall Definition \ref{typedef} and comments afterwards).

Particularizing for $\theta = \frac{1}{2}$:
\begin{equation}\label{Type_Int1/2} 
\mathrm{T}_{ 2 }^{(n^2)} \left( \s^{n,m} \otimes_{  (\ve,\pi)_\frac{1}{2}  } \s^{n,m}  \right)  \lesssim_{\log}    n^{ \frac{3}{4}}.
\end{equation}

This is the key type-estimate to obtain part II. of the main Theorem \ref{mainThm}.

For the sake of concreteness, we explicit here the logarithmic corrections in \eqref{Type_Int1/2}:
$$
\mathrm{T}_{ 2 }^{(n^2)} \big( \s^{n,m} \otimes_{  (\ve,\pi)_\frac{1}{2}  } \s^{n,m}  \big)  \lesssim    n^{ 3/4 }  \log^{1/2} (n m) \log(n)  .
$$

\begin{proof}[Proof of Proposition \ref{Type_Inttheta}]
	The proof proceeds in two steps. First, using techniques from \cite{Pisier90,Pisier1992}, we obtain the estimate 
	\beq \label{type_epsilon} \mathrm{T}_2 (\Ss_1^{n,m} \otimes_\ve \Ss_1^{n,m}) \lesssim_{\log}   n ^{1/2}  .\eeq
	
	With this at hand, Proposition \ref{Type_Inttheta} follows from how type constants interact with the complex interpolation method, Proposition \ref{Prop_IntType}. In particular, it is enough to fix $p_0 =2$, $p_1 = 1$ in that result and consider the trivial bound $\rmT_1 (\s^{n,m} \otimes_\pi \s^{n,m}) =1 $. 
	
	Therefore, there remains to provide a proof for \eqref{type_epsilon}. To prove the stated estimate we bound the cotype-2 constant of the dual, $M_{n,m} \otimes_\pi M_{n,m}$. Therefore, from the duality between type and cotype,  Equation \eqref{typecotype_duality}, we obtain:
	$$
	\mathrm{T}_2 (\Ss_1^{n,m} \otimes_\ve \Ss_1^{n,m}) \lesssim \log(nm) \, \mathrm{C}_2(M_{n,m} \otimes_\pi M_{n,m}).
	$$
	
	To estimate $ \mathrm{C}_2(M_{n,m} \otimes_\pi M_{n,m})$ we use the following bound on the cotype of the projective tensor product, implicit in \cite{Pisier90}\footnote{ The key result here is Theorem 5.1 in \cite{Pisier90}. The bound we use is obtained keeping track of the constants appearing in the isomorphic statement of that theorem.  We are indebted to Jop Briët for kindly sharing with us some very useful private notes on Pisier's method.}
	$$
	\mathrm{C}_2(M_{n,m} \otimes_\pi M_{n,m}) \lesssim \mathrm{C}_2 (M_{n,m}) \, \mathrm{UMD}(M_{n,m})\, \mathrm{T}_2^2(M_{n,m}),
	$$
	where $\mathrm{UMD}(X)$ is the analytic UMD (unconditional martingale difference) parameter of the Banach space $X$. We now bound each of the quantities in the RHS of the last inequality:
	\begin{itemize}
		\item recalling \eqref{Eq4_TypeConsts_1} we have that $\mathrm{C}_2 (M_{n,m}) \lesssim n^{1/2}$ and $ \mathrm{T}_2(M_{n,m}) \lesssim \log^{1/2}(n)$;
		\item we estimate $\mathrm{UMD}(M_{n,m})$ from known bounds for the UMD constant of the $p$-Schatten class $\Ss_p$, for $ 1<p <\infty$. It is known that these spaces are UMD and the following estimate for $\mathrm{UMD}(\Ss_p)$ is available \cite{Randrianantoanina2002}:
		$$
		\mathrm{UMD}(\Ss_p) \lesssim p.
		$$
		This also translates to the same bound for the subspace $\Ss_p^{n,m}$. Now, we take into account the following relation between the UMD constants of arbitrary spaces $X$ and $Y$ at Banach-Mazur distance $d(X,Y)$. This is a direct consequence of the geometric characterization of the UMD property due to Burkholder \cite{Burkholder81} -- see also \cite{Burkholder86}:
		$$
		\mathrm{UMD}(X) \lesssim d(X,Y) \, \mathrm{UMD}(Y).
		$$
		Finally, with this at hand, we obtain the bound 
		$$
		\mathrm{UMD}(M_{n,m}) \lesssim d(M_{n,m},\Ss_p^{n,m}) \, \mathrm{UMD}(\Ss_p^{n,m}) \lesssim n^{1/p} \, p.
		$$
		Adjusting the parameter $p$ as $ p =  \log(n)$ we obtain
		$$
		\mathrm{UMD}(M_{n,m}) \lesssim \log(n) ,
		$$
		that is enough to conclude that
		$$  \mathrm{T}_2 (\Ss_1^{n,m} \otimes_\ve \Ss_1^{n,m}) \lesssim  \log(nm)\, \log^2(n)\, n ^{1/2} .$$
	\end{itemize}
	
\end{proof}

 \section{Cheating strategies for $\mathbf{G_{Rad}}$}\label{Sec3}
 
 In this section we describe in detail the action of cheaters in our PV protocol $G_{Rad}$.    Recall that in 1-D PV, we consider a privileged point $x$ and a couple of verifiers, $V_L$, $V_R$, at locations $ x \pm \delta$. See Section \ref{Sec4_2}, page 7, for the definition of $G_{Rad}$. In the dishonest scenario  two cheaters, Alice and Bob,  hold locations $x \pm \delta'$ for some $0 < \delta' < \delta$.  The strategy of Alice and Bob is restricted to the \emph{s2w} scenario already described in Figure \ref{s2wStrategies}.

 A strategy in this scenario is determined by --  cf. Figure \ref{s2wStrategies}:
\begin{itemize}
	\item a shared entangled state $ \varphi  \in \mathcal{D}(\HH_{E_a} \otimes \HH_{E_b} ) $ that we assume here to be pure\footnote{ It can be easily checked that, by convexity, the success probability achieved in $G_{Rad}$ by strategies using mixed states is always upper bounded by the success probability when using pure states. Since the quantity we are interested in  is the optimal cheating probability, restricting ourselves to strategies using pure states would be enough.    }. From now on we use  interchangeably the notations $\varphi$ or $|\varphi\xa \varphi|$ to refer to that state; 
	\item a family of tuples of four ``local'' channels: for each $\ve \in \calQ_{n^2}$,
	 $$\mathcal{A} \in \CPTP(\HH_A \otimes \HH_B \otimes \HH_{E_a}, \HH_{A \shortrightarrow A }\otimes \HH_{A \shortrightarrow B }),\quad \BB_\ve \in \CPTP( \HH_{E_b}, \HH_{B \shortrightarrow B } \otimes \HH_{B \shortrightarrow A}), $$  $$ \tilde{\mathcal{A}}_\ve \in \CPTP( \HH_{A \shortrightarrow A }\otimes \HH_{B \shortrightarrow A }, \HH'_{A}), \quad \tilde \BB_\ve \in \CPTP(  \HH_{B \shortrightarrow B }\otimes \HH_{A \shortrightarrow B }, \HH'_{B} ) .$$
	 
	  For verification, $\HH'_A$, $\HH'_B$ should be communicated to $V_L$ and $V_R$ respectively. Therefore, according to the definition of the protocol, these registers should be isomorphic to the originals $\HH_A$ and $\HH_B$.
\end{itemize}

Understood as a family of quantum channels, the strategy defined by these elements reads:
	\begin{equation}\label{StratS2w}
\begin{array}{cccc}
\Ss_\ve: &				\D( \HH_A \otimes \HH_B 	)			   &  \longrightarrow  &     \D(\HH_A \otimes \HH_B)  \\[1em]
&  \psi  &  \mapsto    &  \Ss_\ve (\psi  ) = (\tilde{\mathcal{A}}_\ve \otimes \tilde \BB_\ve)\circ(\mathcal{A} \otimes \BB_\ve) (\psi  \otimes \varphi)
\end{array} ,
\end{equation}
 for each $ \ve \in \calQ_{n^2}$. 
 
  The probability that the verifiers accept the output of such a strategy is given by:
 	 		\begin{align}\label{DefvalueStrat_MROQG_GRad2}
 	\omega (G_{Rad};\lbrace \Ss_\ve \rbrace_\ve )  &:=   \mathbb{E}_{\ve} \ \Tr\left[\: |\psi_{\ve} \xa\psi_{\ve} |  \,    (\id_{C}\otimes \mathcal{S}_{\ve})  (|\psi \xa \psi|  ) \: \right].
 \end{align}
 Optimizing over any strategy allowed in the s2w scenario leads to the value:
\begin{align}\label{DefvalueStrat_MROQG_GRad1}
	\omega_{s2w}(G_{Rad}) = \sup_{\lbrace \Ss_\ve \rbrace_\ve \in \mathfrak{S}_{s2w} } \omega (G_{Rad};\lbrace \Ss_\ve \rbrace_\ve )  ,
\end{align}
where $\mathfrak{S}_{s2w}$ denotes the set of strategies in the s2w scenario.

 In this language, the existence of general attacks for arbitrary PV protocols translates into the coincidence of the value in the $s2w$ scenario with the honest value:
 \beq
 \label{CoincidenceValues}
 \omega_{s2w}(G_{Rad}) =1.
 \eeq 
 
 As we said in the introduction, the main question we are interested in is the amount of entanglement necessary to establish this equality. It is natural then to define a restricted version of $\omega_{s2w}(G_{Rad})$  considering only strategies using a limited amount of resources. Here, we restrict the local dimension at any time during the protocol. For $\tilde k,\, k\in \mathbb{N}$ we define  the scenario $\mathfrak{S}_{s2w,\tilde k , k}$ as the set of strategies in the form of \eqref{StratS2w} but with the following restrictions:
  $$\dim(\HH_{E_{a (b)}})\le k, \qquad \dim(\HH_{A (B)  \shortrightarrow A(B) } )  \times  \dim( \HH_{A (B)  \shortrightarrow B (A) } )\le \tilde k.$$ I.e., we restrict,
 $$
 \varphi \in \mathcal{D}(\ell_2^{k^2})
 $$
 and, for each $\ve \in \calQ_{n^2}$,
 $$\mathcal{A} \in \CPTP(\ell_2^{n^2 k}, \, \ell_2^{\tilde k}),\quad \BB_\ve \in \CPTP( \ell_2^{ k}, \, \ell_2^{\tilde k}  ), $$  $$ \tilde{\mathcal{A}}_\ve \in \CPTP( \ell_2^{\tilde k}, \, \ell_2^{n}   ), \quad \tilde \BB_\ve \in \CPTP(  \ell_2^{\tilde k}, \, \ell_2^{n}  ) .$$

Given this model, we define:
  \begin{align}\label{eq:ws2w}
	\omega_{s2w;\tilde k, k}(G_{Rad}) :=   \sup_{\lbrace \Ss_\ve \rbrace_\ve \in \mathfrak{S}_{s2w;\tilde k ,k} } \omega (G_{Rad};\lbrace \Ss_\ve \rbrace_\ve ) . 
\end{align}

Clearly,
\beq\label{eq:LimS2wH}
\lim_{\tilde k, k \shortrightarrow \infty} \omega_{s2w; \tilde k,k}(G_{Rad}) = \omega_{s2w}(G_{Rad}) = 1.
\eeq
We want to study the rate of convergence of this limit. To the best of our knowledge, it is not even known whether the limit is in general attained for finite $k,\, \tilde k$. We worry about lower bounds in $k,\, \tilde k$   when a given degree of approximation is achieved in \eqref{eq:LimS2wH}. More precisely,  we upper bound $ \omega_{s2w;\tilde k, k}(G_{Rad})$ in terms of $k,\, \tilde k$ and properties of the strategies considered. However, we postpone those results until Section \ref{Sec4}. Before that,  we need to provide here two reductions to the kind of strategies we consider in order  to prepare the ground for next section.

\subsection{Use of classical communication in cheating strategies}
First, we consider the role of classical communication between Alice and Bob. In our model, we regard this resource as free and, in fact, we built into the structure of the considered strategies the free communication of the classical information about $\ve$ (in the second round of local operations this parameter was considered as public). This is justified by the fact that our interest is in bounding the \emph{quantum} resources used for attacking $G_{Rad}$, which are assumed to be  much more expensive than classical communication. However, there is a potential problem with  this approach. That  is the possibility of the players using further classical communication apart from that of $\ve$ -- \emph{extra classical communication} from now on. In our model, this extra classical communication would be included in the definition of the channels $\A$ and $\BB_\ve$. In the $\mathfrak{S}_{s2w,\tilde k , k}$ scenario, this  would affect the dimension $\tilde k $ being  no longer a reliable witness for the quantum resources spent by a given strategy: $\tilde k$ would also include the dimension  of the extra classical messages shared by Alice and Bob. Nonetheless, we show that the amount of \emph{useful} extra classical communication in our setting is bounded by the initial dimension of the quantum system manipulated by the players, that is, by $k$ and $n$. The following lemma lets us control the contribution of the classical part of players action to $\tilde k$. 

\begin{lemma}\label{Lemma_ClassCom}
	The optimization over $\Ss \in \mathfrak{S}_{s2w,\tilde k , k} $ in \eqref{eq:ws2w} can be restricted to strategies using extra classical communication of local dimension $  \tilde k_{cl} \le n^4 k^2$.
	\end{lemma}


\begin{proof}
	The result follows from convexity taking into account that the extreme points of the set of instruments acting on a given Hilbert space of dimension $d$ has at most $d^2$ outcomes. See, for instance, \cite[Rmk. 7.9., p.158]{BuschBook} (also \cite[Corollary 1.36]{Kubicki_thesis}).
	
	 Consider an arbitrary strategy $\Ss = \lbrace  \tilde \A_\ve, \tilde \BB_\ve,\A,\BB_\ve, \varphi\rbrace \in  \mathfrak{S}_{s2w;m , k} $ using extra classical communication of local dimension $m_{cl}$. The dimensions $m$, $m_{cl}$ are free parameters that will  be fixed at the end of the proof. Therefore, we can further specify these classical messages in the structure of the channels $\A$ and $\BB_\ve$:
	$$
\A(\, \cdot \,) = \sum_{c_a=1}^{m_{cl}} \A^{c_a}(\, \cdot \,) \otimes |c_a\xa c_a|\quad : \quad \A^{c_a}\in\mathrm{CP}(\ell_2^{n^2 k},\ell_2^{m/m_{cl}}) \text{ for any }c_a,
$$
$$
\BB_\ve(\, \cdot \,) = \sum_{c_b=1}^{m_{cl}} \BB_\ve^{c_b}(\, \cdot \,) \otimes |c_b\xa c_b|\quad : \quad \BB^{c_b}_\ve\in\mathrm{CP}(\ell_2^{k},\ell_2^{m/m_{cl}}) \text{ for any }c_b.
$$
 These expressions are nothing else than the description of some instruments in $\mathrm{Ins}(\ell_2^k, \ell_2^{m/m_{cl}})$ ($\mathrm{Ins}(\ell_2^{n^2 k }, \ell_2^{m/m_{cl}})$ in the first case) with $m_{cl}$ outcomes each.    As we said at the beginning of the proof, the extreme points  of $\mathrm{Ins}(\ell_2^k, \ell_2^{m/m_{cl}})$ consist of instruments with at most $k^2$ outcomes ($n^4 k^2$ in the first case). Therefore, we can rewrite the channels $\A$, $\BB_\ve$ as a convex combination of such extreme points:
  $$
 \A(\, \cdot \,) = \sum_s \alpha_s  \A_s(\, \cdot \,),
 $$
 $$
 \BB_\ve(\, \cdot \,) = \sum_s \beta_{\ve,s}  \BB_{\ve,s} (\, \cdot \,),
 $$
 where, for each $s, \, \ve$:
 \begin{itemize}
 	\item   $0\le \alpha_s,\, \beta_{\ve,s} \le 1$ : $\sum_s \alpha_s = 1 = \sum_s \beta_{\ve,s}$;
 	\item  $\A_s \in \mathrm{Ins}(\ell_2^{n^4k^2}, \ell_2^{m/m_{cl}}) $, $\BB_{\ve;s} \in \mathrm{Ins}(\ell_2^k, \ell_2^{m/m_{cl}})$ with at most $n^4 k^2$ and $k^2$ outcomes, respectively. For simplicity we just fix $\tilde k_{cl}$ bounded by the largest of these bounds, $ \tilde k_{cl} \le n^4 k^2$.
 \end{itemize}

 Denote  $\Ss_{s,s'}$ the strategy specified by elements $\lbrace \tilde \A_\ve, \tilde \BB_\ve,\A_s,\BB_{\ve,s'}, \varphi \rbrace_\ve$ and $\Ss_{\ve;s,s'}(\, \cdot \,)$ the corresponding channels, defined by the generic prescription \eqref{StratS2w}. Notice that $\Ss_\ve (\, \cdot \,)= \sum_{s,s'} \,  \alpha_s \beta_{s'} \, \Ss_{\ve;s,s'}(\, \cdot \,) $.
   Now,  focus on the value achieved in $G_{Rad}$. It turns out that $\omega(G_{Rad};\Ss)$ is linear in $\Ss$, fact that allows us to write:
 $$
 \omega(G_{Rad};\Ss) = \sum_s \alpha_s \, \mathbb{E}_\ve \, \sum_{s'} \, \beta_{\ve,s}\,  \omega(G_{Rad}; \lbrace\Ss_{\ve;s,s'}\rbrace_\ve) 
 \le
  \max_{s} \left\lbrace  \mathbb{E}_\ve  \max_{s'} \left\lbrace  \omega(G_{Rad}; \lbrace\Ss_{\ve;s,s'}\rbrace_\ve)     \right\rbrace   \right\rbrace.
 $$
 Denoting $s^*$, ${s'_\ve}^*$ the indexes at which  the maxima above are attained, the strategy $\lbrace \tilde \A_\ve, \tilde \BB_\ve,$ $\A_{s^*},\BB_{\ve,{s'_\ve}^*}, \varphi \rbrace_\ve$, that uses extra classical communication of local dimension at most $\tilde k_{cl} \le n^4 k^2$, can be now regarded as an element in $\mathfrak{S}_{s2w;\tilde k , k}$ with    $\tilde k =  m \tilde k_{cl}/m_{cl}$. This proves the claim.
 
\end{proof}

\subsection{Pure strategies}
The second reduction consists in purifying arbitrary strategies. We start fixing some notation. We say that a strategy $\Ss = \lbrace  \tilde \A_\ve, \tilde \BB_\ve,\A,\BB_\ve, \varphi\rbrace_\ve \in \mathfrak{S}_{s2w}$ is \emph{pure} if the channels $\tilde \A_\ve, \tilde \BB_\ve,\A,\BB_\ve$ can be written as:
\begin{align} 
 \A (\, \cdot \, ) &=  V (\, \cdot \, ) V^\dagger,   &\BB_\ve (\, \cdot \, ) &= W_\ve (\, \cdot \, ) W_\ve^\dagger, 
 \\  
 \tilde \A_\ve (\, \cdot \, ) &=  \Tr_{anc_a} \,  \tilde V_\ve (\, \cdot \, )\tilde V_\ve^\dagger, &\tilde \BB_\ve (\, \cdot \, ) &= \Tr_{anc_b} \, \tilde W_\ve (\, \cdot \, )\tilde W_\ve^\dagger,
\end{align}
for some contractive operators
$$
V:\HH_A \otimes \HH_B \otimes \HH_{E_a} \longrightarrow \HH_{A  \shortrightarrow A }\otimes \HH_{A   \shortrightarrow B },
 \qquad W_\ve:  \HH_{E_b} \longrightarrow \HH_{B  \shortrightarrow B }\otimes \HH_{B   \shortrightarrow A }, 
$$
$$
\tilde V_\ve:  \HH_{A  \shortrightarrow A }\otimes \HH_{B   \shortrightarrow A } \longrightarrow \HH_A \otimes \HH_{anc_a},
 \qquad \tilde W_\ve:  \HH_{B  \shortrightarrow B }\otimes \HH_{A   \shortrightarrow B } \longrightarrow \HH_B \otimes \HH_{anc_b},
$$
where $\HH_{anc_a}$, $\HH_{anc_b}$ are arbitrary ancillary Hilbert spaces.
In the restricted scenario $\mathfrak{S}_{s2w;\tilde k , k}$, these operators are of the form:
\beq\label{pureStrategy}
V:\ell_2^{n^2 k} \longrightarrow \ell_2^{\tilde k}, \qquad W_\ve:\ell_2^{k} \longrightarrow \ell_2^{\tilde k}, \qquad \tilde V_\ve,\,\tilde W_\ve: \ell_2^{\tilde k} \longrightarrow \ell_2^{n r},
\eeq
 where $r$ is some (arbitrary) natural number. For convenience, we identify pure strategies with families of such \emph{pure} objects, setting the notation $\Ss^\UU = \lbrace \tilde V_\ve, \tilde W_\ve, V,W_\ve , |\varphi\ra \rbrace_\ve$.

 We further denote $\mathfrak{S}^\UU_{s2w}$ the subset of pure strategies in the s2w scenario and $\mathfrak{S}^\UU_{s2w;\tilde k, k}$ the corresponding subset in the model with limited dimension.  Due to Stinespring dilation theorem \cite{Stinespring55}, it turns out that $\mathfrak{S}^\UU_{s2w} = \mathfrak{S}_{s2w} $. However, when we restrict the dimension of the considered strategies, the situation is a bit subtler and the Stinespring dilation of the channels involved affects the relevant dimensions defining the models $\mathfrak{S}_{s2w;\tilde k, k}$ and $\mathfrak{S}^\UU_{s2w;\tilde k, k}$. This is taken care of by the following lemma:
\begin{lemma}\label{Lemma_Purifications}
	Any strategy $\Ss \in \mathfrak{S}_{s2w;\tilde k, k}$ can be regarded as a pure strategy $\Ss^\UU \in \mathfrak{S}^\UU_{s2w;\tilde k', k}$  where
$
\tilde k' \le n^2 k \tilde k^2.
$
That is, the chain of containments $\mathfrak{S}^\UU_{s2w;\tilde k, k} \subseteq \mathfrak{S}_{s2w;\tilde k, k} \subseteq \mathfrak{S}^\UU_{s2w;\tilde k', k}$  holds.
\end{lemma}

\begin{proof}
	Set a strategy $\Ss = \lbrace  \tilde \A_\ve, \tilde \BB_\ve,\A,\BB_\ve, \varphi\rbrace $ in $ \mathfrak{S}_{s2w;\tilde k, k}$.
	
	We are going to consider Stinespring dilations to purify the corresponding channels 	\begin{equation}\label{StratS2w2}
	 \Ss_\ve (\, \cdot \,  ) = (\tilde{\mathcal{A}}_\ve \otimes \tilde \BB_\ve)\circ(\mathcal{A} \otimes \BB_\ve) (\, \cdot \,  \otimes \varphi). 
	\end{equation}
	We start with
$$
\tilde{\mathcal{A}}_\ve ,\,  \tilde \BB_\ve \in \CPTP(  \ell_2^{\tilde k}, \, \ell_2^{n}  ).
$$

These channels can be lifted (due to a Stinespring dilation) to be of the form:
$$
\tilde{\A}_\ve (\, \cdot \, ) = \Tr_{\widetilde{anc}} \tilde V_\ve (\, \cdot \, ) \tilde V_\ve^\dagger ,
\qquad 
\tilde{\BB}_\ve (\, \cdot \, ) = \Tr_{\widetilde{anc}} \tilde W_\ve (\, \cdot \, )\tilde W_\ve^\dagger ,
$$
where $\tilde V_\ve,\, \tilde W_\ve: \ell_2^{\tilde k} \longrightarrow \ell_2^n \otimes \HH_{\widetilde{anc}}$ are Stinespring isometries and $\dim(\HH_{\widetilde{anc}})$ can be upper bounded by $n \tilde k$.

Proceeding similarly with $\A \in \CPTP(\ell_2^{n^2 k}, \, \ell_2^{\tilde k})$ and $\BB_\ve \in \CPTP( \ell_2^{ k}, \, \ell_2^{\tilde k}  )$ we obtain:
$$
\A (\, \cdot \, ) = \Tr_{anc_1}  V (\, \cdot \, ) V^\dagger  ,
\qquad 
\BB_\ve (\, \cdot \, ) = \Tr_{anc_2}   W_\ve (\, \cdot \, ) W_\ve^\dagger   ,
$$
for Stinespring dilations $V : \ell_2^{n^2 k}  \longrightarrow  \ell_2^{\tilde k}\otimes \HH_{anc_1} $, $ W_\ve : \ell_2^k  \longrightarrow  \ell_2^{\tilde k}\otimes \HH_{anc_2} $ such that $\dim( \HH_{anc_1}) \le n^2 k \tilde k$, $\dim( \HH_{anc_2}) \le k \tilde k$.

With all that, and denoting $\HH_{anc_a} \equiv \HH_{anc_1} \otimes \HH_{\widetilde{anc}}$, $\HH_{anc_b} \equiv \HH_{anc_2} \otimes \HH_{\widetilde{anc}}$, we define the channels
\begin{align*}
\tilde \A_\ve^\UU (\, \cdot \, ) &:= \Tr_{anc_a} ( \tilde V_\ve \otimes \id_{anc_1} ) \, (\, \cdot \, )  \, ( \tilde V_\ve^\dagger \otimes \id_{anc_1}) ,
\\[0.2em]
\tilde \BB_\ve^\UU (\, \cdot \, ) &:= \Tr_{anc_b} ( \tilde W_\ve \otimes \id_{anc_2} )\, (\, \cdot \, ) \, ( \tilde W_\ve^\dagger \otimes \id_{anc_2} ),
\end{align*}
$$
  \A^\UU (\, \cdot \, ) :=  V_\ve  (\, \cdot \, )  V_\ve^\dagger , \qquad 
  \BB_\ve^\UU (\, \cdot \, ) :=  W_\ve  (\, \cdot \, )  W_\ve^\dagger .
  $$
 Then, we can rewrite \eqref{StratS2w2} as: 
\begin{equation*}
\Ss_\ve (\, \cdot \,   ) = (\tilde \A^\UU_\ve \otimes \tilde \BB^\UU_\ve ) \circ  ( \A^\UU \otimes \BB^\UU_\ve  ) 
(\, \cdot \,  \otimes |\varphi\xa \varphi| ) .
\end{equation*}

But clearly the strategy $\Ss^\UU:= \lbrace  \tilde \A_\ve^\UU, \tilde \BB_\ve^\UU,\A^\UU,\BB_\ve^\UU, \varphi\rbrace_\ve$ is pure, finishing the proof of the lemma.  A careful look at the definition of the channels defining $\Ss^\UU$ reveals that $\Ss^\UU  \in \mathfrak{S}^\UU_{s2w;\tilde k',k} $ with 
$
\tilde k' \le n^2 k \tilde k^2.
$

\end{proof}

With Lemmas \ref{Lemma_ClassCom} and \ref{Lemma_Purifications} at hand  we can focus now on the study of strategies in $\mathfrak{S}^\UU_{s2w;\tilde k',k} $.  Given a general strategy $\Ss  \in \mathfrak{S}_{s2w;\tilde k,k} $,  Lemma \ref{Lemma_ClassCom} guarantees that $\Ss$ can be taken such that the dimension of the classical resources used is upper bounded by
\begin{equation}\label{Corresponde_k-k'2}
\tilde k_{cl} \le n^4 k^2.
\end{equation}

Then, Lemma \ref{Lemma_Purifications} allows us to relate  $\Ss$ with a pure strategy  $ \Ss^\UU \in \mathfrak{S}^\UU_{s2w;\tilde k',k}$ such that
\begin{equation}\label{Corresponde_k-k'1}
\tilde k' \le n^2 k \tilde k^2.
\end{equation}

Accordingly, in the rest of this manuscript we work in the model $\mathfrak{S}^\UU_{s2w;\tilde k',k}$  redirecting the reader to \eqref{Corresponde_k-k'1} and \eqref{Corresponde_k-k'2} for the relation  with the  resources used by more general strategies. However, notice that these correspondences are at most polynomial in  $n$, $k$ and $\tilde k$ and, in fact, will only introduce corrections by constant factors in the bounds we state later on. In this sense, the precise exponents in \eqref{Corresponde_k-k'1}, \eqref{Corresponde_k-k'2} are irrelevant. This will become clearer in the next section.  In order to obtain a cleaner notation, from now on we will use $\tilde k$ to refer to the same as $\tilde k'$ above.

For convenience, we finish this section recalling the expression of $\omega(G_{Rad}; \Ss^\UU)$, Equation \eqref{DefvalueStrat_MROQG_GRad2}, particularized for pure strategies $\Ss^\UU =\lbrace \tilde V_\ve , \tilde W_\ve, V, W_\ve,\varphi \rbrace$: 
 		\begin{align}\label{DefvalueStrat_MROQG_GRad3}
\omega (G_{Rad};\Ss^\UU )  &=   \mathbb{E}_{\ve} \ \Tr\left[\: |\psi_{\ve} \xa\psi_{\ve} |\,   (\id_{C}\otimes \mathcal{S}^\UU_{\ve})\: \big(|\psi \xa \psi| \big) \: \right],
\end{align}
where now:
$$
\Ss_\ve^\UU (\, \cdot \,   ) =\Tr_{\HH_{anc_a \otimes anc_b}}\, \left [ (\tilde V_\ve \otimes \tilde W_\ve) \, (V \otimes W_\ve ) \,
(\, \cdot \,  \otimes |\varphi\xa \varphi| ) \,
(V^\dagger \otimes W_\ve^\dagger  ) \,  (\tilde V_\ve^\dagger \otimes \tilde W_\ve^\dagger )\right].
$$

Notice that for strategies in the more specific model $ \mathfrak{S}_{s2w;\tilde k',k}^\UU $, the operators $\tilde V_\ve , \tilde W_\ve, V, W_\ve$ are specified as in \eqref{pureStrategy} and, therefore, $\HH_{anc_a}$ and $\HH_{anc_b}$ in this case are identified with $\ell_2^{ r}$ for some $r\in \mathbb{N}$. Finally, we provide an alternative expression for \eqref{DefvalueStrat_MROQG_GRad3} that establishes a first connection with normed spaces:
\begin{proposition}\label{DefvalueStrat_MROQG_GRad4}
	For any pure strategy $\Ss^\UU =\lbrace \tilde V_\ve , \tilde W_\ve, V, W_\ve,\varphi \rbrace$:
		\begin{align*}
		\omega (G_{Rad};\Ss^\UU )  &=   \, \mathbb{E}_\ve \,     \Big\| \frac{1}{n^2}  \sum_{i,j} \ve_{ij}   (\la i|\tilde V_\ve \otimes \la j|\tilde W_\ve ) \, (V|ij\ra \otimes  W_\ve ) \,   |\varphi  \ra  \Big\|^2_{\HH_{anc_a}\otimes \HH_{anc_b}}.
	\end{align*}
\end{proposition}

	Before showing the easy proof of this proposition, let us clarify the notation used above. By  $V|ij\ra $ we mean the operator $V \in M_{\tilde k',n^2 k }$ with its indices corresponding to $\ell_2^{n^2}$ contracted with the vector $|ij\ra $. That is, if we expand $V$ on its coordinates, $V = \sum_{ k,l=1}^n  \sum_{m=1}^{k} \sum_{p}^{ \tilde k'}$ $  V_{p,klm}  |p \xa kl m|  $, and then $V|ij\ra = \sum_{m=1}^{k} \sum_{p}^{ \tilde k' }  V_{p,i j m}  |p \xa m| \in M_{\tilde k', k}$.  Similarly with $\la i |\tl{V}_{\ve}$ and $\la j|\tl{ W}_{\ve}$.
	
	\begin{proof} 
		The proof is completely elementary and follows the next lines: 
	
	In the first place, we notice that for any vectors $|\xi\ra\in \HH$, $|\eta\ra\in\HH' $ and any operator $U \in \BB(\HH',\HH\otimes \KK)$
	\begin{align*}
	\Tr \left[ |\xi\xa \xi |  \,\Tr_\KK U \, |\eta\xa \eta|\, U^\dagger    \right]
	 &= 
	 \Tr \left[( |\xi\xa \xi |  \otimes \id_\KK) \, U \, |\eta\xa \eta| \, U^\dagger    \right]
	\\
	& = \la \eta| U^\dagger ( |\xi \ra \otimes \id_\KK)  \  (\la \xi |  \otimes \id_\KK) \, U \, |\eta\ra    
	\\
	&= \big\|  (\la \xi |  \otimes \id_\KK) \, U \, |\eta\ra \big\|^2_\KK.
	\end{align*}
	Applying this elementary identity to $|\psi_\ve \ra\in \HH_A \otimes \HH_B \otimes \HH_C \equiv \HH $, $ |\psi\ra \otimes |\varphi\ra \in \HH_A \otimes \HH_B \otimes \HH_C \otimes \HH_E\equiv \HH'$ and the operator $\id_C\otimes (\tilde V_\ve \otimes \tilde W_\ve)\,(V \otimes W_\ve ) \in \BB(\HH', \HH \otimes \HH_{anc_a}\otimes \HH_{anc_b})$ we have that, for each $\ve\in \calQ_{n^2}$:
	\begin{align*}
	&\Tr\left[\: |\psi_{\ve} \xa\psi_{\ve} |\,   (\id_{C}\otimes \mathcal{S}^\UU_{\ve})\: \big(|\psi \xa \psi| \big) \: \right]
	\\
	&\qquad = \, \mathbb{E}_\ve \,  \left\| ( \la \psi_\ve |\otimes \id_{anc_a \otimes anc_b}) \left( \id_C\otimes(\tilde V_\ve \otimes \tilde W_\ve)\, (V \otimes W_\ve) \right)\, ( |\psi\ra \otimes |\varphi  \ra ) \right\|^2_{\HH_{anc_a}\otimes \HH_{anc_b}}.
		\end{align*}
		 The claim in the Proposition is obtained from the last line above just recalling the definitions $|\psi_\ve \ra = \frac{1}{n} \sum_{i,j=1}^n \ve_{ij} |ij\ra_{AB} \otimes |ij\ra_C$ and $|\psi \ra = \frac{1}{n} \sum_{i,j=1}^n  |ij\ra_{AB}  \otimes |ij\ra_C $.
		\end{proof}

 \section{Bounds for ``smooth'' strategies. Theorem \ref{mainThm}}\label{Sec4}

This section is devoted to the proof of  Theorem \ref{mainThm}, which provides lower bounds on resources needed to break $G_{Rad}$ by strategies  characterized by  regularity measures based on  parameter $\sigma$ defined in Section \ref{Sec2.4.3}. When we refer here to a cheating strategy for $G_{Rad}$, unless the opposite is explicitly specified, we mean a pure strategy $\Ss^\UU =\lbrace   \tilde V_\ve,\, \tilde W_\ve, \,  V,\, W_\ve,\, |\varphi \ra \rbrace_\ve \in \mathfrak{S}^\UU_{s2w;\tilde k, k}$.

As explained in the introduction, the main idea leading to Theorem \ref{mainThm} is the understanding  of cheating strategies for $G_{Rad}$ as assignments on the hypercube $\calQ_{n^2}$, i.e.,  vector-valued functions $\Phi: \calQ_{n^2} \rightarrow X$ where $X$ is a suitable Banach space. Given a strategy $\Ss^\UU$, the corresponding assignment $\Phi$ must be related with the value  $\omega(G_{Rad};\Ss^\UU) $. Ideally, we hope to bound $\omega(G_{Rad};\Ss^\UU) $ with the expected value of the norm of $\Phi$, quantity for which we can use Corollary \ref{Cor1} to obtain upper bounds. Proposition \ref{DefvalueStrat_MROQG_GRad4} gives us a first hint on how to construct  $\Phi$. Given $\Ss^\UU =\lbrace   \tilde V_\ve,\, \tilde W_\ve, \,  V,\, W_\ve,\, |\varphi \ra \rbrace_\ve$, consider the map:
\begin{equation}\label{Def_Phi0}
\begin{array}{rccc}
\Phi : &  \calQ_{n^2} & \longrightarrow &  \ell_2^{r^2}
\\[1em]
& \ve & \mapsto  &  \Phi^{i} (\ve)  = \frac{1}{n^2}  \, \sum_{ij} \,   \, \varepsilon_{ij} \, (\la i |\tl{V}_{\ve} \otimes \la j|\tl{ W}_{\ve})\, (V|ij\ra  \otimes W_{ {\ve}}) |\varphi\ra
\end{array} ,
\end{equation}
where $ r $ is determined by the strategy, recall \eqref{pureStrategy}.

 Proposition \ref{DefvalueStrat_MROQG_GRad4} can be now read as: 
\beq\label{Equivalence_Phi-omega}
\omega(G_{Rad}; \Ss^\UU) = \mathbb{E}_\ve \| \Phi (\ve ) \|_{\ell_2^{ r^2}}^2   , 
\eeq
so we are on a good track.  It is easy to check that, by construction, $  \| \Phi (\ve ) \|_{\ell_2^{ r^2 }} \le 1$ for any $\varepsilon \in \calQ^{n^2}$ and therefore the trivial bound $\mathbb{E}_\ve \| \Phi (\ve ) \|_{\ell_2^{ r^2 }}^2   \le \mathbb{E}_\ve \| \Phi (\ve ) \|_{\ell_2^{ r^2}}$  holds. With this and  Corollary \ref{Cor1}, we can obtain -- recall Definition \ref{Def_sigma} for $\sigma_{\Phi}$:
\begin{equation*}
\omega(G_{Rad}; \Ss^{\UU}) 
\le  \big \|\mathbb{E}_\ve  \Phi (\ve) \big\|_{\ell_2^{ r^2} } + C \  \sigma_{\Phi} \mathrm{T}^{(n^2)}_2 (\ell_2^{ r^2} ).
\end{equation*}
Furthermore, $\mathrm{T}^{(n^2)}_2 (\ell_2^{ r^2} ) =1$ since, more generally, $\mathrm{T}_2 (\ell_2^{ r^2} ) = 1 $.
 
 The main problem with this approach is that  the quantity $\big \|\mathbb{E}_\ve  \Phi (\ve) \big\|_{\ell_2^{ r^2} }$ might be of the same order as  $\omega(G_{Rad}; \Ss^{\UU})$, making the previous bound trivial. In fact, for any given $ \Phi$ there exists a rather trivial modification of it that \emph{does not affect any dimension involved} but provides a function $\Phi'$  for which $\big \|\mathbb{E}_\ve  \Phi' (\ve) \big\|_{\ell_2^{  r^2} } = \mathbb{E}_\ve \big \|  \Phi' (\ve) \big\|_{\ell_2^{ r^2} } =  \omega  (G_{Rad}; \Ss^{\UU}).$ This modified version of $\Phi$ can be constructed as follows. For each $\ve \in \calQ_{n^2}$, consider a unitary  $R_\ve$ that rotates the vector $\Phi (\ve)$ into the direction of a reference unit vector $|0\ra$. $\Phi'$ is defined simply as $\Phi' (\ve):= R_\ve \Phi(\ve) $. Even when this adjustment is  completely artificial -- the unitaries $R_\ve$ do not correspond with anything implementable by the cheaters -- the approach presented so far is unable to detect such an artifact.  In part, this is due to the fact that the norm considered on the image of $\Phi$ does not encode any of the structure of the cheating action. We now look at alternative constructions for $\Phi$ that amend this issue.

 What we do next, is simplifying the image of the map $\Phi$ considering more involved choices for the output Banach space. This allows us   to  preserve an equivalence of the kind of \eqref{Equivalence_Phi-omega} while obtaining good upper bounds for $\big \|\mathbb{E}_\ve  \Phi (\ve) \big\|_{ X }$. 
 
Given a strategy $\Ss^\UU =\lbrace   \tilde V_\ve,\, \tilde W_\ve, \,  V,\, W_\ve,\, |\varphi \ra \rbrace_\ve$ we define the following two alternatives to $\Phi$:
$$
\begin{array}{rccc}
\Phi^{i} : &  \calQ_{n^2} & \longrightarrow &  M_{ r, k } \otimes_{min}  M_{r,\tilde k} 
\\[1em]
& \ve & \mapsto  &  \Phi^{i} (\ve)  = \frac{1}{n^2}  \, \sum_{ij} \,   \, \varepsilon_{ij} \, (\la i |\tl{V}_{\ve} \otimes \la j|\tl{ W}_{\ve})\, (V|ij\ra  \otimes \id_{\ell_2^{\tilde k}}),
\\[2em]
\Phi^{ii} : &  \calQ_{n^2} & \longrightarrow &  \s^{\tilde k, n } \otimes_{  (\ve, \pi)_{ \sfrac{1}{2} } } \s^{\tilde k, n}
\\[1em]
& \ve & \mapsto  &  \Phi^{ii} (\ve)  = \frac{1}{n^2}  \, \sum_{ij} \,   \, \varepsilon_{ij} \, \la i | \otimes \la j | \otimes (V|ij\ra  \otimes W_\ve ) \,  |\varphi \ra.
\end{array} 
$$
These are the central objects we study to obtain Theorem \ref{mainThm}.   {For $\Phi^i$, recall Section \ref{Sec.2.4.1} for the definition of the $\otimes_{min}$ norm, which in this particular case can be also understood as defined by the (completely) isometric equivalence $M_{ r, k }  { \otimes_{min}  M_{r,\tilde k} } = M_{r^2,k\tilde k}$. In the case of $\Phi^{ii}$ the norm on the output space was defined at the end of Section \ref{Sec2.4.3} as the interpolation space $( \s^{m, n } \otimes_{  \ve  } \s^{m, n},   \s^{m, n } \otimes_{ \pi } \s^{m, n})_{1/2}  $.}

Now we comment on  the idea behind the definitions of these maps: recall that a strategy $\Ss^\UU = \lbrace   \tilde V_\ve,\, \tilde W_\ve, \,  V,\, W_\ve,\, |\varphi \ra \rbrace_\ve$  consists of two rounds of local operations with a communication stage in between. Fixing the first round, that is  related to  $V,\, W_\ve$ and $ |\varphi \ra $, and understanding the optimization over any $\tilde V_\ve,\, \tilde W_\ve$ as computing a particular norm leads us to define $\Phi^{ii}$. When we fix $\tilde V_\ve,\, \tilde W_\ve$ and $V$ -- this last one is $	\ve$-independent --, optimizing then over any possible $(\id_{\ell_2^{k}} \otimes W_\ve)|\varphi \ra$, we obtain $\Phi^{i}$. 

Next we describe how these maps are related to $G_{Rad}$.

\begin{lemma}\label{Lemma_Main1}
	For any strategy $\Ss^\UU$,
	$$
	\omega ( G_{Rad}; \Ss^\UU) \le \, \mathbb{E}_\ve \, \big \| \Phi^{i(ii)} (\ve) \big \|_{X^{i(ii)}},
	$$
	where we have denoted $X^i =  {M_{r,k} \otimes_{min} M_{r,\tilde k}}$ and $X^{ii} = \s^{ {\tilde k}, n } \otimes_{  (\ve, \pi)_{ \sfrac{1}{2} } } \s^{ {\tilde k}, n}$. 
\end{lemma} 
\begin{remark}\label{Lemma_Main1_Rmk}
		For $\Phi^{ii}$, the previous statement can be strengthened to 
	$$
	\omega ( G_{Rad}; \Ss^\UU) \le \, \mathbb{E}_\ve \, \big \| \Phi^{ii} (\ve) \big \|_{\tilde X^{ii}},
	$$
	where  $\tilde X^{ii} = \s^{ {\tilde k}, n } \otimes_{ \mathfrak{S}_2^{w-cb} } \s^{ {\tilde k}, n}$. Recall Definition \ref{Def_cbWeakSchatten} for this last norm.
\end{remark}
{
\begin{proof}
	The proof of both items in the lemma follows the same structure. We start with the bound regarding $\Phi^i$:
	
	Recalling Proposition \ref{DefvalueStrat_MROQG_GRad4}:
	\begin{align*}
	\omega(G_{Rad}; \Ss^{\UU}) 
 = \, \mathbb{E}_\ve \,     \Big\| \frac{1}{n^2}  \sum_{i,j} \ve_{ij}   (\la i|\tilde V_\ve \otimes \la j|\tilde W_\ve ) \, (V|ij\ra \otimes  W_\ve ) \,   |\varphi  \ra  \Big\|^2_{\ell_2^{ r^2}}.
	\end{align*}
	
	We bound this quantity as follows:
	
	\begin{align*}
	\omega(G_{Rad}; \Ss^{\UU}) 
	& = \, \mathbb{E}_\ve \,     \left\| \frac{1}{n^2}  \sum_{i,j} \ve_{ij}   (\la i|\tilde V_\ve \otimes \la j|\tilde W_\ve ) \, (V|ij\ra \otimes  \id_{\ell_2^{\tilde k}} ) \,   (\id_{\ell_2^{k}} \otimes W_\ve)\, |\varphi  \ra   \right\|^2_{\ell_2^{ r^2}}
	\\
	& \le \, \mathbb{E}_\ve \, \sup_{ |\varphi \ra \in B_{\ell_2^{k \tilde k}} }    \left\| \frac{1}{n^2}  \sum_{i,j} \ve_{ij}   (\la i|\tilde V_\ve \otimes \la j|\tilde W_\ve ) \, (V|ij\ra \otimes  \id_{\ell_2^{\tilde k^2}} ) \,   |\varphi  \ra  \right\|^2_{\ell_2^{r^2}}
	\\
	& =\, \mathbb{E}_\ve \, \sup_{ |\varphi \ra \in B_{\ell_2^{k \tilde k }} }     \left\|   \Phi^{i} (\ve)  (|\varphi\ra) \right\|^2_{\ell_2^{r^2}} = \, \mathbb{E}_\ve \,  \left\|   \Phi^{i} (\ve)  \right\|^2_{M_{r^2 , k \tilde k}} 
	\\
	&\le\, \mathbb{E}_\ve \,  \left\|   \Phi^{i} (\ve)  \right\|_{M_{ r^2 , k \tilde k }} \equiv  \, \mathbb{E}_\ve \,  \left\|   \Phi^{i} (\ve)  \right\|_{X^i} 
	.
	\end{align*}
	 {The inequality in the last line holds since $ \left\|   \Phi^{i} (\ve)  \right\|_{M_{ r^2 , k \tilde k }}\le 1 $ for any $\ve\in \calQ_{n^2}$. This can be  checked by direct calculation or, alternatively, as a consequence of  Remark \ref{Rmk_normalizationInt1/2}.}

	For $\Phi^{ii}$, we prove  {the} stronger result  {stated in Remark \ref{Lemma_Main1_Rmk}.} That is, considering  the map $\Phi^{ii}$ taking values on the space $\tilde X^{ii} =\s^{\tilde k, n } \otimes_{\mathfrak{S}^{w-cb}_2} \s^{\tilde k, n}$, we show that:
		\beq\label{Bound2_Lemma_Main1}
	\omega ( G_{Rad}; \Ss^\UU) \le \, \mathbb{E}_\ve \, \big \| \Phi^{ii} (\ve) \big \|_{\tilde X^{ii}}.
	\eeq
	
	Since the norm in $\tilde X^{ii}$ is smaller than in $X^{ii}$, recall Proposition \ref{Prop_RelNorms}, the statement of the lemma is also true. Following the   proof of the first item, we start bounding:
	
	\begin{align*}
	\omega(G_{Rad}; \Ss^{\UU}) 
	& = \, \mathbb{E}_\ve \,     \left\| \frac{1}{n^2}  \sum_{i,j} \ve_{ij}   (\la i|\tilde V_\ve \otimes \la j|\tilde W_\ve ) \, (V|ij\ra \otimes  W_\ve ) \,   |\varphi  \ra  \right\|^2_{\ell_2^{r^2}}
	\\
	& \le \, \mathbb{E}_\ve \,   \sup_{ \tilde V, \tilde W \in B_{M_{n r,\tilde k}}}   \left\| \frac{1}{n^2}  \sum_{i,j} \ve_{ij}   (\la i|\tilde V\otimes \la j|\tilde W ) \, (V|ij\ra \otimes  W_\ve ) \,  |\varphi  \ra  \right\|^2_{\ell_2^{r^2}}
\\
 & = \, \mathbb{E}_\ve \, \sup_{ \begin{subarray}{c} \tilde V, \tilde W \in B_{M_{n r,\tilde k}}   \end{subarray}}
	\left\| (  \tilde V  \otimes  \tilde W ) \left( \frac{1}{n^2}  \sum_{i,j} \ve_{ij}\,  (\la i |  \otimes \la j | \otimes V |ij\ra \otimes W_\ve) \, |\varphi \ra \right) \right\|^2_{\ell_2^{r^2}} 
	\\
	& = \, \mathbb{E}_\ve \, \sup_{ \begin{subarray}{c} \tilde V, \tilde W \in B_{M_{n r,\tilde k}}  \end{subarray}}
	\left\|   (  \tilde V  \otimes  \tilde W )    \left(  \Phi^{ii} (\ve)   \right) \right\|^2_{\ell_2^{r^2}}
		\\& \hspace{-0.65cm} \!  \stackrel{\text{(Lemma \ref{lemma_CharacSigma} )}}{\le} \, \mathbb{E}_\ve \,  \left\|   \Phi^{ii} (\ve)  \right\|^2_{\s^{\tilde k, n } \otimes_{\mathfrak{S}^{w-cb}_2} \s^{\tilde k, n}}
	\\
	& \le \, \mathbb{E}_\ve \,  \left\|   \Phi^{ii} (\ve)  \right\|_{\s^{\tilde k, n } \otimes_{\mathfrak{S}^{w-cb}_2} \s^{\tilde k, n}} 
	\equiv  \, \mathbb{E}_\ve \,  \left\|   \Phi^{ii} (\ve)  \right\|_{ \tilde X^{ii}}.
	\end{align*}
 {As before, the last inequality is true given that $\left\|   \Phi^{ii} (\ve)  \right\|_{\s^{\tilde k, n } \otimes_{\mathfrak{S}^{w-cb}_2} \s^{\tilde k, n}} \le 1$ for any $\ve \in \calQ^{n^2}$. Again, this fact can be shown by direct computation (see Remark \ref{Rmk_normalizationInt1/2} for an alternative proof).}
\end{proof}
}

The regularity of   {the} maps  {$\Phi^{i(ii)}$} can be characterized by parameters $  \sigma^{i}_{\Ss^\UU}:=  \sigma_{\Phi^{i}}  $ and $  \sigma^{ii}_{\Ss^\UU}:=  \sigma_{\Phi^{ii}}$ -- recall Definition \ref{Def_sigma}. More explicitly:
\beq \label{Sigma_i.}
\sigma^{i(ii)}_{\Ss^\UU} =  \log(n^2)\, \mathbb{E}_\ve\,  \left(\sum_{i,j} \| \partial_{ij}  \Phi^{i {(ii)}} (\ve)  \|_{ X^{i(ii)} }^2  \right)^{1/2}.
\eeq
In the case of an arbitrary (possibly \emph{non-pure}) strategy $\Ss $, we can assign parameters $ \sigma^{i}_\Ss$, $\sigma^{ii}_\Ss$ with the simple prescription:
$$
\sigma^{i(ii)}_{\Ss } :=\inf_{\begin{subarray}{c}\Ss^\UU   \\  \text{purifying } \Ss \end{subarray}}  \sigma^{i(ii)}.
$$
 {From now on, we omit the subindex specifying the strategy, which is always determined by the context, and refer to these parameters as $\sigma^i$, $\sigma^{ii}$.}

The above expressions for $\sigma^{i}$, $\sigma^{ii}$ can be bounded by the easier expressions appearing in the introduction. See Appendix \ref{Appendix_1} for details. In Equation \eqref{Sigma_i.} the analytic nature of these parameters is clearer while the approximate expressions in Section \ref{Sec1} are closer to an operational interpretation of them.


With   definition \eqref{Sigma_i.} at hand and taking into account  Corollary \ref{Cor1}, we can obtain:
\begin{lemma}\label{Lemma_Main2}
		For any strategy $\Ss^\UU$,
		\begin{enumerate}[i.]
			\item $$
			\omega(G_{Rad}; \Ss^{\UU}) 
			\le  \big \|\mathbb{E}_\ve  \Phi^{i} (\ve) \big\|_{X^i} + C \  \sigma^{i} \ {\mathrm{T}^{(n^2)}_2}  \left( X^i \right),
			$$
			\item $$
			\omega(G_{Rad}; \Ss^{\UU}) 
			\le  \big \|\mathbb{E}_\ve  \Phi^{ii} (\ve) \big\|_{\tilde X^{ii}} + C \  \sigma^{ii} \ {\mathrm{T}^{(n^2)}_2}  \left( X^{ii} \right).
			$$
		\end{enumerate}
\end{lemma}

\begin{comment}
Notice the change of norms in the second item of the lemma.  This refinement is needed later on in order to obtain Proposition \ref{Prop_Main2} below. 
\end{comment}

 {
\begin{proof}
	
	The first item is a direct consequence of Corollary \ref{Cor1} applied to the bound  in Lemma \ref{Lemma_Main1}, i. 
	
	The second item proceeds similarly but with a small detour. Using now Pisier's inequality, Lemma \ref{lemmaPisier} (with $p=1$ and a trivial triangle inequality, as in the proof of Corollary \ref{Cor1}), in the stronger inequality \eqref{Bound2_Lemma_Main1} we obtain
	$$
	\mathbb{E}_\ve \,  \left\|   \Phi^{ii} (\ve)  \right\|_{ \tilde X^{ii} } \le  \,  \left\| \mathbb{E}_\ve   \Phi^{ii} (\ve)  \right\|_{\tilde X^{ii} }  +  C  \, \log(n) \, \mathbb{E}_{\ve,\tilde \ve} \left\|  \sum_{k,l=1}^n  \tilde \ve_{kl}  \partial_{kl} \Phi^{ii}(\ve)   \right\|_{ \tilde X^{ii} }.
	$$
	
	Now, according to Proposition \ref{Prop_RelNorms}, we can upper bound the last summand above changing the norm $\tilde X^{ii}$ by $X^{ii}$. Considering  that  {(recall again Comment \ref{Kahane})}
		$$
	\mathbb{E}_{\ve,\tilde \ve} \Big\|  \sum_{k,l=1}^n  \tilde \ve_{kl}  \partial_{kl} \Phi^{ii}(\ve)   \Big\|_{X^{ii}} 
	\lesssim \mathrm{T}_2^{(n^2)} (X^{ii})  \ \ \mathbb{E}_\ve \,  \Big (   \sum_{k,l=1}^n  \|  \partial_{kl} \Phi^{ii}(\ve)\|_{X^{ii}}^2   \Big)^{1/2},
	$$
	we have:
	$$
	\omega(G_{Rad};\Ss^\UU) \le    
	  \,  \left\| \mathbb{E}_\ve   \Phi^{ii} (\ve)  \right\|_{\tilde X^{ii} }  
	+  C  \, \log(n)  \, \mathrm{T}_2^{(n^2)} (X^{ii})  \ \ \mathbb{E}_\ve \,  \left (   \sum_{k,l=1}^n  \|  \partial_{kl} \Phi^{ii}(\ve)\|_{X^{ii}}^2   \right)^{1/2}.
	$$
	We obtain Lemma \ref{Lemma_Main2}, ii., identifying $\sigma^{ii}$ above.

\end{proof}

}

Lemma \ref{Lemma_Main2} allows us to somehow exchange the lack of control on the behaviour of a general strategy by the control of some properties of the Banach spaces involved. Bounding the quantities appearing there, we obtain our main result:

\begin{theorem}[Formal statement of Theorem \ref{mainThm}] \label{mainThm_2}
	Given an arbitrary (possibly non-pure)  strategy   $\Ss \in \mathfrak{S}_{s2w;\tilde k, k}$,
	\begin{enumerate}[I.]
		\item 
		$$
		\omega(G;\Ss)  \le  C_1 +   C_2 \ {\sigma^i} \, \log^{1/2}( k \tilde k ) +  O\left(\frac{1}{n^{1/2}} \right)  ;
		$$
		\item   
				$$
		\omega(G;\Ss)	  \le  \tilde C_1 + C_3 \ \tilde \sigma^{ii}  \, \log^{1/2} (n k \tilde k )   +  O \left(\frac{1}{n^{1/2}} + \frac{\log(n)  \log^{1/2}( n k \tilde k )}{n} \right)
		,
		$$
		where we have denoted  $\tilde \sigma^{ii} = n^{3/4} \log(n) \, \sigma^{ii}$.
	\end{enumerate}
	Above, $C_1,\, \tilde C_1 <1, \, C_2,\, C_3 $ are positive constants.
	
\end{theorem}

 {The following proposition precisely takes care of bounding part of the terms appearing in Lemma \ref{Lemma_Main2}, as a key step to prove the theorem.}

 {
\begin{proposition} \label{Prop_Main2}	For any pure strategy $\Ss^\UU \in \mathfrak{S}_{s2w; {\tilde k}, k}$:
	\begin{enumerate}[i.]
		\item $$
		\big \|\mathbb{E}_\ve  \Phi^{i} (\ve) \big\|_{X^i}
		\le  \frac{3}{4} + \frac{C}{\sqrt{n}}.
		$$
		
		\item $$
		\big \|\mathbb{E}_\ve  \Phi^{ii} (\ve) \big\|_{ \tilde X^{ii} }
		\le  \frac{\sqrt{3}}{2} + \frac{C}{\sqrt{n}} + C' \frac{\log (n) \log^{1/2}(k  {\tilde k})}{ n},
		$$
	\end{enumerate}
	where $C,\, C'$ are universal constants. 
\end{proposition}
}

 {The sequence leading to Theorem \ref{mainThm_2} is the following: Proposition \ref{Prop_Main2}.i $\Rightarrow$ Theorem \ref{mainThm_2}.I $\Rightarrow$ Proposition  \ref{Prop_Main2}.ii $\Rightarrow$  Theorem  \ref{mainThm_2}.II. To simplify the presentation, we will first write the proof of both statements of Theorem \ref{mainThm_2}, assuming  the corresponding statements of Proposition  \ref{Prop_Main2}.  Then we will prove Proposition \ref{Prop_Main2}, using Theorem \ref{mainThm_2}.I in the proof of statement (ii).}

\begin{proof}[ {Proof of Theorem \ref{mainThm_2}}]
	To obtain the statement of the theorem, as we already said, we start considering Lemma \ref{Lemma_Main2}. Then, we need to bound:
	\begin{enumerate}
		\item the type constants ${\mathrm{T}^{(n^2)}_2}(X^{i})$ and ${\mathrm{T}^{(n^2)}_2}(X^{ii})$. These bounds are already provided in Equations \eqref{Eq4_TypeConsts_1} and \eqref{Type_Int1/2}, respectively. We recall these bounds here for reader's convenience:
		$$
		{\mathrm{T}^{(n^2)}_2}(X^{i}) \le {\mathrm{T}_2}(X^{i})  \lesssim \log^{1/2}(k\tilde k), \qquad {\mathrm{T}^{(n^2)}_2}(X^{ii}) \lesssim  n^{3/4} \log(n) \log^{1/2} (n \tilde k) ;
		$$
		
		\item the terms $ \big \|\mathbb{E}_\ve  \Phi^{i} (\ve) \big\|_{X^i}$ and $\big \|\mathbb{E}_\ve  \Phi^{ii} (\ve) \big\|_{\tilde X^{ii}}$. These quantities are controlled by  Proposition \ref{Prop_Main2}. 
	\end{enumerate}
With this we obtain the stated bound in the case of pure strategies. Nonetheless, statements about pure strategies can be transformed into statements about general strategies taking into account the relation \eqref{Corresponde_k-k'1}. As we said at the end of Section \ref{Sec3}, this relation is polynomial in the parameters involved and therefore, the change from pure to general strategies only induces corrections by constant factors that we absorbed in the constants $C_2,\, C_3$ present in the statement. Similar considerations deal with the amount of classical communication included in $\tilde k$, in this case one has to recall Equation \eqref{Corresponde_k-k'2}. See Appendix \ref{Appendix_MainThm} for further details.
\end{proof}

 {
\begin{proof}[Proof of Proposition \ref{Prop_Main2}, i.]
	The norm in the L.H.S. of Proposition \ref{Prop_Main2}, i., is attained at unit vectors $|\varphi\ra \in \ell_2^{k \tilde k},\, |\xi\ra \in \ell_2^{ r^2}$ (independent of $\ve$)\footnote{ Recall the isometric identity $M_{r,k} \otimes_{min} M_{r,\tilde k} = M_{r^2,k\tilde k} $.}:
	$$
	\big \|\mathbb{E}_\ve  \Phi^{i} (\ve) \big\|_{M_{r ^2,k \tilde k}} = \big| \mathbb{E}_\ve\  \la \xi|\,  \Phi^{i} (\ve)\, |\varphi\ra \big| \le    \mathbb{E}_\ve\ \big|  \la \xi|\,  \Phi^{i} (\ve)\,  |\varphi\ra  \big| .
	$$
	Expanding this expression we have:
	$$
	\| \mathbb{E}_\ve \Phi^{i} (\ve) \|_{M_{ r^2,k \tilde k}} \le
	\mathbb{E}_\ve\ \big| \, \frac{1}{n^2} \sum_{i j} \ve_{ij}      \la \xi| \, \big( \la i |\tilde  V_\ve\otimes  \la j |\tilde W_\ve \big) \, \big( V \ |ij\ra  \otimes \id_{\ell_2^{\tilde k}} \big) |\varphi\ra\, \big|  = \mathbb{E}_\ve \ \big| \la \overline \xi_\ve \, |\, \overline \varphi \ra \big|,
	$$
	where we have defined the unit vectors:
	$$
	\la \overline \xi_\ve | := \frac{1}{n} \sum_{i,j} \ve_{ij}   \la ij|_{C} \otimes \la \xi| \, \big( \la i |\tilde  V_\ve\otimes  \la j |\tilde W_\ve \big) ,
	$$
	$$
	|\overline \varphi \ra := \frac{1}{n} \sum_{i,j}    |ij\ra_{C} \otimes \big( V \ |ij\ra  \otimes \id_{\ell_2^{\tilde k'}} \big) |\varphi\ra.
	$$
	
	Now, notice that there exists at least one $\ve^*$ such that $ |\la \overline \xi_{\ve^*} \, |\, \overline \varphi \ra | \ge \| \mathbb{E}_\ve \Phi^{i} (\ve) \|_{M_{ r^2,k \tilde k}}$. Consider this $\ve^*$ to rewrite $|\overline \varphi \ra = |\overline \xi_{\ve^*} \ra  + 		(|\overline \varphi \ra - |\overline \xi_{\ve^*} \ra )$. An application of Cauchy-Schwarz inequality gives us the following:
	\begin{align}\label{eq3}
	\| \mathbb{E}_\ve \Phi^{i} (\ve) \|_{M_{ r^2,k \tilde k}}  &\le \mathbb{E}_\ve\ \big| \la \overline \xi_\ve |\overline \varphi \ra \big|
	\le \mathbb{E}_\ve\ \big| \la \overline \xi_\ve |\overline \xi_{\ve^*} \ra \big|  + \big | \la \overline \varphi-\overline \xi_{\ve^*} |\overline \varphi-\overline \xi_{\ve^*} \ra\big|^{1/2}.
	\end{align}
	Now we bound both summands in the R.H.S. of the previous expression separately:
	\begin{itemize}
		\item For the second:
		\begin{align}
		\label{eq4}
		\big | \la \overline \varphi-\overline \xi_{\ve^*} |\overline \varphi-\overline \xi_{\ve^*} \ra\big|^{1/2} &\le \left( 2 ( 1-  \la \overline \xi_{\ve^*}|\overline \varphi\ra  ) \right)^{1/2} \le \left( 2(1- \| \mathbb{E}_\ve\ \Phi (\ve) \|_{M_{r^2,k \tilde k}}) \right)^{1/2} 
		\nn\\
		&\le \frac{7}{4} - \frac{4}{3}  \| \mathbb{E}_\ve\ \Phi^{i} (\ve) \|_{M_{r^2,k \tilde k}},\nn
		\end{align} 
	 {where the last inequality is based on an elementary linear approximation of  $\sqrt{2(1-x)}$ (simply by the line tangent to the function in a suitably chosen point).}
		
		\item For the first one, we will find that:
		$$
				\mathbb{E}_\ve\ \big| \la \overline \xi_\ve |\overline \xi_{\ve^*} \ra \big| = O\big(\frac{1}{\sqrt{n}}\big).
		$$
		In order to show this bound, we start observing that:
		\begin{align*}
		\mathbb{E}_\ve\ \big| \la \overline \xi_\ve |\overline \xi_{\ve^*} \ra \big|
		&= \mathbb{E}_\ve \  \Big|  \frac{1}{n^2} \sum_{ij} \ve_{ij} \ve^*_{ij}    \la \xi |\,   \big( \la i | \tilde V_\ve \, \tilde V_{\ve^*}^\dagger |i\ra \otimes  \la j |  \tilde W_\ve \, \tilde W_{\ve^*}^\dagger|j\ra  \big) \, 	 |\xi\ra   \Big|  \\
		& \le   \mathbb{E}_\ve \     \sup_{\begin{subarray}{c}
			|\xi_i \ra, |\varphi_j\ra \in B_{\ell_2^{r^2}}\\
			\text{for }i,j = 1,\ldots, n
			\end{subarray}} \ 
		\Big| \frac{1}{n^2} \sum_{ij} \ve_{ij} \ve^*_{ij}   \la \xi_i | \varphi_j\ra    \Big|. 
		\end{align*}
	 An application of the Grothendieck inequality \cite{Grothendieck53} allows us to restrict $r=1$ in the last supremum at the cost of the complex Grothendieck constant $K_G^\mathbb{C}$. Furthermore, Krivine's result that the two dimensional Grothendieck constant is equal to $\sqrt{2}$ \cite{Krivine79} allows to further restrict the supremum to the choice of signs loosing another factor of $\sqrt{2}$ (see \cite[Claim 4.7]{Regev2015} for an explicit argument). In conclusion, we have the following bound:
	$$
 \mathbb{E}_\ve \     \sup_{\begin{subarray}{c}
		|\xi_i \ra, |\varphi_j\ra \in B_{\ell_2^{r^2}}\\
		\text{for }i,j = 1,\ldots, n
\end{subarray}} \ 
\Big| \frac{1}{n^2} \sum_{ij} \ve_{ij} \ve^*_{ij}   \la \xi_i | \varphi_j\ra    \Big|
 \le 
 \sqrt{2}   K_G^\mathbb{C} \  \mathbb{E}_\ve      \sup_{\begin{subarray}{c}
		t_i, s_j \in \{\pm 1\}\\
		\text{for }i,j = 1,\ldots, n
\end{subarray}} \ 
 \frac{1}{n^2} \sum_{ij} \ve_{ij} \ve^*_{ij}  t_i s_j   
$$

But the last quantity is of order $O(1/\sqrt{n})$. One can understand this as a consequence of  Hoeffding's inequality \cite{Hoeffding63}: for each choice $( s_i, t_j )_{i,j}  \in \{\pm 1\}^n \times \{\pm1\}^n $ the probability that $\frac{1}{n^2} \sum_{ij} \ve_{ij} \ve^*_{ij}  t_i s_j $ is larger than $2 /\sqrt{n}$ is upper bounded by $e^{-2n}$. Then, a union bound over the $2^{2n}$ possible sequences $( s_i, t_j )_{i,j} \in \{\pm 1\}^n \times \{\pm1\}^n$ concludes the argument. 
	\end{itemize}

	Joining everything in \eqref{eq3} we obtain the bound in Proposition \ref{Prop_Main2}, i.
	
\end{proof}

\begin{proof}[Proof Proposition \ref{Prop_Main2}, ii.]
	{Remember that we can already use} Theorem \ref{mainThm_2}.I  {here}. It turns out that Proposition \ref{Prop_Main2}, ii. is a consequence of this first part of our main theorem.	
	
	The key idea is understanding the norm $	\big \|\mathbb{E}_\ve  \Phi^{ii.} 	\big \|_{\tilde X^{ii}}$ as the optimization over some family of  strategies with  small enough parameter $\sigma^i$. 	Concretely, considering the characterization of the norm $\tilde X^{ii}$ given in Lemma \ref{lemma_CharacSigma}, we can prove that
	\beq\label{Prop3.ii._eq1}
	\big \|\mathbb{E}_\ve  \Phi^{ii} (\ve) \big\|_{\tilde X^{ii}}
	\le \, \sup_{\begin{subarray}{c}
			r \in \mathbb{N}\\
			    \tilde V, \, \tilde W \in B_{M_{ n r , \tilde k'}}
		\end{subarray}  } \, \omega \left(G_{Rad} ;  \lbrace   \tilde V,\, \tilde W, \,  V,\, W_\ve,\, |\varphi \ra \rbrace_\ve \right)^{1/2}.
	\eeq
	
	The desired bound follows now from realizing that in the   strategies on which this optimization is performed, the second round of local operations, $\tilde V \otimes \tilde W$, is $\ve$-independent. Therefore, for these strategies, according to Example \ref{Example_sigma1}, $\sigma^{i} \approx \frac{\log(n)}{n}$, which, in conjunction with Theorem \ref{mainThm_2}, I., leads to the desired statement. To obtain the precise statement appearing there, we have considered the elementary inequality $( 1 + x )^{1/2} \le 1 + x/2$.
	
	Then, to finish, let us prove the claim \eqref{Prop3.ii._eq1}.

	Recall that, according to Lemma \ref{lemma_CharacSigma} we can write:
	\begin{align*}
		 \big \| \mathbb{E}_\ve \Phi^{ii} (\ve) \big\|_{\tilde X^{ii}}
		&= \sup_{ \begin{subarray}{c}
			r\in\mathbb{N}    \\     \tilde V, \, \tilde W \in B_{M_{ n r , \tilde k}}
			\end{subarray}     } 
		\Big \|  \, \mathbb{E}_\ve \,  ( \tilde V \otimes \tilde W)  \Big( \frac{1}{n^2}  \sum_{i,j} \ve_{ij}\,  ( \la i | \otimes \la j | )\, (V |ij\ra \otimes W_\ve) \, |\varphi \ra \Big) \Big\|_{\ell_2^{r^2}}
		\\
		&\le    \sup_{ \begin{subarray}{c}
			r\in\mathbb{N}    \\     \tilde V, \, \tilde W \in B_{M_{ n r , \tilde k}}
			\end{subarray}     } \,  \mathbb{E}_\ve\, 
		\Big \|    \frac{1}{n^2}  \sum_{i,j} \ve_{ij}\,  ( \la i |\tilde V  \otimes\la j | \tilde W )\, (V |ij\ra \otimes W_\ve) \, |\varphi \ra \Big\|_{\ell_2^{r^2}}.
		\end{align*}
	
	Furthermore, 
		considering the elementary bound
		$
		\, \mathbb{E}_\ve   \phi(\ve)
		\le \Big( \mathbb{E}_\ve   {\phi(\ve)}^2 \Big )^{\frac{1}{2}}\!,
		$
		valid for any function $\phi:\calQ_{n^2} \rightarrow \mathbb{R}$, we can finally write:
		\begin{align*}
		\big \| \mathbb{E}_\ve \Phi^{ii} (\ve) \big\|_{\tilde X^{ii}}
		&\le \sup_{ \begin{subarray}{c}
			r\in\mathbb{N}    \\     \tilde V, \, \tilde W \in B_{M_{ n r , \tilde k}}
			\end{subarray}     } \, \Big( \mathbb{E}_\ve\, 
		\Big \|    \frac{1}{n^2}  \sum_{i,j} \ve_{ij}\,  (\la i |\tilde V  \otimes \la j |\tilde W )\, (V |ij\ra \otimes W_\ve) \, |\varphi \ra \Big\|_{\ell_2^{r^2}}^2 \Big)^{\frac{1}{2}}
		\\
		& =  \sup_{\begin{subarray}{c}
			r\in\mathbb{N}    \\     \tilde V, \, \tilde W \in B_{M_{ n r , \tilde k}}
			\end{subarray}  } \, \omega \left(G_{Rad} ;  \lbrace   \tilde V,\, \tilde W, \,  V,\, W_\ve,\, |\varphi \ra \rbrace_\ve \right)^{1/2},
	\end{align*}
	as claimed.
\end{proof}

}

We make a final comment that, in some sense, connects with the next section where we will discuss possible extensions of the approach presented up to this point.

\begin{remark}\label{Rmk_normalizationInt1/2}
	The appearance of the norms $X^{i} = M_{ r,k}   {\otimes_{min} M_{r,\tilde k}  },\ X^{ii} = \s^{ {\tilde k}, n} \otimes_{  (\ve, \pi)_{ \sfrac{1}{2} }  } \s^{ {\tilde k}, n}$ above might seem, at some point, arbitrary, in the sense that we have  used these norms   {merely} to \emph{upper} bound the value $\omega(G_{Rad}, \Ss^\UU)$. Part of the motivation to consider these spaces is the fact that we are able to properly understand their type properties. But we can wonder: is any norm upper bounding $\omega(G_{Rad}, \Ss^\UU)$ a reasonable choice provided that we can control the relevant type constants? Obviously, this is not the case. Actually, in Section \ref{Sec6} we explore further this issue. By now, let us note that the chosen norms also satisfy some basic normalization conditions. In particular, it can be shown that  the elements constituting $\Phi^{i} $, $\Phi^{ii}$ are well normalized when regarded as elements in $X^i$ and $ X^{ii}$, respectively.  Concretely, for each $i,\, j \in [n]$
	$$
	 \left \| (\la i |\tl{V}_{\ve} \otimes \la j|\tl{ W}_{\ve})\, (V|ij\ra  \otimes \id_{\ell_2^{\tilde k}}) \right \|_{X^i} \le 1,
	$$
	 {for any contractive operators $\tl{V}_{\ve}, \, \tl{W}_{\ve}, \, V,\, W_\ve$ ,}	and
	$$
	 \left \| | i \ra \otimes | j \ra \otimes (V|ij\ra  \otimes W_\ve) |\varphi \ra \right \|_{ X^{ii}  } \ \le 1,
	$$
	 {for any unit vector $|\varphi\ra$ and contractive $V$, $W_\ve$.}

	The first bound is straightforward. Since $\tl{V}_{\ve} \otimes \tl{ W}_{\ve}$ and $V   \otimes \id_{\ell_2^{\tilde k}}$ are contractive operators, $\la i |\tl{V}_{\ve} \otimes \la j|\tl{ W}_{\ve}$ and  $V|ij\ra  \otimes \id_{\ell_2^{\tilde k}}$ are also contractive and  the same applies to  their composition.
		
		For the second bound, fixing $i,\, j$, we first notice that $|\tilde \varphi \ra :=  (V|ij\ra  \otimes \id_{\ell_2^{\tilde k^2}}) (\id_{\ell_2^{k}} \otimes W_\ve)|\varphi \ra$ has norm $\| |\tilde \varphi \ra \|_{\ell_2^{\tilde k^2}} \le 1$. Furthermore, considering the norm-one injections $\iota_i: \ell_2^{\tilde k}\ni |\varphi \ra \mapsto | i \ra \otimes |\varphi \ra \in \s^{\tilde k, n}$, we have that $  | i \ra \otimes | j \ra  \otimes |\tilde \varphi \ra =  \iota_i \otimes \iota_j \big( | \tilde \varphi \ra \big)$. Therefore
		$$
		\big \| | i \ra \otimes | j \ra  \otimes |\tilde \varphi \ra  \big\|_{  \s^{\tilde, n} \otimes_{  (\ve, \pi)_{ \sfrac{1}{2} }  } \s^{\tilde, n} } \le \big \| |\tilde \varphi \ra \big \|_{\ell_2^{\tilde k^2}} \ \big \|  \iota_i \otimes \iota_j: \ell_2^{\tilde k} \rightarrow  \s^{\tilde k, n} \otimes_{  (\ve, \pi)_{ \sfrac{1}{2} }  } \s^{\tilde k, n}\big \| \le 1.
		$$
		It remains to justify that, in fact,
		$$
		\big \|  \iota_i \otimes \iota_j: \ell_2^{\tilde k} \rightarrow  \s^{\tilde k, n} \otimes_{  (\ve, \pi)_{ \sfrac{1}{2} }  } \s^{\tilde k, n}\big \| \le 1.
		$$
		This can be proved recalling that $\s^{\tilde k, n} \otimes_{  (\ve, \pi)_{ \sfrac{1}{2} }  } \s^{\tilde k, n}$ is the interpolation space $ (    \s^{\tilde k, n} \otimes_{ \ve } \s^{\tilde k'\!, n}, $ $ \s^{\tilde k, n} \otimes_{  \pi  } \s^{\tilde k, n}    )_{\frac{1}{2}}$ and $\ell_2^{\tilde k^2}$ can be also regarded as the space  $ \left(    \ell_2^{\tilde k} \otimes_{ \ve } \ell_2^{\tilde k}, \ell_2^{\tilde k }\otimes_{  \pi  } \ell_2^{\tilde k}    \right)_{\frac{1}{2}}$.  { The last assertion can be shown noticing that $ \ell_2^{\tilde k} \otimes_\ve \ell_2^{\tilde k} = M_{\tilde k}$ and $\ell_2^{\tilde k} \otimes_\pi \ell_2^{\tilde k} = \s^{\tilde k}$. Given that, the isometric equivalence (recall Theorem \ref{Int_Sp's}) $(M_{\tilde k},  \s^{\tilde k} )_{1/2}  = \mathcal{S}_2^{\tilde k}=\ell_2^{\tilde k^2}$ provides the stated fact.} Then,
		\begin{align*}
				&\big \|  \iota_i \otimes \iota_j: \ell_2^{\tilde k} \rightarrow  \s^{\tilde k, n} \otimes_{  (\ve, \pi)_{ \sfrac{1}{2} }  } \s^{\tilde k, n}\big \| \hspace{-2cm} 
				\\
				 &\qquad\le 		\big \|  \iota_i \otimes \iota_j: \ell_2^{\tilde k}  \otimes_ \ve \ell_2^{\tilde k}\rightarrow  \s^{\tilde k, n} \otimes_{ \ve } \s^{\tilde k, n}  \big \|^{\frac{1}{2}}  
				\ \big\|  \iota_i \otimes \iota_j: \ell_2^{\tilde k}  \otimes_\pi \ell_2^{\tilde k}\rightarrow  \s^{\tilde k, n} \otimes_{ \pi} \s^{\tilde k, n}  \big \|^{\frac{1}{2}}
				\\
			   & \qquad \le 		\big \|  \iota_i : \ell_2^{\tilde k}  \rightarrow  \s^{\tilde k, n}   \big \|\  \big \|  \iota_j : \ell_2^{\tilde k}  \rightarrow  \s^{\tilde k, n}   \big \| \le 1.
		\end{align*}
		
\end{remark}

 \section{A conjecture towards unconditional lower bounds}\label{Sec6}
 
In the previous section, we have modified the naïve choice \eqref{Def_Phi0} for $\Phi$ in order to circumvent the problem that $\big \|\mathbb{E}_\ve  \Phi (\ve) \big\|_{\ell_2^{r^2} } $ can be in general																																																																																																																																												
too large, damning that way the bounds  obtained through Corollary \ref{Cor1} to be trivial. The variations $\Phi^i$, $\Phi^{ii}$ allowed us to obtain the bounds in Theorem \ref{mainThm_2}. An unsatisfactory feature of this result is that, in order to obtain concrete bounds on the quantum resources employed by a given strategy for $G_{Rad}$, we still need to make some additional assumption on that strategy. Recall that, in particular, the bounds in Theorem \ref{mainThm_2} depend on the regularity parameters $\sigma^i$, $\sigma^{ii}$. Ideally, we would like to obtain bounds only depending on the dimension of the quantum systems Alice and Bob manipulate. 

 Following this line of thought, one could ask whether, given a strategy, is possible to construct a corresponding assignment $\Phi$ that additionally displays the property of being regular enough, that is, with $\sigma_{\Phi}\lesssim_{\log} 1/n$.  The answer is affirmative, but the cost of doing so is that the output Banach space of $\Phi$ becomes more involved and its type properties escape from the techniques used in this work. We  define:
  \beq\label{Def_Phi3}
  \begin{array}{rccc}
  \Phi^{iii} : &  \calQ_{n^2} & \longrightarrow &  \left( \s^{\tilde k, n } \otimes_{\mathfrak{S}_2^{cb-w}} \s^{\tilde k, n} \right)\otimes_\ve \ell_2^{k \tilde k}
  \\[1em]
  & \ve & \mapsto  &  \Phi^{ii {i}} (\ve)  = \frac{1}{n^2}  \, \sum_{ij} \,   \, \varepsilon_{ij} \, \la i | \otimes \la j | \otimes (V|ij\ra  \otimes \id_{\ell_2^{\tilde k}})  
  \end{array},
\eeq
  that relates with the value of the game $G_{Rad}$ as stated in the following
  \begin{lemma}\label{Conjecture_Lemma1}
For any pure strategy $\Ss^{\UU} \in \mathfrak{S}_{s2w;\tilde k, k}$:
$$
 \omega (G_{Rad}; \Ss^\UU)^{1/2}  \lesssim \mathbb{E}_\ve \ \| \Phi^{iii} (\ve ) \|_{X^{iii} },
$$
where $X^{iii} := \left( \s^{\tilde k, n } \otimes_{\mathfrak{S}_2^{cb-w}} \s^{\tilde k, n} \right)\otimes_\ve \ell_2^{k \tilde k} $
  \end{lemma}
\begin{proof}
	For each $\ve\in \calQ_{n^2}$, we have to interpret the tensor $\Phi^{iii} (\ve )$ as the mapping:
	  $$
	\begin{array}{rccc}
	\Phi^{iii}(\ve) : &  \ell_2^{k \tilde k}  & \longrightarrow &  \s^{\tilde k, n } \otimes_{\mathfrak{S}_2^{cb-w}} \s^{\tilde k, n} 
	\\[1em]
	& |\varphi\ra & \mapsto  &   \Phi^{iii}(\ve) ( |\varphi\ra ) = \frac{1}{n^2}  \, \sum_{ij} \,   \, \varepsilon_{ij} \, \la i | \otimes \la j | \otimes (V|ij\ra  \otimes \id_{\ell_2^{\tilde k}})\,|\varphi\ra 
	\end{array} .
	$$ 
	Then, the norm of this map is
	\begin{align*}
	\|\Phi^{iii}(\ve) \| &= \sup_{|\varphi \ra \in B_{ \ell_2^{k \tilde k} } }   \Big\|
	  \frac{1}{n^2}  \, \sum_{ij} \,   \, \varepsilon_{ij} \, \la i | \otimes \la j | \otimes  (V|ij\ra  \otimes \id_{\ell_2^{\tilde k}})\,|\varphi\ra  
	    \Big\|_{  \s^{\tilde k, n } \otimes_{\mathfrak{S}_2^{cb-w}} \s^{\tilde k, n}  } 
	    \\
	    &= \sup_{ W \in B_{ M_{ \tilde k} } } \sup_{|\varphi \ra \in B_{ \ell_2^{k \tilde k} } }   \Big\|
	    \frac{1}{n^2}  \, \sum_{ij} \,   \, \varepsilon_{ij} \, \la i | \otimes \la j | \otimes  (V|ij\ra  \otimes W)\,|\varphi\ra 
	    \Big\|_{  \s^{\tilde k, n } \otimes_{\mathfrak{S}_2^{cb-w}} \s^{\tilde k, n}  }    .
	\end{align*}  
Recalling once more Lemma \ref{lemma_CharacSigma}, we can write explicitly the norm above as:
	\begin{align*}
\|\Phi^{iii}(\ve) \| 
&=
 \sup_{ \begin{subarray}{c}
	m \in \mathbb{N} \\
	\tilde V, \tilde W \in B_{M_{\tilde k, n m}}
	\end{subarray} }
\sup_{ \begin{subarray}{c}
	 W \in B_{ M_{ \tilde k} } \\
 |\varphi \ra \in B_{ \ell_2^{k \tilde k} } \end{subarray} }
	   \Big\|
\frac{1}{n^2}  \, \sum_{ij} \,   \, \varepsilon_{ij} \,( \la i | \tilde V \otimes\la j | \tilde W) \,  (V|ij\ra  \otimes W )\,|\varphi\ra  
\Big\|_{ \ell_2^{m^2}  }    .
\end{align*}
	
	Finally, squaring this last expression and taking the expectation over $\ve$ we conclude that:
		\begin{align*}
	\mathbb{E}_\ve \, \|\Phi^{iii}(\ve) \| ^2
	&=
	\mathbb{E}_\ve \,
	\sup_{ \begin{subarray}{c}
		m \in \mathbb{N} \\
		\tilde V, \tilde W \in B_{M_{\tilde k, n m}}
		\end{subarray} }
	\sup_{ \begin{subarray}{c}
		W \in B_{ M_{ \tilde k} } \\
		|\varphi \ra \in B_{ \ell_2^{k \tilde k} } \end{subarray} }
	\Big\|
	\frac{1}{n^2}  \, \sum_{ij} \,   \, \varepsilon_{ij} \,( \la i | \tilde V \otimes\la j | \tilde W) \,  (V|ij\ra  \otimes W)\,|\varphi\ra  
	\Big\|_{ \ell_2^{m^2}  } ^2
	\\
	& \ge 
		\mathbb{E}_\ve \,
	\Big\|
	\frac{1}{n^2}  \, \sum_{ij} \,   \, \varepsilon_{ij} \,( \la i | \tilde V_\ve \otimes\la j | \tilde W_\ve) \,  (V|ij\ra  \otimes W_\ve)\,|\varphi\ra  
	\Big\|_{ \ell_2^{m^2}  } ^2 = \omega(G_{Rad}; \Ss^\UU),
	\end{align*}
	where we have considered that $\Ss^\UU = \lbrace \tilde V_\ve,\tilde W_\ve, V, W_\ve, |\varphi\ra \rbrace_\ve$. With that we are almost done. This last expression is enough to obtain 
	$$
	\omega(G_{Rad}; \Ss^\UU) \le 	\mathbb{E}_\ve \, \|\Phi^{iii}(\ve) \| ^2    {\approx \left( \mathbb{E}_\ve \, \|\Phi^{iii}(\ve)\|  \right)^2 },
	$$
	 {where the last equality (up to constants) can be obtained } using Kahane inequality \cite{KahaneBook}. {This is the claim of the lemma.} 
\end{proof}

 Now, notice that $\Phi^{iii}$ is by construction a linear map of the kind of Example \ref{Example_sigma1}, and, consequently, $\sigma_{\Phi^{iii}} \lesssim \log(n)/n $. Furthermore, by symmetry, $\mathbb{E}_\ve \Phi^{iii}_{\Ss^\UU } = 0$. Therefore, Corollary \ref{Cor1}  applied to the statement of Lemma \ref{Conjecture_Lemma1} implies:
 \begin{equation}\label{Eq_Conjecture_1}
 \omega(G_{Rad}; \Ss^{\UU}) 
\lesssim_{\log}   \left(\frac{  \mathrm{T}_2^{(n^2)} (X^{iii})}{n} \right)^2.
 \end{equation}

 The problem now reduces to find a good estimate for the type-2 constant in the last expression. 
 
 We note that the norm $ X^{iii}$ is the smallest one for which we were able to prove an equivalent to Lemma \ref{Conjecture_Lemma1}. However, the whole argument from this lemma until here would be valid for any norm larger than $ X^{iii} $ fulfilling a normalization condition with respect to the elements that sum up to $\Phi^{iii}(\ve)$. We will be more explicit later on. An example of such a norm is $ X^{ii} \otimes_\ve \ell_2^{k \tilde k}$ where $X^{ii} = \s^{\tilde k, n } \otimes_{ (\ve, \pi)_{\sfrac{1}{2}} } \s^{\tilde k, n}$. Motivated by the result obtained previously about the type of $X^{ii}$, Equation \eqref{Type_Int1/2}, we are led to conjecture that:
 \begin{conjecture}[strongest form] For any natural numbers $n,\, m , \, p$:
 	\beq\label{Conjecture_1}
 	\mathrm{T}_2^{(n^2)} \left( \big( \s^{m, n } \otimes_{(\ve, \pi)_{\sfrac{1}{2}} } \s^{m, n } \big)\otimes_\ve \ell_2^{p}  \right) \lesssim_{\log} \mathrm{T}_2^{(n^2)} \left( \s^{p, n } \otimes_{ (\ve, \pi)_{\sfrac{1}{2}} } \s^{p, n}  \right) \lesssim_{\log} n^{3/4}.
 	\eeq
 \end{conjecture}
	A weaker conjecture which would also imply the desired bounds in the setting of PBC is:
 \begin{conjecture}[weaker form]
	\beq\label{Conjecture_1'}
	\mathrm{T}_2^{(n^2)} \left( \big( \s^{m, n } \otimes_{ (\ve, \pi)_{\sfrac{1}{2}} } \s^{m, n} \big)\otimes_\ve  \ell_2^{ p }  \right) \lesssim_{\log}   n^{\beta}  \qquad \text{for some } \beta <1.
	\eeq
\end{conjecture}
 
According to what we explained above, there is a plethora of norms for which the positive resolution of the corresponding conjecture would imply unconditional exponential lower bounds for the resources in attacks for PBC. Next, we formalize this discussion characterizing those norms and then we rewrite the Conjecture  in a unified form.

First,  we characterize what we need from a norm $X$ to follow the previous argument substituting $X^{iii}$ by this $X$.  In this section we refer to $X$ as a \emph{valid norm} if it satisfies:
\begin{enumerate}[P.i.]
	\item $X$ is a norm on the algebraic tensor product  $\s^{m, n } \otimes \s^{m, n} \otimes \ell_2^{p}$;
	\item $\| x \|_X \gtrsim \| x \|_{X^{iii}}$ for any $ x \in X$;
	\item 	 $  \big\|  |i\ra \otimes |j\ra \otimes (V|ij\ra  \otimes \id_{\ell_2^{\tilde k'}})   \big \|_X  \le 1 $ for any  {contraction} $V$. 
\end{enumerate}
Notice that P.ii. guarantees a relation with the value of $G_{Rad}$ in analogy with Lemma \ref{Conjecture_Lemma1} and P.iii. guarantees that $\Phi^{iii.} :  {\calQ}_{n_2} \rightarrow X$ still falls in the setting of Example \ref{Example_sigma1}, i.e., we still have $\sigma_{\Phi^{iii}} \lesssim \log(n)/n $. These two properties therefore translate in the fact that the bound \eqref{Eq_Conjecture_1} is still true with the type-2 constant of any valid norm $X$ instead of $X^{iii} $.

We can state
 \begin{conjecture}[even weaker form]
	For some valid norm, i.e. a norm $X$ satisfying properties \emph{P.i., P.ii.} and \emph{P.iii.} above, and some dimension independent constant $\beta <1:$
	\beq\label{Conjecture_3}
	\mathrm{T}_2^{(n^2)} \left( X  \right) \lesssim_{\log}   n^{\beta}.
	\eeq
\end{conjecture}

Now, to state our conjecture in its weakest form we need to introduce the notion of type constant of an operator $F:X\rightarrow Y$. The type-2 constant of a linear map $F:X\rightarrow Y$ is the infimum of the constants $T$ such that
\beq \nn
\left( \mathbb{E}_\ve   \Big[ \big\|  \sum_{i} \varepsilon_i F( x_i) \big\|_Y^2 \Big]   \right)^{1/2}
\hspace{-1mm} \le	 \rmT \left( \sum_{i} \|x_i\|_X^2 \right)^{1/2},
\eeq
for any finite sequence $\lbrace x_i \rbrace_{i} \subset X$. In analogy with the case of the type constant of a Banach space, when the cardinal of this sequence is restricted, we refer to the type-2 constant with $m$ vectors of $F:X\rightarrow Y$ and denote $ \mathrm{T}_2^{(m)} (F:X\rightarrow Y) $.

We are  interested here in the type of the identity map $\id: X \rightarrow X^{iii}$, being X a \emph{valid norm}. In fact, the final statement of our conjecture is as follows:
 \begin{conjecture}[weakest form]
	For some valid norm, i.e. a norm $X$ satisfying  properties \emph{P.i., P.ii.} and \emph{P.iii.} above, and some dimension independent constant $\beta <1:$
	\beq\label{Conjecture_4}
	\mathrm{T}_2^{(n^2)} \left(  \id: X \rightarrow X^{iii}  \right) \lesssim_{\log}   n^{\beta}.
	\eeq
\end{conjecture}

\begin{remark}
Notice that in particular, $\mathrm{T}_2^{(n^2)} \left(  \id: X \rightarrow X^{iii}  \right)  \lesssim \mathrm{T}_2^{(n^2)} \left( Y  \right) $ for any \emph{valid norm} $Y$ such that $\| x \|_{X^{iii}} \lesssim \| x \|_Y \lesssim \| x \|_X$. Therefore, the last statement for our conjecture, Equation \eqref{Conjecture_4}, is indeed weaker than the previous ones. 
\end{remark}

Within the family of \emph{valid} norms characterized by properties P.i., P.ii., P.iii. we obviously find the spaces $X^{iii}$ and $\big( \s^{m, n } \otimes_{ (\ve, \pi)_{\sfrac{1}{2}} } \s^{m, n} \big )\otimes_\ve  \ell_2^{ p}$. But also,  {the} space
$\big( \s^{m, n } \otimes_{\mathfrak{S}^{w}_2} \s^{m, n} \big)\otimes_\ve  \ell_2^{ p}$, see Section \ref{Sec.2.4.1} for the definition. An obstruction for the techniques used in this work to obtain upper bounds for the type constants of these spaces is the pathological behaviour of the injective tensor product with respect to interpolation methods \cite{Lemerdy98}.  In order to support the validity of the stated conjecture,  we explore next the most direct approaches to disprove it, lower bounding the type-2 constant of the  spaces involved. We  find that these approaches do not lead to bounds stronger than $\T_2 (X) \gtrsim_{\log}  n^{3/4}$ for at least some \emph{valid} norm $X$. 

 {In first place, one can obtain lower bounds for the type constants of a space X by estimating the
type constant of its subspaces, since $\rmT_p(X) \ge \rmT_p(S)$ for any subspace $S \subseteq X$. Restricting to the case of valid
norms, the type constants of the simplest subspaces are not large enough to disprove our conjecture (see
\cite[Section 4.7.1 ]{Kubicki_thesis} for details). As we explain in the next section, another, less trivial,  way to obtain lower bounds for
the type 2-constant of a normed space X is by studying its volume ratio.}

\subsection{Volume ratio}

Although the Banach spaces that appear in this work are prominently complex, for the sake of simplicity we will restrict ourselves to real spaces in this section. There exist standard tools \cite{Michal41,Taylor43,Wenzel95,Munoz1999} to   {transfer} results in this case to the complex domain, albeit some technicalities might appear in that process \cite{Wenzel97}. Since our aim here is restricted to  showing some evidence in favour of our conjecture, we do not think that these intricacies add anything of essential importance to the following discussion.   

A standard approach to understand the type/cotype properties of a space $X$ consists in the computation of its volume ratio, $\mathrm{vr}(X)$, a notion originated in \cite{Szarek_78,Szarek_80}. The reason is that this parameter provides a lower bound for the cotype-2 constant. This is the content of the  following result due to Milman and Bourgain:
\begin{theorem}[\cite{Bourgain_87}]\label{ThmBourgain}
	For a Banach space $X$,
	$$
	\mathrm{C}_2(X) \log\left( 2 \, \mathrm{C}_2(X)  \right) \gtrsim \mathrm{vr} (X) .
	$$
\end{theorem} 
 Taking into account the duality between type and cotype constants, Equation \eqref{typecotype_duality2}, the last result translates into a lower bound for the type-2 constant of the dual space:
 $$
 \T_2(X) \ge \mathrm{C}_2(X^*) \gtrsim_{\log} \mathrm{vr} (X^*),
 $$
giving as another technique to try to disprove  \eqref{Conjecture_3}. In this section we upper bound the volume ratio of various \emph{valid} norms obtaining results that are again compatible with a positive resolution of the conjecture of the previous section.



We start defining the volume ratio of a normed space $X$, $\mathrm{vr}(X)$. Given a $d$-dimensional Banach space $X$,
\beq\label{Vr_Eq0}
\mathrm{vr} (X)  = \left( \frac{\mathrm{vol}_d (B_X) }{  \mathrm{vol}_d  (\mathcal{E}_X)  }  \right)^{1/d},
\eeq
where $\mathcal{E}_X$ is the ellipsoid of maximal volume contained in $B_X$ and $\mathrm{vol}_d( \, \cdot \, )$ denotes the $d$-dimensional Lebesgue measure.  Before stating the main result of this section, we make now a tiny digression about the relation between volume ratio and cotype. In few words, this relation is still far from being well understood. The question about the existence of some direct relation between the  volume ratio of a space and its cotype -- in the opposite direction to Theorem \ref{ThmBourgain} -- was already asked in the seminal work \cite{Szarek_80} and also in the more recent \cite{Tomczak16}. While it is known that volume ratio and cotype cannot be equivalent in general\footnote{ It can be seen that the volume ratio of the space   {$\ell_\infty^{n^\alpha} \oplus_\infty\ell_2^n$} is bounded by a universal constant for any $n \in \mathbb{N}$ and any $ 0 < \alpha <1$. However, the cotype-2 constant of  this space is of order $n^{\alpha/2}$. We are indebted to Elisabeth Werner for kindly communicating us this counterexample. }, it is not known whether a converse to Theorem \ref{ThmBourgain} might hold (maybe up to factors that are logarithmic in the dimension) for spaces with additional structure such as tensor norms on tensor products of $\ell_p$ spaces, for instance. 
  Studying further these questions is an extremely interesting avenue to tackle the problems we are concerned with in this work, at the same time as shedding light on the relation between two very fundamental notions in local Banach space theory.

 
 We focus on  spaces of the form $\big( \s^{m, n } \otimes_{\alpha} \s^{m, n} \big)\otimes_\ve  \ell_2^{ p }$,   where  $\s^{m,n} $ must be understood as $\ell_2^m \otimes_\pi \ell_2^n$ and the $\ell_2$ spaces that appear from now on, as real Hilbert spaces unless the opposite is indicated.  We prove:

\begin{theorem}\label{Thm_vr}
	Let $\alpha$ be  a \emph{tensor norm}  such that, 	for any $ x \in \s^{m, n } \otimes_{\alpha} \s^{m, n}$:
\begin{enumerate}
	\item 	$
	\| x  \|_{\s^{m, n } \otimes_{\alpha} \s^{m, n} } \le   \| x  \|_{\s^{m, n } \otimes_{\pi} \s^{m, n} } ^{1/2}   \| x  \|_{\s^{m, n } \otimes_{\ve} \s^{m, n} } ^{1/2}  ;
	$
	\item  $    \| x  \|_{ \ell_2^{n^2 m^2} }  \le \| x  \|_{\s^{m, n } \otimes_{\alpha} \s^{m, n} } .$
	\end{enumerate}  
	Then, considering $X = \left( \s^{m, n } \otimes_{\alpha} \s^{m, n}   \right) \otimes_\ve \ell_2^{p}$,
	$$
	\mathrm{vr} (X^*) \lesssim n^{3/4}.
	$$
\end{theorem}

The proof uses several standard tools from geometric Banach space theory, mainly following the approach of   \cite{Tomczak16}. But before going into the proof, we note that some of our \emph{valid} norms indeed fulfill the conditions of the theorem. { For illustrative purposes, we briefly comment on the case of the (complex) spaces $(\s^{m, n } \otimes_{\mathfrak{S}_2^{w}} \s^{m, n}) \otimes_\ve \ell_2^p $  {and}  $( \s^{m, n } \otimes_{(\varepsilon,\pi)_{1/2}} \s^{m, n} )\otimes_\ve \ell_2^p$}   { that have appeared before in our work. Both fulfil conditions 1 and 2 in the statement above.}   {Let's see that:

From Proposition \ref{Prop_RelNorms} we know that 
\begin{equation}\label{eq:vr-1}
\| x  \|_{\s^{m, n } \otimes_{\mathfrak{S}_2^{w}} \s^{m, n} } \le \| x  \|_{\s^{m, n } \otimes_{(\varepsilon,\pi)_{1/2}} \s^{m, n} }.
\end{equation}
Standard properties of interpolation (Theorem \ref{Int_IntProp}) guarantee that $(\varepsilon,\pi)_{1/2}$ is a tensor norm\footnote{at least in the category of all finite dimensional normed spaces, which is the relevant setup for this work} fulfilling 
\begin{equation}\label{eq:vr-2}
\| x  \|_{(Y_0,Y_1)_{1/2}}\le \| x  \|_{Y_0} ^{1/2}   \| x  \|_{Y_1 } ^{1/2}\text{, \quad for any finite dimensional Banach spaces $Y_0$, $Y_1$.}
\end{equation} In order to see this last inequality, just apply, for a given $x\in Y_0  \otimes Y_1 $, Theorem \ref{Int_IntProp} with $X_0=X_1=\mathbb R$ and $f(\lambda)=\lambda x$. 

Finally, using the fact that ${\mathfrak{S}_2^{w}}$ coincides with the Euclidean (or Hilbert-Schmidt) norm in the tensor product of Hilbert spaces (fact that follows directly from the definition of the ${\mathfrak{S}_2^{w}}$ norm), together with the facts that it is a tensor norm and the identity map $\mathcal{S}_1^{m,n}\rightarrow \mathcal{S}_2^{m,n}$ has norm $\le 1$, one gets
 \begin{equation}\label{eq:vr-3} \| x  \|_{ \ell_2^{n^2 m^2} }  \le \| x  \|_{\s^{m, n } \otimes_{\mathfrak{S}_2^{w}} \s^{m, n}.} 
 \end{equation}
The desired claim follows putting together equations (\ref{eq:vr-1}),  (\ref{eq:vr-2}) and  (\ref{eq:vr-3}.  { An important point to stress here is that, in order to apply Theorem \ref{Thm_vr}, real versions of these spaces must be considered.  In the first case, one obtains a real version of $\mathfrak{S}_2^w$ just restricting  the underlying field to $\mathbb R$ in Definition \ref{Def_WeakSchatten}. The second case is a bit more subtle since the complex interpolation method is inherently defined over complex normed spaces. A way to formalize the discussion at this point is considering the real interpolation method \cite[Chapter 3]{BerghLofstrom76}.  Following the standpoint fixed at the beginning of this section, we leave aside these technicalities that we think do not add much to our discussion.}
 
 }

An important feature of  spaces  {of the form  $\big( \s^{m, n } \otimes_{\alpha} \s^{m, n} \big)\otimes_\ve  \ell_2^{ p }$} is {the fact} that, by construction, they \emph{have enough symmetries}. This will be exploited in the following proof with no further mention. The reader can find some additional information in Appendix \ref{Appendix_EnoughSym}. 

\begin{proof}
	
We start noticing that  $\alpha$ being a tensor norm translates into  the fact  that $X$   \emph{has enough symmetries}. This means that the only operator on that space that commutes with every isometry is the identity (or a multiple of it). The same happens with the dual $X^*$.  Next we give an alternative way to compute the volume ratio using this property. To simplify notation, denote $d = \dim(X)  = n^2 m^2 p$. Then, we can bound \eqref{Vr_Eq0} as follows:

	\begin{align}
	\mathrm{vr} (X^*)   
	&\stackrel{(i.)}{=} 
	\left( \frac{\mathrm{vol}_d (B_{X^*}) }{  \mathrm{vol}_d  (B_{\ell_2^d})  }  \right)^{1/d} \left\| \id: \ell_2^d \rightarrow X^* \right\|
	\nn\\
	&\!\stackrel{(ii.)}{\le} 
		\left( \frac{  \mathrm{vol}_d  (B_{\ell_2^d})  }{  \mathrm{vol}_d (B_{X})  }  \right)^{1/d} \left\| \id: \ell_2^d \rightarrow X^* \right\|
	\nn	\\
		&\!\stackrel{(iii.)}{\le} 
		\frac{\left\| \id: \ell_2^d \rightarrow X^* \right\|}{\sqrt{d}}    \left( \frac{  1  }{  \mathrm{vol}_d (B_{X})  }  \right)^{1/d} 
		\nn \\\label{Vr_Eq1}
		&\!\stackrel{(iv.)}{\le} 
		\frac{\left\| \id: \ell_2^d \rightarrow X^* \right\|}{\sqrt{d}}   \ \mathbb{E} \, \| G \|_{X} \, ,
	\end{align}
		where $G = \sum_{i,j,k,l,h} g_{ijklh} |i\xa j|\otimes |k\xa l| \otimes \la h |$ is a tensor in $X^*$ with i.i.d. gaussian entries $g_{ijklh}$. The expectation is over these random variables. With respect to  the chain of claims implicit in the previous manipulation: (i.) follows from the fact that the maximal volume ellipsoid $\mathcal{E}_{X^*}$ coincides with $\left\| \id: \ell_2^d \rightarrow X^* \right\|^{-1}  B_{\ell_2^d}$ when $X^*$ has enough symmetries \cite[Section 16]{Tomczak1989banach}, in (ii.) we have used the famous Blaschke-Santaló inequality \cite[Section 7]{pisier89_book}, in (iii.), the standard volume estimate for the Euclidean ball $\mathrm{vol}_d (\mathsf{ball} ({\ell_2^d})) \approx d^{-d/2}$ and  (iv.) follows from Lemma 3.4. in \cite{Tomczak16}.

	As a consequence, to obtain the stated bound we have to estimate the quantities   $\| \id: \ell_2^d \rightarrow X^* \|$   and $ \mathbb{E}  \, \| G \|_{X}$.
	\begin{itemize}
	\item Upper bounding  $\left\| \id: \ell_2^d \rightarrow X^* \right\|$:
	
	We show two complementary bounds for this quantity. The first one uses  the second condition in the statement of the theorem, that can be equivalently stated as: $ \big \| \id : \s^{m, n } \otimes_{\alpha} \s^{m, n} \longrightarrow \ell_2^{n^2m^2} \big\| \le 1.$ This allows us to bound:
	\begin{align*}
	\left\| \id: \ell_2^d \rightarrow X^* \right\|  &= \left\| \id: X \rightarrow \ell_2^d  \right\| 
	\\
	&= \left\| \id:  \left( \s^{m, n} \otimes_\alpha \s^{m, n} \right) \otimes_\ve \ell_2^{p}  \longrightarrow \ell_2^{n^2 m^2 p} \right\|
	\\
	& \le \left\| \id:  \ell_2^{n^2m^2} \otimes_\ve \ell_2^{p}  \longrightarrow \ell_2^{n^2 m^2 p} \right\|
	\\
	&\le \sqrt{p}.
	\end{align*}
	The mentioned hypothesis was used in the first inequality above.
	
	Our second bound comes from the observation  that the operator norm we want to bound is indeed upper bounded by the 2-summing norm of the identity between $\s^{m, n} \otimes_\alpha \s^{m, n}$ and $\ell_2^{ n^2 m^2}$. We can alternatively understand the studied norm as:
	\begin{align*}
	\left\| \id: \ell_2^d \rightarrow X^* \right\|  &= \left\| \id: X \rightarrow \ell_2^d  \right\| 
	\\
	&= \left\| \id: \ell_2^{p} \otimes_\ve \left( \s^{m, n} \otimes_\alpha \s^{m, n} \right) \longrightarrow \ell_2^{p}(   \ell_2^{n^2 m^2}  ) \right\|
	\\
	& \le \sup_{ p\in\mathbb{N} } \left\| \id: \ell_2^{p} \otimes_\ve \left( \s^{m, n} \otimes_\alpha \s^{m, n} \right) \longrightarrow \ell_2^{p}(   \ell_2^{n^2 m^2}  ) \right\|
	\\
	&= \pi_2\left(   \id:  \s^{m, n} \otimes_\alpha \s^{m, n}  \longrightarrow  \ell_2^{n^2 m^2}   \right),
	\end{align*}
	where the last equality is simply the definition of the 2-summing norm of the indicated map -- recall \eqref{Def_2summing}. While now we don't need the hypothesis used before, we need to invoke the tensor norm properties of $\alpha$. Hopefully,   thanks to this property\footnote{ See again Appendix \ref{Appendix_EnoughSym} for clarification.}, Lemma 5.2. of \cite{Defant06} provides us a satisfactory way to compute the above norm. Under the consideration that $\s^{m, n} \otimes_\alpha \s^{m, n}$ as well as $\ell_2^{n^2 m^2}$ have \emph{enough symmetries in the orthogonal group} -- see Appendix \ref{Appendix_EnoughSym} --, the cited lemma allows to write the following  identity:
	$$
	\pi_2\left(   \id:  \s^{m, n} \otimes_\alpha \s^{m, n}  \longrightarrow  \ell_2^{n^2 m^2}   \right) = \frac{ n m  }{ \left\| \id:  \ell_2^{n^2 m^2}  \longrightarrow \s^{m, n} \otimes_\alpha \s^{m, n} \right\| }.
	$$
	Taking into account the two bounds above, we can state that, under the conditions in the theorem:
	\beq\label{Vr_Eq2}
		\left\| \id: \ell_2^d \rightarrow X^* \right\| \le  \min\left(\sqrt{p}, \ \frac{ n m  }{ \left\| \id:  \ell_2^{n^2 m^2}  \longrightarrow \s^{m, n} \otimes_\alpha \s^{m, n} \right\| } \right).
		\eeq 
	\item Upper bounding $ \mathbb{E} \, \| G \|_{X}$:
	
	The upper estimate of this quantity follows from Chevet's inequality \cite{Chevet78}, see also \cite[Section 43]{Tomczak1989banach}. According to that:
	\begin{align*}
\mathbb{E}  \, \| G \|_{X} &= \mathbb{E} \,  \| G \|_{\left( \s^{m, n } \otimes_{\alpha} \s^{m, n}   \right) \otimes_\ve \ell_2^{p}}
\\
&\le  
 \sup_{\varphi \in B_{\left( \s^{m, n } \otimes_{\alpha} \s^{m, n}   \right)^*} } \left( \sum_{i,j,k,l} \Big| \varphi\left( |i\xa j|\otimes |k\xa l| \right) \Big|^2  \right)^{1/2}  \,   
	\mathbb{E} \, \Big\| \, \sum_h g_h \la h | \, \Big\|_{ \ell_2^{p} }
	\\
	& + 
	\sup_{\varphi \in B_{\left( \ell_2^{p}  \right)^*} } \left( \sum_{h} \left| \varphi\left( \la h | \right) \right|^2  \right)^{1/2}  \,   
	\mathbb{E} \, \Big\| \, \sum_{i,j,k,l} g_{ijkl} |i\xa j|\otimes |k\xa l| \, \Big\|_{ \s^{m, n } \otimes_{\alpha} \s^{m, n}   }.
	\end{align*}
	Here we note the coincidence of the 2-sums above with the norm of the following identity maps  {(in both equations below, the LHS is merely the explicit expression of the norms in the RHS)}:
	\begin{align*} 
 \sup_{\varphi \in B_{\left( \s^{m, n } \otimes_{\alpha} \s^{m, n}   \right)^*} } \left( \sum_{i,j,k,l} \left| \varphi\left( |i\xa j|\otimes |k\xa l| \right) \right|^2  \right)^{1/2}
	&=
	\big\| \id: \left(  \s^{m, n} \otimes_\alpha \s^{m, n}  \right)^* \longrightarrow \ell_2^{n^2 m^2}  \big\|,
	\\
	\sup_{\varphi \in B_{\left( \ell_2^{p}  \right)^*} } \left( \sum_{h} \left| \varphi\left( \la h| \right) \right|^2  \right)^{1/2}  
	&=   \big\| \id: \ell_2^{p} \longrightarrow \ell_2^{p}  \big\| .
	\end{align*}
	Furthermore, to simplify the presentation we also introduce the notation $  \g=  \sum_{i,j,k,l} g_{ijkl} |i\xa j|\otimes |k\xa l| $. With these comments, we can write
	\begin{align*} 
	\mathbb{E}  \, \| G \|_{X} &
	\le
	  \Big\| \id: \left(  \s^{m , n} \otimes_\alpha \s^{m , n}  \right)^* \longrightarrow \ell_2^{n^2 m^2}  \Big\|  \  
	 \mathbb{E} \ \Big\| \, \sum_h g_h \la h | \, \Big\|_{ \ell_2^{p} } \
	 \\
	 & + \ 
	 \big\| \id: \ell_2^{p} \longrightarrow \ell_2^{p}  \big\|  \   
	 \mathbb{E} \ \|  \g  \|_{ \s^{m, n } \otimes_{\alpha} \s^{m, n}   }
	 \\
	 & \approx
	   \big\| \id: \left(  \s^{m, n} \otimes_\alpha \s^{m, n}  \right)^* \longrightarrow \ell_2^{n^2 m^2}  \big\|  \   
	 \sqrt{p} \
	  +   
	 \ \mathbb{E} \   \|  \g  \|_{ \s^{m, n } \otimes_{\alpha} \s^{m, n}   }.
	\end{align*}
	Now, it just left to bound $\mathbb{E} \,  \|  \g  \|_{ \s^{m, n } \otimes_{\alpha} \s^{m, n}   }$. For that, we make use of  hypothesis 1 in the statement, that is:
	\begin{align*}
	  \mathbb{E} \,  \|  \g  \|_{ \s^{m, n } \otimes_{\alpha} \s^{m, n}   }
	&\le 
	\mathbb{E} \, \left(   \|  \g  \|_{ \s^{m, n } \otimes_{\pi} \s^{m, n}   }^{1/2}      \|  \g  \|_{ \s^{m, n } \otimes_{\ve} \s^{m, n}   }^{1/2}     \right) 
	\\
	&\le 
	 \left( \mathbb{E} \,  \|  \g  \|_{ \s^{m, n } \otimes_{\pi} \s^{m, n}   } \right)^{1/2} \left(  \mathbb{E} \,     \|  \g  \|_{ \s^{m, n } \otimes_{\ve} \s^{m, n}   }    \right)^{1/2}  . 
	\end{align*}
	
	The first term  can be bounded as follows:  {we use the isometric equivalence $\s^{m,n} \simeq \ell_2^{m} \otimes_\pi \ell_2^n$ and the fact that the projective tensor norm is commutative to obtain that $\s^{m, n } \otimes_{\pi} \s^{m, n}\simeq \ell_2^n \otimes_\pi \ell_2^n \otimes_\pi \s^m$ isometrically. Furthermore, considering that $\| \id : \ell_1^{n^2}   \simeq  \ell_1^n \otimes_\pi \ell_1^n  \rightarrow \ell_2^n \otimes_\pi \ell_2^n  \| \le 1$ and $\| \id: \Ss_2^m \rightarrow \s^m \| \le \sqrt{m}$, it is also true that:
		\begin{align*}
		\mathbb{E} \,  \|  \g  \|_{  \ell_2^n \otimes_\pi \ell_2^n \otimes_\pi \s^{m}  }
		\le
		\sqrt{m} \ \mathbb{E} \,  \|  \g  \|_{ \ell_1^{n^2} ( \Ss_2^{m} ) }.
\end{align*}
Similarly, using now that $\| \id: \ell_2^{n^2} \rightarrow \ell_1^{n^2} \| \le  n$,
		\begin{align*}
	\sqrt{m} \ \mathbb{E} \,  \|  \g  \|_{ \ell_1^{n^2} ( \Ss_2^{m} ) }
	\le
	n \sqrt{m} \ \mathbb{E} \,  \|  \g  \|_{ \ell_2^{n^2} ( \Ss_2^{m} ) }
	= n \sqrt{m} \ \mathbb{E} \,  \|  \g  \|_{ \ell_2^{n^2 m^2}  }.
\end{align*}
The estimate $\mathbb{E} \,  \|  \g  \|_{ \ell_2^{n^2 m^2}  } \lesssim nm$ allows us to conclude:
	\begin{align*}
	\mathbb{E} \,  \|  \g  \|_{ \s^{m, n } \otimes_{\pi} \s^{m, n} }
	\lesssim
	n \sqrt{m} n m = n^{2} m^{3/2}.
\end{align*} }
	
	For the other term, we use again Chevet's inequality:
		\begin{align*}
	\mathbb{E} \,    \Big\| \, \sum_{i,j,k,l} g_{ijkl} |i\xa j|\otimes |k\xa l| \, \Big\|_{ \s^{m, n } \otimes_{\ve} \s^{m, n}   }  
	&\le
	2 \left\| \id:  \big( \s^{m, n} \big)^* \longrightarrow \ell_2^{n m}  \right\| \
	\mathbb{E} \, \Big\| \, \sum_{i,j} g_{ij} |i\xa j| \, \Big\|_{ \Ss_1^{m, n }  }
	\\
	& =
	2   \sqrt{n} \
	\mathbb{E} \, \Big\| \, \sum_{i,j} g_{ij} |i\xa j| \, \Big\|_{ \Ss_1^{m, n }  }
	\\
	& \lesssim
	 \sqrt{n}  \, n \sqrt{m} = n^{3/2} m^{1/2} .
	\end{align*}
	With the previous bounds, we obtain:
	\begin{equation}\label{Vr_Eq3}
	 \mathbb{E} \, \| G \|_{X} \ \lesssim \ \big\| \id:   \ell_2^{n^2 m^2}  \longrightarrow    \s^{m, n} \otimes_\alpha \s^{m, n}    \big\| \ \sqrt{p} +  n^{7/4} m.
	 \end{equation}
	\end{itemize}

	To finish, we introduce in \eqref{Vr_Eq1} the information given by \eqref{Vr_Eq2} and \eqref{Vr_Eq3}:
	\begin{align*}
	\mathrm{vr} (X^*)   
	&{\le} 
	\frac{\left\| \id: \ell_2^d \rightarrow X^* \right\|}{\sqrt{d}}   \ \mathbb{E} \, \| G \|_{X} \, 
	\\
	&\lesssim   \frac{   \min\left(\sqrt{p}, \ \frac{ n m }{ \left\| \id:  \ell_2^{n^2 m^2}  \longrightarrow \s^{m, n} \otimes_\alpha \s^{m, n} \right\| } \right)  }{n m \sqrt{p}}  
			 \ \left( \big\| \id:   \ell_2^{n^2 m^2}  \longrightarrow    \s^{m, n} \otimes_\alpha \s^{m, n}   \big\| \ \sqrt{p} \ + \   n^{7/4} m \right)
	\\
	& \le 		 \frac{ n m }{n m \sqrt{p}  \left\| \id:  \ell_2^{n^2 m^2}  \longrightarrow \s^{m, n} \otimes_\alpha \s^{m, n} \right\| }
	\ \big\| \id:   \ell_2^{n^2 m^2}  \longrightarrow    \s^{m, n} \otimes_\alpha \s^{m, n}   \big\| \ \sqrt{p}
	\\&+
	\frac{\sqrt{p}}{n m\sqrt{p}}   n^{7/4} m
 = 1 + n^{3/4},
	\end{align*}
	that is enough to conclude the proof of the statement of the theorem.

\end{proof}

 \section{Discussion}\label{Sec7}

In this work we have proposed a protocol for PV, referred as $G_{Rad}$ throughout the text, and proved lower bounds on the quantum resources necessary to break it. Our bounds, appearing in Theorem \ref{mainThm_2}, do not answer in a definite way Question \ref{question1} and, in particular, are not enough for providing security guarantees for $G_{Rad}$ in full generality. The reason is that the bounds presented in Theorem \ref{mainThm_2} depend on some additional properties of the strategy under consideration: the parameters $\sigma^i$, $\sigma^{ii}$, related with the regularity of the strategy when regarded as a vector-valued assignment on the Boolean hypercube, cf. Section \ref{Sec4}. However, our Theorem \ref{mainThm_2} is strong enough to encapsulate some previous results. As mentioned in Section \ref{Sec1}, the hypotheses of Corollary \ref{mainCor} are satisfied by the teleportation based attacks of \cite{Buhrman2011} and \cite{KonigBeigi} and also by Universal Programmable Quantum Processors, rederiving in that way some  results in \cite{Buhrman2011,KonigBeigi,Kubicki_19}.   Furthermore, we have related our Question \ref{question1}  with the type/cotype properties of specific Banach spaces and, in fact, the   obtained results led us to put forward a conjecture about these mathematical objects. The positive solution of this conjecture would imply a major progress in the understanding of $G_{Rad}$ -- See Section \ref{Sec6} for a formal statement of the conjecture and details about the connection with the security of $G_{Rad}$. In this last section we have also  provided some estimates supporting the conjecture. We have proven bounds for the volume ratio of the spaces involved there relating, as a byproduct, our conjecture, and therefore the problem about the security of $G_{Rad}$ with open problems in Banach space theory concerning the relation between cotype and volume ratio.  
 
 The future direction for this work is clear: trying to resolve the status of the security of $G_{Rad}$. Starting with the setting we introduced in Section \ref{Sec6}, the most direct approach consists in developing new techniques to estimate type/cotype constants of tensor norm spaces. This is in fact an interesting avenue also in the context of local Banach space theory and we hope that this work could serve as motivation to pursue  it. Extending the family of spaces whose type/cotype constants can be accurately estimated might shed new light on several poorly understood  questions in this context, as it is the relation between volume ratio and cotype or the prevalence of type/cotype in tensor norms.
 
 Coming back to our $\sigma$--dependent bounds, Theorem \ref{mainThm_2}, it would be also a desirable development to achieve a better understanding of the regularity parameters introduced there, $\sigma^i$ and $\sigma^{ii}$. For example, it would be very clarifying to understand how the structure of strategies is restricted under the assumption of these parameters being small (in the sense of Corollary \ref{mainCor}) or whether general strategies can be \emph{made more regular} in order to have a better behaviour in terms of these parameters. Another interesting question in this direction is understanding whether $\sigma^i$, $\sigma^{ii}$ can be related with some physical properties of the strategies involved, such as their robustness against noise or  the complexity of the operations performed.
 
 Beyond the specific setting studied here, we have introduced a whole toolbox of constructions and connections that can be of interest in other related contexts. Firstly, most of the ideas we have used to study $G_{Rad}$ can be explored in other quantum games.  More specifically, we can consider a modified version of $G_{Rad}$ in which the verifiers only have to communicate an $n^2$-dimensional quantum system to the prover -- without the assistance of any further classical communication.  { A more detailed account of this tentative line of research was given in Section \ref{Sec4_2}. Even when some of the results achieved in the present work carry over this modified setting, some new challenges appear whose exploration we leave for the future.} 
 
  Being more speculative, the recent connection between PBC and AdS/CFT \cite{May_2019,May_2020} seems to indicate that the tools we use here might have potential application to the understanding of holographic duality. Along this line, we can ask, for example, whether the notions of regularity studied here can be related with properties of the mapping between bulk and boundary theories in this context.  In \cite{May_2020} it was claimed that  properties of the AdS/CFT holographic correspondence allow to find  cheating strategies that break PBC with polynomial resources. According to that,  the exponential lower bounds in Corollary \ref{mainCor} opens the possibility to impose restrictions on the regularity of such holographic correspondence. This would be in consonance   with  a recent result of Kliesch and Koenig \cite{Kliesch20}, based on previous work of Jones \cite{Jones18}. In \cite{Kliesch20}, the continuum limit of discrete tensor-network toy models for holography  was studied finding that, generically, this limit is extremely discontinuous.

 \paragraph{Acknowledgements.}
 
 We thank Jop Briët for kindly sharing some personal notes on Pisier's method for bounding the cotype-2 constant of the projective tensor product of type-2 spaces. We also thank Elisabeth Werner, Matthias Christandl and Alex May for their kind correspondence during the preparation of this work.

 We acknowledge financial support from MICINN (grants MTM2017-88385-P and SEV-2015-0554), from Comunidad de Madrid (grant QUITEMAD-CM, ref. S2018/TCS-4342), and the European Research Council (ERC) under the European Union’s Horizon 2020 research and innovation programme (grant agreement No 648913). A.M.K. also acknowledge support from Spanish MICINN project MTM2014-57838-C2-2-P.


\printbibliography



\appendix

\section{Handier expressions for $\sigma^i$, $\sigma^{ii}$}\label{Appendix_1}

In this appendix we provide some expressions upper bounding $\sigma^{i}$ and $\sigma^{ii}$. The advantage of these expressions is that they are easier to compute and can be expressed directly in terms of the elements of a given strategy. However, we stress that in general these bounds might be inaccurate.

\begin{proposition}\label{Prop.A1}
	Given a pure strategy $\Ss^\UU = \lbrace \tilde V_\ve, \tilde W_\ve, V,W_\ve, |\varphi\ra  \rbrace_\ve \in \mathfrak{S}_{s2w}$,
	\begin{enumerate}[i.]
		\item  $$    \sigma^i 
		\lesssim_{\log}   \ \mathbb{E}_\ve  \   
		\left(  \sum_{i,j} \frac{1}{2} \left\|  \tilde  V_\ve \otimes \tilde W_\ve - \tilde V_{\overline{\ve}^{ij}} \otimes \tilde W_{\overline{\ve}^{ij}} \right\|_{  M_{ r^2, k \tilde k' }  }^2   \right)^{1/2}  + O\left(\frac{1}{n} \right); $$ 
		\item $$	\sigma^{ii}  
		\lesssim_{\log}  \ \mathbb{E}_\ve \  \left(   \sum_{i,j} \frac{1}{2}  \left \| \left( \id_{\ell_2^{k'}} \otimes (W_\ve -  W_{\overline{\ve}^{ij}})\right) |\varphi\ra \right\|_{\ell_2^{k \tilde k' }}^2 \right)^{1/2} + O\left(\frac{1}{n} \right) .$$
	\end{enumerate}
\end{proposition}

\begin{proof}
	We  provide simple, likely far from tight, bounds for the quantity
	$$
	\left(\sum_{i,j} \left\| \partial_{ij}  \Phi^{i(ii)} (\ve)  \right\|_{X^{i(ii)}}^2  \right)^{1/2}
	$$
	appearing in  \eqref{Sigma_i.} (recall that $X^{i} = M_{\tilde k'^2,k\tilde k'}$, $X^{ii} = \s^{\tilde k'\!, n} \otimes_{  (\ve, \pi)_{ \sfrac{1}{2} }  }\s^{\tilde k'\!, n}$). Recall also that $\partial_{ij}  \Phi (\ve) = \frac{\Phi (\ve_{11},\ldots,\ve_{ij},\ldots , \ve_{n n} ) -  \Phi (\ve_{11},\ldots,-\ve_{ij},\ldots , \ve_{n n} ) }{2}$. In the rest of the proof we shorten notation denoting $ (\ve_{11},\ldots,-\ve_{ij},\ldots , \ve_{n n} )$ as $ \overline \ve^{ij} $.
	
	In the case of $\Phi^{i} $,

	\begin{align*}
	&\| \partial_{ij}  \Phi^{i} (\ve)  \|_{M_{\tilde k'^2, k \tilde k'}}   
	\\
	& \qquad = \frac{1}{2}   \Bigg\| \frac{1}{n^2}  \, \sum_{p,q\ne i,j} \,   \, \varepsilon_{pq} \, \left( (\la p |\tl{V}_{\ve} \otimes \la q|\tl{ W}_{\ve}) - (\la p |\tl{V}_{\overline\ve^{ij}} \otimes \la q|\tl{ W}_{\overline\ve^{ij}}) \right)\, (V|pq\ra  \otimes \id_{\ell_2^{\tilde k'}})     
	\\
	& \qquad +     \frac{1}{n^2} \, \varepsilon_{ij} \, \left( (\la i |\tl{V}_{\ve} \otimes \la j|\tl{ W}_{\ve}) + (\la i |\tl{V}_{\overline\ve^{ij}} \otimes \la j|\tl{ W}_{\overline\ve^{ij}}) \right)
	\, (V|ij\ra  \otimes \id_{\ell_2^{\tilde k'}})    
	\Bigg\|_{M_{ r^2, k \tilde k' }}
	\\
	& \qquad \le \frac{1}{2}   \Bigg\| \frac{1}{n^2}  \, \sum_{p,q} \,   \, \varepsilon_{pq} \, \left( (\la p |\tl{V}_{\ve} \otimes \la q|\tl{ W}_{\ve}) - (\la p |\tl{V}_{\overline\ve^{ij}} \otimes \la q|\tl{ W}_{\overline\ve^{ij}}) \right)\, (V|pq\ra  \otimes \id_{\ell_2^{\tilde k'^2}})     \Bigg\|_{M_{ r^2 ,  k \tilde k' }} 
	+     \frac{2}{n^2}   
	\\
	& \qquad = \frac{1}{2}   \Bigg\| \, \la \psi_{\ve} | \, \left[   \left( (\tl{V}_{\ve} \otimes \tl{ W}_{\ve}) - (\tl{V}_{\overline\ve^{ij}} \otimes \tl{ W}_{\overline\ve^{ij}}) \right) \, (V  \otimes \id_{\ell_2^{\tilde k'}}) \otimes \id_{\HH_C} \right] |\psi\ra    \Bigg\|_{M_{ r^2, k \tilde k'}} 
	+    O\left(\frac{1}{n^2} \right)   
	\\
	& \qquad \le  \frac{1}{2}   \Bigg\| \,  (\tl{V}_{\ve} \otimes \tl{ W}_{\ve}) - (\tl{V}_{\overline\ve^{ij}} \otimes \tl{ W}_{\overline\ve^{ij}})    \Bigg\|_{M_{ r^2 , k \tilde k'}} 
	+    O\left(\frac{1}{n^2} \right)   .
	\end{align*}
	
	For $\Phi^{ii} $,
	
	\begin{align*}
	& \| \partial_{ij}  \Phi^{ii} (\ve)  \|_{{\s^{\tilde k'\!, n } \otimes_{ (\ve,\pi)_{_{  \sfrac{1}{2} } }    } \s^{\tilde k'\!, n}}}  
	\\
	& \qquad = \frac{1}{2}   \Bigg\| \frac{1}{n^2}  \, \sum_{p,q \ne i,j} \,   \, \varepsilon_{pq} \,|p\ra \otimes |q\ra \otimes \big(V|pq\ra  \otimes \id_{\ell_2^{\tilde k'^2}}\big) \big(\id_{\ell_2^{k}} \otimes( W_\ve- W_{\overline{\ve}^{ij}})\big)|\varphi \ra         
	\\
	&\qquad +    \frac{1}{n^2}   \,   \varepsilon_{ij} \,|i\ra \otimes |j\ra \otimes \big(V|ij\ra  \otimes \id_{\ell_2^{\tilde k'}}\big) \big(\id_{\ell_2^{k}} \otimes( W_\ve  {+} W_{\overline{\ve}^{ij}})\big)|\varphi \ra    
	\Bigg\|_{{\s^{\tilde k'\!, n } \otimes_{ (\ve,\pi)_{_{  \sfrac{1}{2} } }  } \s^{\tilde k'\!, n}}}
\\
	&\qquad \stackrel{(*)}{\le} \frac{1}{2}   \Bigg\| \frac{1}{n^2}  \, \sum_{p,q} \,   \, \varepsilon_{pq} \,|p\ra \otimes |q\ra \otimes \big(V|pq\ra  \otimes \id_{\ell_2^{\tilde k'^2}}\big) \big(\id_{\ell_2^{k}} \otimes( W_\ve- W_{\overline{\ve}^{ij}})\big)|\varphi \ra     \Bigg\|_{{\s^{\tilde k'\!, n } \otimes_{ (\ve,\pi)_{_{  \sfrac{1}{2} } } } \s^{\tilde k'\!, n}}} 
\\
& \qquad +   \frac{2}{n^2}   
		\end{align*}
\begin{align*}
	&\qquad \le \frac{1}{2n^2}  \, \sum_{p,q} \,   \Bigg\|     \,|p\ra \otimes |q\ra \otimes \big(V|pq\ra  \otimes \id_{\ell_2^{\tilde k'^2}}\big) \big(\id_{\ell_2^{k}} \otimes( W_\ve- W_{\overline{\ve}^{ij}})\big)|\varphi \ra     \Bigg\|_{  \s^{\tilde k'\!, n } \otimes_{ (\ve,\pi)_{_{  \sfrac{1}{2} } }  } \s^{\tilde k'\!, n} }
	\\
	&\qquad +     O\left(\frac{1}{n^2} \right)   
\\
	& \qquad\hspace{-1mm}\stackrel{(**)}{\le} \frac{1}{2n^2}  \, \sum_{p,q} \,   \Bigg\|      \big(V|pq\ra  \otimes \id_{\ell_2^{\tilde k'^2}}\big) \big(\id_{\ell_2^{k}} \otimes( W_\ve- W_{\overline{\ve}^{ij}})\big)|\varphi \ra     \Bigg\|_{  \ell_2^{ \tilde k'^2 }  }  
		+     O\left(\frac{1}{n^2} \right)  
	\\
	& \qquad \le  \frac{1}{2 n^2}  \, \sum_{p,q} \,     \Bigg\| \big(\id_{\ell_2^{k}} \otimes( W_\ve- W_{\overline{\ve}^{ij}})\big)|\varphi \ra     \Bigg\|_{\ell_2^{k \tilde k'^2}}
	+     O\left(\frac{1}{n^2} \right)   
	\\
	& \qquad \le  \frac{1}{2 }  \,    \Bigg\| \big(\id_{\ell_2^{k}} \otimes( W_\ve- W_{\overline{\ve}^{ij}})\big)|\varphi \ra     \Bigg\|_{\ell_2^{k \tilde k'^2}}
	+     O\left(\frac{1}{n^2} \right)  
   .
	\end{align*}
	
	The previous two bounds lead automatically to the claimed statement.
	
	In (*) we have applied a simple triangle inequality and used the fact that the elements in the sum are well normalized in the considered norm, recall  Remark \ref{Rmk_normalizationInt1/2}. For (**), if we denote $|\tilde \varphi_{pq} \ra :=  \big(V|pq\ra  \otimes \id_{\ell_2^{\tilde k'^2}}\big) \big(\id_{\ell_2^{k'}} \otimes( W_\ve- W_{\overline{\ve}^{ij}})\big)|\varphi \ra  $, we have to notice that, for each $p, q$, $|p\ra \otimes |q\ra \otimes |\tilde \varphi_{pq} \ra = \iota_p \otimes \iota_q (|\tilde \varphi_{pq} \ra)$, where $\iota_p$, $\iota_q$ are the  injections considered in Remark \ref{Rmk_normalizationInt1/2}. In that remark, we have proven that $\iota_p \otimes \iota_q$ is a contractive map from $ \s^{\tilde k'\!, n } \otimes_{ (\ve,\pi)_{_{  \sfrac{1}{2} } }  } \s^{\tilde k'\!, n}$ into $\ell_2^{ \tilde k'^2}$. Inequality (**) follows from this observation.
	
\end{proof}

\section{Non-pure strategies in Theorem \ref{mainThm_2}}\label{Appendix_MainThm}

We give here some further details towards the proof of Theorem \ref{mainThm_2}. We first explicit the statement we obtain in the case of \emph{pure} strategies and then, how to obtain the general statement appearing in \ref{mainThm_2}.

\begin{claim}\label{Claim_B1}
	For $\Ss^\UU \in \mathfrak{S}_{s2w;\tilde k',k}^\UU$:
		\begin{enumerate}[I.]
		\item 
		$$
		\omega(G;{\Ss^\UU})  \le  C_1 +   C_2' \ {\sigma^i} \, \log^{1/2}( k \tilde k' ) +  O \left(\frac{1}{n^{1/2}} \right)  ;
		$$
		\item   
		$$
		\omega(G;{\Ss^\UU})	  \le  \tilde C_1 + C_3' \ \tilde \sigma^{ii}  \, \log^{1/2} (n k \tilde k' )   +  O \left(\frac{1}{n^{1/2}} + \frac{\log(n)  \log^{1/2}( k \tilde k )}{n} \right)
		,
		$$
		where we have denoted  $\tilde \sigma^{ii} = n^{3/4} \log(n) \, \sigma^{ii}$.
	\end{enumerate}
	Above, $C_1,\, \tilde C_1 <1, \, C_2',\, C_3' $ are positive constants.
	
	\end{claim}
\begin{proof}
	Lemma \ref{Lemma_Main2} provides the following bounds:
		\begin{align}\label{Eq_B1}
		\omega(G_{Rad}; \Ss^{\UU}) 
		&\le  \big \|\mathbb{E}_\ve  \Phi^{i} (\ve) \big\|_{X^i} + C \  \sigma^{i} \ {\mathrm{T}^{(n^2)}_2}  \big( X^i \big),
		\\
\label{Eq_B2}
		\omega(G_{Rad}; \Ss^{\UU}) 
		&\le  \big \|\mathbb{E}_\ve  \Phi^{ii} (\ve) \big\|_{\tilde X^{ii}} + C \  \sigma^{ii} \ {\mathrm{T}^{(n^2)}_2}  \big( X^{ii} \big).
		\end{align}

Taking into account the estimates
	$$
{\mathrm{T}^{(n^2)}_2}(X^{i}) \le {\mathrm{T}_2}(X^{i})  \lesssim \log^{1/2}(k\tilde k'), \qquad {\mathrm{T}^{(n^2)}_2}(X^{ii}) \lesssim  n^{3/4} \log(n) \log^{1/2} (n \tilde k') ,
$$
and Proposition \ref{Prop_Main2}, Equations \eqref{Eq_B1}, \eqref{Eq_B2} transform in:
		\begin{equation*}
\omega(G_{Rad}; \Ss^{\UU}) 
\lesssim 
  \frac{3}{4} + O \left( \frac{1}{n^{1/2}} \right)
 +
  C \  \sigma^{i} \ \log^{1/2}(k\tilde k'),
\end{equation*}
\begin{equation*}
\omega(G_{Rad}; \Ss^{\UU}) 
\lesssim  
\frac{\sqrt{3}}{2} +  O \left( \frac{1}{\sqrt{n}} + \frac{\log (n) \log^{1/2}(k \tilde k')}{ n} \right)
 +
 C \  \sigma^{ii} \  n^{3/4} \log(n) \log^{1/2} (n \tilde k').
\end{equation*}
\end{proof}

Now, we use Lemma \ref{Lemma_Purifications} to translate the previous bound to the case of a general strategy $\Ss$, obtaining that way the statement appearing in the main text.
 
 \begin{claim}
 	The previous claim implies, for any $\Ss \in \mathfrak{S}_{s2w;\tilde k,k}$, the bounds:
 	\begin{enumerate}[I.]
 		\item 
 		$$
 		\omega(G;\Ss)  \le  C_1 +   C_2 \ {\sigma^i} \, \log^{1/2}(n k \tilde k ) +  O \left(\frac{1}{n^{1/2}} \right)  ;
 		$$
 		\item   
 		$$
 		\omega(G;\Ss)	  \le  \tilde C_1 + C_3 \ \tilde \sigma^{ii}  \, \log^{1/2} (n k \tilde k )   +  O \left(\frac{1}{n^{1/2}} + \frac{\log(n)  \log^{1/2}( n k \tilde k )}{n} \right)
 		,
 		$$
 		where we have denoted  $\tilde \sigma^{ii} = n^{3/4} \log(n) \, \sigma^{ii}$.
 	\end{enumerate}
 	Above, $C_1,\, \tilde C_1 <1, \, C_2,\, C_3 $ are positive constants.
 \end{claim}
\begin{proof}
	Lemma \ref{Lemma_Purifications} allows us to consider $\Ss$ as a pure strategy in $\mathfrak{S}_{s2w;\tilde k',k}^\UU$. The relevant estimate, also provided in that lemma, is that $\tilde k'$ can be taken to be lower or equal than $n^2 k \tilde k^4 $. I.e., $\Ss$ satisfies Claim \ref{Claim_B1} with $\tilde k' \le n^2 k \tilde k^4 $. Furthermore, we can roughly bound
	$$
	k \tilde k' \le n k \tilde k' \le (n k \tilde k')^\alpha,
	$$
	for some positive constant $\alpha$. Since those factors appear in Claim  \ref{Claim_B1} only inside a logarithm, the exponent $\alpha$ only changes the constants $C_2'$, $C_3'$ appearing there.
\end{proof}

If one wants to state Theorem \ref{mainThm_2} in terms of the raw quantum dimension $\tilde k_q := \frac{\tilde k}{\tilde k_{cl}}$, where $\tilde k_{cl}$ was the dimension of the classical messages used in the strategy, it is possible to argue similarly as above using this time Lemma \ref{Lemma_ClassCom}. The result is exactly the same, only the constants $C_2$, $C_3$ are affected.

\section{Tensor norms and enough symmetries}\label{Appendix_EnoughSym}

In this appendix we give some additional information about spaces \emph{with enough symmetries} and spaces \emph{with enough symmetries in the orthogonal group}, properties used in our Theorem \ref{Thm_vr}. Given a Banach space $X$, we refer to the group of isometries on that space as the \emph{symmetry group} of $X$.
\begin{definition}
	A Banach space $X$ has enough symmetries if the only operators on $X$ that commutes with the symmetry group of the space are $\lambda \, \id_X$ for some scalar $\lambda$.
\end{definition}

It easy to see that if $X$ has enough symmetries the same happens with $X^*$. Furthermore, it is a piece of folklore that tensor norms also respect this property. That is, for any tensor norm $\alpha$, $X \otimes_\alpha Y$ has enough symmetries when $X$ and $Y$  have enough symmetries. This fact follows from noticing that for any isometries $f$, $g$ in $X$ and $Y$, respectively, $f\otimes g$ is also an isometry in $X \otimes_\alpha Y$. This is guaranteed by the metric mapping property \eqref{MetricMapProp}. 

Finally, in \cite{Defant06} the notion of enough symmetries in the orthogonal group appears in the statement  of \cite[Lemma 5.2]{Defant06}, result used in our proof of Theorem \ref{Thm_vr}. 
\begin{definition}
	An n-dimensional Banach space $X$ has enough symmetries in the orthogonal group if the symmetry group of $X$ includes a subgroup of GL(n) verifying the property that the only operators on $X$ that commutes with that subgroup are $\lambda \, \id_X$ for some scalar $\lambda$.
\end{definition}

We finally comment that tensor norms also preserves the property of having enough symmetries in the orthogonal group. The reason is the same as in the previous case of simply having enough symmetries. Furthermore, it is obvious from the definition that $\ell_2^n$ has enough symmetries in the orthogonal group and, therefore, $\ell_2^n \otimes_\alpha \ell_2^{n'}$, $(\ell_2^n \otimes_\alpha \ell_2^{n'}) \otimes_{\alpha'} \ell_2^{n''}$, \ldots are also spaces with enough symmetries in the orthogonal group when $\alpha$, $\alpha',$ $\ldots$ are tensor norms. In particular, the spaces considered in Theorem \ref{Thm_vr} have this property.

\end{document}